%% file: GEL-arXiv-v2.tex
\documentclass[a4paper,UKenglish,cleveref, autoref, thm-restate, numberwithinsect]{lipics-v2021}

\title{On the parameterized complexity of computing good edge-labelings} 

\titlerunning{On the parameterized complexity of computing good edge-labelings} 

\author{Davi de Andrade}{Departamento de Ci\^encias da Computaç\~ao,  Universidade Federal do Cear\'a, Fortaleza, Brazil}
{daviandradeiacono@gmail.com}{https://orcid.org/0000-0002-4690-1432}{}

\author{J\'ulio Ara\'ujo}{Departamento de Matem\'atica,  Universidade Federal do Cear\'a, Fortaleza, Brazil}
{julio@mat.ufc.br}{https://orcid.org/0000-0001-7074-2753}{CNPq (Brazil) projects 313153/2021-3 and 404613/2023-3, and Inria Associated Team CANOE (Brazil-France).}

\author{Laure Morelle}{LIRMM, CNRS, Universit\'e de Montpellier, Montpellier, France}
{laure.morelle@lirmm.fr}{https://orcid.org/0009-0000-1001-1801}{}

\author{Ignasi Sau}{LIRMM, CNRS, Universit\'e de Montpellier, Montpellier, France}
{ignasi.sau@lirmm.fr}{https://orcid.org/0000-0002-8981-9287}{ANR French project ELIT (ANR-20-CE48-0008-01).}

\author{Ana Silva}{Departamento de Matem\'atica,  Universidade Federal do Cear\'a, Fortaleza, Brazil, and Università degli Studi di Firenze, Italy}
{anasilva@mat.ufc.br}{https://orcid.org/0000-0001-8917-0564}{CNPq (Brazil) 303803/2020-7 and 404479/2023-5,
COFECUB (Brazil-France) 88887.712023/2022-00,
FUNCAP (Brazil) MLC-0191-00056.01.00/22, and
MUR (Italy) PRIN 2022ME9Z78-NextGRAAL and PRIN PNRR P2022NZPJA-DLT-FRUIT.
}

\authorrunning{D. de Andrade, J. Araújo, L. Morelle, I. Sau, and A. Silva} 

\Copyright{D. de Andrade, J. Araújo, L. Morelle, I. Sau, and A. Silva} 

\ccsdesc[500]{Theory of computation~Parameterized complexity and exact algorithms}

\keywords{Good edge-labelings, parameterized complexity, structural parameters, treewidth, polynomial kernel.} 

\category{} 

\relatedversion{} 

\supplement{}

\nolinenumbers 

\hideLIPIcs  

\usepackage[shortlabels]{enumitem}
\usepackage{color,soul,dsfont,cite,xspace}
\usepackage{multirow}
\usepackage{complexity}
\usepackage[table]{xcolor} 
\usepackage{hhline} 

\usepackage[normalem]{ulem} 

\usepackage{tikz-network}

\newtheorem{question}{Question}

\newtheorem{redrule}{Rule}[section]

\usepackage[english,algoruled,longend,ruled,vlined,linesnumbered]{algorithm2e}

\newcommand{\gel}{{\sf gel}\xspace}
\newcommand{\GEL}{{\sc GEL}\xspace}
\newcommand{\gels}{{\sf gels}\xspace}
\newcommand{\nd}{{\sf nd}\xspace}
\newcommand{\vc}{{\sf vc}\xspace}
\newcommand{\sfm}{{\sf sfm}\xspace}
\newcommand{\fvs}{{\sf fvs}\xspace}
\newcommand{\tw}{{\sf tw}\xspace}
\renewcommand{\NP}{{\sf NP}\xspace}
\renewcommand{\FPT}{{\sf FPT}\xspace}

\newcommand{\Tcal}{\mathcal{T}}

\newcommand{\Ocal}{{\mathcal O}\xspace}

\newcommand{\no}{{\sf no}\xspace}
\newcommand{\sig}{{\sf sig}\xspace}

\renewcommand{\SAT}{{\sf SAT}\xspace}
\newcommand{\NAE}{{\sf NAE}\xspace}
\newcommand{\MSOL}{{\sf MSOL}\xspace}

\newcommand{\T}{`{\sf true}'\xspace}
\newcommand{\F}{`{\sf false}'\xspace}

\usepackage{subcaption}

\definecolor{medium-blue}{rgb}{0,0,0.5}
\definecolor{dark-red}{rgb}{0.7,0.15,0.15}
\definecolor{dark-blue}{rgb}{0.15,0.15,0.4}

\newcommand{\defproblem}[3]{\par
 \vspace{1mm}
\noindent
\begin{center}
\fbox{
 \begin{minipage}{0.7\textwidth}
 \begin{tabular*}{\textwidth}{@{\extracolsep{\fill}}lr} #1 &  \vspace{1mm} \\ \end{tabular*}
{\textbf{Input:}} #2
  \vspace{1mm}\\%
 {\textbf{Question:}} #3
 \end{minipage}
 }
 \end{center}
 \vspace{1mm}\par
}

\hypersetup{
    colorlinks, linkcolor={dark-blue},
    citecolor={dark-red}, urlcolor={medium-blue}
}

\begin{document}

\maketitle

\begin{abstract}
   A \emph{good edge-labeling} (\gel for short) of a graph $G$ is a function $\lambda: E(G) \to \mathbb{R}$ such that, for any ordered pair of vertices $(x, y)$ of $G$, there do not exist two distinct increasing paths from $x$ to $y$, where ``increasing'' means that the sequence of labels is non-decreasing. This notion was introduced by Bermond et al.~[Theor. Comput. Sci. 2013] motivated by practical applications arising from routing and wavelength assignment problems in optical networks. Prompted by the lack of algorithmic results about the problem of deciding whether an input graph admits a \gel, called \GEL, we initiate its study from the viewpoint of parameterized complexity. We first introduce the natural version of \GEL where one wants to use at most $c$ distinct labels, which we call $c$-GEL, and we prove that it is \NP-complete for every $c \geq 2$ on very restricted instances.

   We then provide several  positive results, starting with simple polynomial kernels for \GEL and $c$-\GEL parameterized by neighborhood diversity or vertex cover. As one of our main technical contributions, we present an \FPT algorithm for \GEL parameterized by the size of a modulator to a forest of stars, based on a novel approach via a 2-\SAT formulation which we believe to be of independent interest. We also present \FPT algorithms based on dynamic programming for $c$-\GEL parameterized by treewidth and $c$, and for \GEL parameterized by treewidth and the maximum degree. Finally, we answer positively a question of Bermond et al.~[Theor. Comput. Sci. 2013] by proving the \NP-completeness of a problem strongly related to \GEL, namely that of deciding whether an input graph admits a so-called UPP-orientation.
\end{abstract}

\newpage

\setcounter{page}{1}
\section{Introduction}
\label{sec:intro}

An \emph{edge-labeling} of a graph $G$ is a function $\lambda: E(G) \to \mathbb{R}$. Given a graph $G$ and an edge-labeling $\lambda$ of $G$, following the notation presented in~\cite{AraujoCGH12}, we say that a path in $G$ is \emph{increasing} if the sequence of edge labels in the path is non-decreasing. An edge-labeling is \emph{good} if, for any ordered pair of
vertices $(x, y)$ of $G$, there do not exist two distinct increasing paths from $x$ to $y$. In this work, we abbreviate ``good edge-labeling'' as $\gel$. If $G$ has a \gel, we say that $G$ is \emph{good}, otherwise it is \emph{bad}.

This notion has been introduced by Bermond et al.~\cite{BermondCP13}, motivated by the well-known \textsc{Routing and Wavelength Assignment} problem in optical networks, but their contributions could also be applied to other contexts such as parallel computing, as discussed in~\cite{BermondCP13}. More precisely, they were interested in the case of acyclic directed networks, thus modeled by a directed acyclic graph (DAG for short) $D$. A family of requests that require a wavelength to be sent is then represented by a family $\mathcal{P}$ of directed paths in $D$. Two paths in $\mathcal{P}$ (i.e., requests) sharing an arc must be assigned distinct colors (i.e., wavelengths). They denoted by $\omega(D,\mathcal{P})$ the minimum number of colors needed to color $\mathcal{P}$ under such constraint, and by $\pi(D,\mathcal{P})$ as the maximum \emph{load} of an arc of $D$, meaning the maximum number of paths of $\mathcal{P}$ sharing the same arc of $D$. 
A directed graph $D$ (also called digraph) satisfies the \emph{Unique Path Property} (UPP) if for any two vertices $u,v\in V(D)$ there is at most one directed $(u,v)$-path in $D$ (note that a $(v,u)$-path is also allowed); in such case, it is called a UPP-digraph or a UPP-DAG if it is additionally acyclic.
Bermond et al.~\cite{BermondCP13} used the notion of \gel and the well-known result by Erd\H{o}s about the existence of graphs with arbitrarily large girth and chromatic number to prove that
there exists a UPP-DAG $D$ and a family of dipaths $\mathcal{P}$ with load $\pi(D, \mathcal{P}) = 2$ and arbitrarily large $\omega(D, \mathcal{P})$.


It is easy to see that $K_3$ and $K_{2,3}$ are bad. Indeed, for any edge-labeling $\lambda$ of $K_3$, any edge $uv\in E(K_3)$ is an increasing $(u,v)$-path and also an increasing $(v,u)$-path; and any $(u,v)$-path of length two is either an increasing $(u,v)$-path
or an increasing $(v,u)$-path (possibly both, if the labels are equal). A similar argument applies to $K_{2,3}$, since it can be seen as three paths of length two linking two vertices. Bermond et al.~\cite{BermondCP13} asked whether any graph not containing  $K_3$ or $K_{2,3}$ as a subgraph is good. Araújo et al.~\cite{AraujoCGH12} answered this question negatively, by constructing an infinite family of incomparable bad graphs, with respect to the subgraph relation, none of them containing $K_3$ or $K_{2,3}$ as a subgraph. \autoref{fig:flower} shows an example of a bad graph not containing $K_3$ or $K_{2,3}$.

In this article we are interested in the following problem:

\defproblem{{\sc Good Edge-Labeling} ({\sc GEL} for short)}
{A graph $G$.}
{Does $G$ admit a \gel?}

\begin{figure}[t]
\center
\vspace{-.35cm}
\includegraphics[scale=0.5]{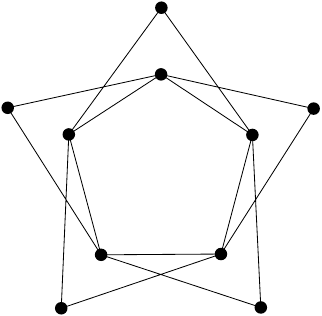}
\caption{A bad graph: for any edge-labeling, in the central 5-cycle there are three adjacent edges $uv,vw,wx$ forming an increasing path $P_1$. But then, there are two other internally-disjoint $(u,x)$-paths $P_2$ and $P_3$ of length two. Given that a 2-path is either increasing or decreasing, two of $P_1,P_2,P_3$ are either increasing or decreasing paths.}
\label{fig:flower}
\end{figure}

Araújo et al.~\cite{AraujoCGH12}  proved
that \GEL is $\NP$-complete even if the input graph is bipartite, and showed some particular classes of good and bad graphs.
It is worth mentioning that most of the classes of good graphs presented by Araújo et al.~\cite{AraujoCGH12} relied on the existence of matching cuts (see \autoref{sec:defs} for the definition) in these families, as minimally (with respect to the subgraph relation) bad graphs  cannot have such edge cuts (see \autoref{lem:matching_cut}).

Extremal combinatorial properties related to the notion of good edge-labeling have recently attracted some interest. Namely, Mehrabian~\cite{Mehrabian12} proved that a good graph $G$ on $n$ vertices such that its maximum degree is within a constant factor of its average degree has at most $n^{1+o(1)}$ edges. From this, Mehrabian deduced that there are bad graphs with arbitrarily large girth. The author also proved that for any $\Delta$, there is a
$g$ such that any graph with maximum degree at most $\Delta$ and girth at least $g$ is good. Mehrabian et al.~\cite{MehrabianMP13} proved that any good graph on
$n$ vertices has at most $n \log_2(n)/2$ edges and that this bound is tight for infinitely many values of $n$, improving previous results. There is also an unpublished work online by Bode et al.~\cite{BFT11} studying some related questions.

\subparagraph*{Our contribution and overview of techniques.} All previous articles dealing with the \GEL problem~\cite{BermondCP13,AraujoCGH12, BFT11,MehrabianMP13} just focused on deciding whether an input graph is good, regardless of the number of distinct labels used by a \gel. In fact, it was commonly assumed in previous work that any \gel is injective, that is, that there are no two edges with the same label (see \autoref{obs:injective}). In this article we introduce the natural variant of the {\sc GEL} problem in which one wants to use few distinct labels in a \gel. Formally, for a non-negative integer $c$, we say that an edge-labeling $\lambda$ of a graph $G$ is a \textit{$c$-edge-labeling} if the function $\lambda$ takes at most $c$ distinct values. A good $c$-edge-labeling is abbreviated as $c$-\gel. A graph admitting a $c$-\gel is called \textit{$c$-good}, otherwise it is called \textit{$c$-bad}.
 The corresponding decision problem is defined as follows, for a fixed non-negative integer $c$.

\defproblem{{\sc Good $c$-Edge-Labeling} ({\sc $c$-GEL} for short)}
{A graph $G$.}
{Does $G$ admit a $c$-\gel?}


Bounding the number of labels is also of interest in the case of temporal graphs, as this concept is known as the \emph{lifetime} of the corresponding temporal graph~\cite{HMNR20,KlobasMMS24}.



In this paper, we study both the \GEL and $c$-\GEL problems from the parameterized complexity point of view. While it was already known that \GEL is in \NP \cite[Theorem 5]{AraujoCGH12}, it is not clear that the same holds for $c$-\GEL for every $c \geq 2$ (note that the case $c=1$ is trivial, as it amounts to testing whether the input graph is acyclic). We prove that it is indeed the case (cf. \autoref{thm:NP}) by generalizing the proof of \cite[Theorem 5]{AraujoCGH12}. We then prove the \NP-hardness of $c$-\GEL for every $c \geq 2$ for very restricted inputs. Namely, by a reduction from \NAE 3-\SAT, we first prove the \NP-hardness of $2$-\GEL for bipartite input graphs of bounded degree that are known to admit a 3-\gel (cf. \autoref{prop:NP_2_GEL}). Then we present a reduction from $2$-\GEL to $c$-\GEL for any $c \geq 2$, proving its \NP-hardness for input graphs of bounded degree that are known to admit a $(c+1)$-\gel (cf. \autoref{thm:c-GEL-NPh}).

The above hardness results provide motivation to consider parameterizations of \GEL and $c$-\GEL in order to obtain positive results, and we do so by considering several structural parameters. We first identify two parameters that easily yield polynomial kernels by just applying simple reduction rules. Specifically, we provide a linear kernel for neighborhood diversity (cf. \autoref{lem:kernel_nd}) and a quadratic kernel for vertex cover (cf. \autoref{lem:kernel_vc}). Considering stronger parameters seems to be really challenging. For instance, we still do not know whether \GEL is \FPT parameterized by feedback vertex set or by treewidth.

To narrow this gap, we do manage to prove that \GEL is \FPT parameterized by the size of a modulator to a forest of stars (cf. \autoref{thm:kernel_sfm}).
Note that this parameter, which we call \sfm (for \textit{star-forest modulator}), is sandwiched between vertex cover and feedback vertex set (cf. \autoref{eq:parameters}) and has recently gained popularity within the parameterized complexity community~\cite{Garvardt-KomusiewiczIPEC24,DasEKMORY19,BodlaenderJK12}. While the parameter \sfm is somehow esoteric, we think that the techniques that we use to prove \autoref{thm:kernel_sfm} constitute one of our main contributions (spanning more than 12 pages), and are a proof of concept that we hope will trigger further positive results for the \GEL and $c$-\GEL problems.
In a nutshell, the algorithm starts by applying exhaustively the same simple reduction rules as for the above kernels. 
Then, all the stars in $G \setminus X$, where $X$ is the star-forest modulator of size at most $k$, are what we call \textit{well-behaved}, which allows us to classify them into three different types.
Two of these types of stars are easy to deal with, and it turns out that the third type, namely those where the center of the star and all leaves have exactly one neighbor in $X$, are much harder to handle.
To do so, we reformulate the problem in terms of \textit{labeling relations} (cf. \autoref{claim:good-labeling-relation}), in particular restricting ourselves to \textit{standard} ones (cf. \autoref{claim:always-standard}).
Intuitively, a labeling relation captures the order relation of the labels associated with every pair of edges, which allows us to guess it in time \FPT when restricted to a set of edges of size bounded by a function of $k$.
After doing so, we show that each of the resulting subproblems is equivalent to the satisfiability of an appropriately constructed 2-SAT formula, which can be decided in polynomial time \cite{Krom67}; see \autoref{alg:FPT-star-modulator}. 

We next move to treewidth, denoted by $\tw$.
Since we still do not know whether \GEL is \FPT by treewidth, we consider additional parameters.
The first natural one is $c$, the number of distinct labels.
It is easy to see that $c$-\GEL can be formulated in monadic second-order logic (\MSOL), where the size of the formula depends on $c$, hence by Courcelle's theorem~\cite{Courcelle90} it is \FPT parameterized by $\tw + c$.
In order to obtain a reasonable dependence on $\tw$ and $c$, we present an explicit dynamic programming algorithm running in time $c^{\Ocal(\tw^2)}\cdot n$ on $n$-vertex graphs (cf. \autoref{th:tw+c}).
The main idea of the algorithm is to store in the tables, for every pair of vertices in a bag (hence the term $\tw^2$ in the exponent of the running time), the existence of a few different types of paths that we prove to be enough for solving the problem.
An important ingredient of the algorithm is the definition of a \textit{partial order} on these paths, which allows us to store only the existence of paths that are {\sl minimal} with respect to that order (cf. \autoref{obs:type}).

We also consider the maximum degree $\Delta$ of the input graph as an additional parameter on top of $\tw$. Note that $\tw+c$ and $\tw + \Delta$ are a priori incomparable parameterizations. We show that \GEL can be solved in time $2^{\Ocal(\tw\Delta^2+\tw^2\log\Delta)}\cdot n$ on $n$-vertex graphs (cf. \autoref{th:tw+D}). In fact, this algorithm can find, within the same running time, either a \gel that minimizes the number of labels, or a report that the input graph is bad. The corresponding dynamic programming algorithm is similar to the one discussed above for $\tw+c$, but we need new ingredients to cope with the fact that the number of required labels may be unbounded. In particular, we use a partial orientation of the line graph of the input graph, which is reminiscent in spirit of the notion of labeling relation used in the proof of \autoref{thm:kernel_sfm}.

We would like to observe that our results constitute the first ``non-trivial'' positive algorithmic results for \GEL and $c$-\GEL, in the sense that the only positive existing so far for \GEL by Araújo et al.~\cite{AraujoCGH12} consisted in proving that, if a graph belongs to some particular graph class (such as planar graphs with high girth), then it is always good.

Finally, we answer positively a question raised by Bermond et al.~\cite{BermondCP13} by proving that deciding whether an input graph can be oriented to obtain a UPP-digraph is \NP-hard (cf. \autoref{thm:UPP-hard}), as they believed\footnote{\label{footnote1}Very recently, Dohnalová et al.~\cite{another-proof-UPP} independently proved the \NP-hardness of a related problem, namely deciding whether a graph admits a so-called \emph{KT orientation}, defined as an orientation such that there is at most one directed path (in any direction) between any pair of vertices. Note that both problems are different, as in a UPP-orientation we allow one path {\sl in each direction}. In particular, digraphs with a KT orientation are acyclic, while it is not always the case of UPP-orientations (such as the directed cycle of length three).}. The proof consists again of a reduction from \NAE 3-\SAT.

\subparagraph*{Further research.} Our work is a first systematic study of the (parameterized) complexity of the \GEL and $c$-\GEL problems, and leaves a number of interesting open questions. As mentioned above, we do not know whether \GEL is \FPT parameterized by feedback vertex set or by treewidth (when the number of labels $c$ is {\sl not} considered as a parameter). In fact, we do not even know whether they are in \XP. In view of \autoref{th:tw+c}, a positive answer to the following question would yield an \FPT algorithm for treewidth.

\begin{question}\label{question:bound-labels-tw}
Does there exist a function $f: \mathbb{N} \to \mathbb{N}$ such that, for any graph $G$, if $G$ is good then $G$ is $f(\tw)$-good, where \tw is the treewidth of $G$?
\end{question}

A question related to \autoref{question:bound-labels-tw} is the following one, that  if answered positively, would again yield an \FPT algorithm for treewidth by \autoref{th:tw+D}.

\begin{question}\label{question:bound-labels-degree}
Does there exist a function $f: \mathbb{N} \to \mathbb{N}$ such that, for any graph $G$, if $G$ is good then $G$ is $f(\Delta)$-good, where $\Delta$ is the maximum degree of $G$?
\end{question}
We think that the answer to \autoref{question:bound-labels-degree} is positive even for $f$ being the identity function.

In view of our \FPT algorithm parameterized by \sfm (\autoref{thm:kernel_sfm}), it seems natural to consider the size of a modulator $X$ to graphs other than stars. If the components of $G \setminus X$ are paths with at most three vertices, it is easy to see that the reduction rules presented in this paper are enough to provide a polynomial kernel (in particular, the problem is \FPT). But if we increase the size of the paths to four, the problem seems to become much more complicated. Thus, a first natural concrete problem to play with is the following.

\begin{question}\label{question:modulator-to-P4}
Are  the \GEL or $c$-\GEL problems \FPT parameterized by the size of a modulator to paths with at most four vertices?
\end{question}
Note that a positive answer to \autoref{question:bound-labels-tw} (resp. \autoref{question:bound-labels-degree}) would imply a positive answer to \autoref{question:modulator-to-P4} by \autoref{th:tw+c} (resp. \autoref{th:tw+D}). If a positive answer to \autoref{question:modulator-to-P4} is found, it would make sense to consider as the parameter the \textit{vertex integrity} of the input graph~\cite{abs-2402-09971,LampisM21}. It is worth mentioning that the \FPT algorithm of \autoref{thm:kernel_sfm} is for the \GEL problem, but we do not know whether it can be generalized to $c$-\GEL. We do not know either whether \GEL admits a polynomial kernel parameterized by $\sfm$. It is easy to see that a trivial AND-composition (see~\cite{CyganFKLMPPS15}) shows that \GEL is unlikely to admit polynomial kernels parameterized by $\tw +c$ or by $\tw + \Delta$.

\subparagraph*{Relation to temporal graphs.} The notion of increasing paths has also been of great interest to the community studying temporal graphs (see e.g.~\cite{KKK.00} and the survey~\cite{Michail2015}).
For this type of graphs, an edge-labeling of the edges is given, representing their availability in time. It thus makes sense to consider only non-decreasing paths, also called \emph{temporal paths}.
There is an active recent line of research addressing classical graph theory problems in the context of temporal graphs~\cite{CasteigtsCS24,CostaLMS24, IS24,MS23,HMNR20}. Another approach that has attracted interest in the past years is that of temporal graph \emph{design}, where one wants to produce a temporal graph satisfying certain constraints, usually related to temporal connectivity~\cite{MMS19,KlobasMMS24,CSS.24} or to the existence of a so-called temporal cycle~\cite{abs-2503-02694}. In this sense, the {\sc GEL} problem can also be seen as temporal graph design problem, where one wants to obtain a \emph{temporization} of at most one time per edge such that each pair of vertices is linked by at most one temporal path. Hence, our results may find applications in the setting of temporal graphs.


In the temporal framework, usually {\sl strict} journeys are used, that is, it is required that the time strictly increases along a temporal path~\cite{MMS19,KlobasMMS24,CSS.24}.
Note that, in the context of \GEL and $c$-\GEL, both problems become trivial as all graphs would be $1$-good: just assign the same label to all edges, and then clearly there is at most one strictly increasing path between any pair of vertices (namely, there is one path if and only if they are adjacent).


Another typical constraint in temporal graphs is to require that the time slots define a {\sl proper} edge coloring of the underlying graph~\cite{KKK.00,Michail2015}. In case we would define a variation of the \GEL problem in which the edges incident to each vertex must have distinct labels,  the positive and negative results for the \GEL problem directly transfer to this variant, since we can assume that all labels are distinct (cf. \autoref{obs:injective}). However, in the case of $c$-\GEL, the situation changes. Our hardness results  for $c$-\GEL may be adapted to this variant, possibly by introducing an extra constant number of labels to satisfy the new constraint.

Also, making a bridge with temporal connectivity, it would make sense to ask that there exists {\sl exactly one} non-decreasing path between any pair of vertices. Note that this would result in a different problem (since in the \GEL problem we ask that {\sl at most one} such a path exists), which may be worth studying.


\subparagraph*{Organization.} In \autoref{sec:defs} we  present basic definitions and preliminary results that we use throughout this work. In
\autoref{sec:NP-complete} we prove that the $c$-\GEL problem is \NP-complete for every $c \geq 2$. In \autoref{sec:kernels} we present simple reduction rules and kernels for \GEL and $c$-\GEL, and in \autoref{sec:FPT-stars} we present the \FPT algorithm for \GEL parameterized by \sfm. Our dynamic programming algorithms for $\tw+c$ and $\tw + \Delta$ are described in \autoref{sec:FPT}. Finally, the \NP-completeness of finding a UPP-orientation can be found in  \autoref{sec:UPP-NPh}.

\section{Definitions and preliminary results}
\label{sec:defs}


We start with some basic definitions that will be used throughout the paper.

\subparagraph*{Graphs.} We use standard graph-theoretic notation, and we refer the reader to~\cite{Diestel16} for any undefined terms.
All graphs we consider are finite and undirected,  except in \autoref{sec:UPP-NPh} where we consider digraphs.
A graph~$G$ has vertex set~$V(G)$ and edge set~$E(G)$. An edge between two vertices $u,v$ is denoted by $uv$. For a graph $G$ and a vertex set $S \subseteq V(G)$, the graph $G[S]$ has vertex set $S$ and edge set $\{uv \mid u,v \in S \text{ and }uv \in E(G)\}$.
 We use the shorthand $G \setminus S$ to denote $G[V(G) \setminus S]$. For a single vertex~$v \in V(G)$, we use~$G \setminus v$ as a shorthand for~$G \setminus \{v\}$. Similarly, for a set of edges~$F \subseteq E(G)$ we denote by~$G \setminus F$ the graph on vertex set~$V(G)$ with edge set~$E(G) \setminus F$. For two sets of vertices $S_1,S_2\subseteq V(G)$, we denote by $E(S_1,S_2)$ the subset of $E(G)$ containing all edges with one endpoint in $S_1$ and the other one in $S_2$. A cycle on three vertices is called a \emph{triangle}. For two positive integers $i,j$ with $i \leq j$, we denote by $[i,j]$ the set of all integers $\ell$ such that $i \leq \ell \leq j$, and by $[i]$ the set $[1,i]$. Given $v \in V(G)$, we denote $N_G(v)=\{u \mid uv \in E(G)\}$, $d_G(v)=|N_G(v)|$
and, given $X \subseteq V(G)$, we denote $N_G(X)=\bigcup_{v \in X}N_G(v) \setminus X$.
Given $X,Y \subseteq V(G)$, we denote by $N_G^Y(X)=N_G(X) \cap Y$. We may omit the subscript~$G$ when it is clear from the context. A path from a vertex $u$ to a vertex $v$ is called a \textit{$(u, v)$-path}. Note that, when dealing with edge-labelings, a pair of consecutive labels in a path is increasing, decreasing, or equal depending on the direction in which the path is traversed.
A subset $S\subseteq V(G)$ is a \textit{clique} (resp. \textit{independent set}) if all its vertices are pairwise adjacent (resp. pairwise non-adjacent) in $G$. A graph $G$ is \emph{complete} if $V(G)$ is itself a clique. The complete graph (resp. cycle) on $p$ vertices is denoted by $K_p$ (resp. $C_p$). A bipartite graph on parts $A,B$ is said to be \emph{complete bipartite} if its edge set is equal to all edges between $A$ and $B$. The complete bipartite graph with parts of sizes $p$ and $q$ is denoted by $K_{p,q}$.

A \textit{cut-vertex} in a connected graph $G$ is a vertex $v \in V(G)$ such that $G \setminus v$ is disconnected. A \textit{separation} of a graph $G$ is a pair $(A,B)$ such that $A,B \subseteq V(G)$, $A \cup B = V(G)$, and there are no edges in $G$ between the sets $A \setminus B$ and $B \setminus A$. The \textit{order} of a separation $(A,B)$ is defined as $|A \cap B|$. A \textit{matching} in a graph $G$ is a set of pairwise disjoint edges. An \textit{edge cut} in a graph $G$ is the set of edges between a set $S \subseteq V(G)$ and its complement $\bar{S}$, assuming that there is at least one such edge, and it is denoted $[S, \bar{S}]$. An edge set $F \subseteq E(G)$ is a \textit{matching cut} if it is both a matching and an edge cut.

\subparagraph*{Parameterized complexity.}
A \emph{parameterized problem} is a language $L \subseteq \Sigma^* \times \mathbb{N}$, for some finite alphabet $\Sigma$.  For an instance $(x,k) \in \Sigma^* \times \mathbb{N}$, the value~$k$ is called the \emph{parameter}. For a computable function~$g \colon \mathbb{N} \to \mathbb{N}$, a \emph{kernelization algorithm} (or simply a \emph{kernel}) for a parameterized problem $L$ of \emph{size} $g$ is an algorithm $A$ that given any instance $(x,k)$ of $L$, runs in polynomial time and returns an instance $(x',k')$ such that $(x,k) \in L \Leftrightarrow (x', k') \in L$ with $|x'|$, $k' \le g(k)$. The function $g(k)$ is called the \textit{size} of the kernel, and a kernel is \textit{polynomial} (resp. \textit{linear, quadratic}) if $g(k)$ is a polynomial (resp. linear, quadratic) function. Consult~\cite{CyganFKLMPPS15,Niedermeier06,FlumG06,DoFe13,FominLSZ19} for background on parameterized complexity.

\subparagraph*{Graph parameters.}
We proceed to define the graph parameters that will be considered in the polynomial kernels presented in \autoref{sec:kernels} and in the \FPT algorithms presented in \autoref{sec:FPT-stars} and \autoref{sec:FPT}. A vertex set $S$ of a graph $G$ is a \textit{vertex cover} (resp. \textit{feedback vertex set}, \textit{star-forest modulator}) if $G \setminus S$ is an edgeless graph  
(resp. forest, forest of stars). For a graph $G$, we denote by $\vc(G)$ (resp. $\fvs(G)$, $\sfm(G)$) the minimum size of a  vertex cover (resp. feedback vertex set, star-forest modulator) of $G$. Clearly, for any graph $G$ it holds that $\fvs(G) \leq \sfm(G) \leq \vc(G)$.

Two vertices $u,v$ of a graph $G$ have the same \textit{type} if $N(u) \setminus \{v\} = N(v) \setminus \{u\}$. Two vertices with the same type are \textit{true twins} (resp. \textit{false twins}) if they are adjacent (resp. non-adjacent).
The \textit{neighborhood diversity} of a graph $G$, denoted by $\nd(G)$, as defined by Lampis~\cite{Lampis12}, is the minimum integer~$w$ such that $V(G)$ can be partitioned into $w$ sets such that all the vertices in each set have the same type.
Note that the property of having the same type is an equivalence relation, and that all the vertices in a given type are either true or false twins, hence defining either a clique or an independent set.
It is worth mentioning that, as proved by Lampis~\cite{Lampis12}, the parameter $\nd(G)$ can be computed in polynomial time. If a graph $G$ has a vertex cover $S$ of size $k$, $V(G)$ can be easily partitioned into at most $2^k + k$   equivalence classes of types, thus implying that, for any graph $G$,
\begin{equation}\label{eq:nd-vc}
\nd(G) \leq 2^{\vc(G)} + \vc(G).
\end{equation}

A \textit{tree decomposition} of a graph $G$ is a pair $(T, \{B_t \mid t \in V(T)\})$, where $T$ is a tree and each set $B_t$, called a \textit{bag}, is a subset of $V(G)$, satisfying the following properties:
\begin{enumerate}
    \item $\bigcup_{t \in V(T)}B_t = V(G)$,
    \item for every edge $uv \in E(G)$, there exists a bag $B_t$ with $u,v \in B_t$, and
    \item for every vertex $v \in V(G)$, the set $\{t \in V(T) \mid v \in B_t\}$ induces a connected subgraph of~$T$.
\end{enumerate}
The \textit{width} of a tree decomposition $(T, \{B_t \mid t \in V(T)\})$ is defined as $\max_{t \in V(T)}|B_t|-1$, and the \textit{treewidth} of a graph $G$, denoted by $\tw(G)$, is the minimum width of a tree decomposition of $G$. Since any forest has treewidth at most one, for any graph $G$ it holds that $\tw(G) \leq \fvs(G) + 1$.

A \textit{nice tree decomposition} of a graph $G$ is a tree decomposition of $G$ with one special bag $B_r$ called the \textit{root} that defines an ancestor/descendant relation in $T$, and in which each other bag is of one of the following types: 
\begin{itemize}
    \item \textit{Leaf bag}: a leaf $x$ of $T$ with $B_x = \emptyset$.
    \item \textit{Introduce bag}: an internal vertex $x$ of $T$ with exactly one child $y$ such that $B_x = B_y \cup \{v\}$ for some $v \notin B_y$.
    \item \textit{Forget bag}: an internal vertex $x$ of $T$ with exactly one child $y$ such that $B_x = B_y \setminus \{v\}$ for some $v \in B_y$.
    \item \textit{Join bag}: an internal vertex $x$ of $T$ with exactly two children $y_1$ and $y_2$ such that $B_x = B_{y_1} \cup B_{y_2}$.
\end{itemize}
Additionally, we can, and will, assume that $B_r=\emptyset$.
We can easily do so by forgetting one by one each vertex of the root.

Given $x\in V(T)$, we denote by $T_x$ the maximal subtree of $T$ rooted at $x$, and we set $V_x:=\{v\in V(G)\mid \exists y\in V(T_x), v\in B_y\setminus B_x\}$ and $G_x:=G[V_x]$.
Let us stress that $G_x$ does not contain the vertices of $B_x$ here, contrary to commonly used notations.

As discussed in~\cite{Kloks94}, any given tree decomposition of a graph $G$ can be transformed in polynomial time into a {\sl nice} tree decomposition of $G$ with the same width. Hence, we will assume in \autoref{sec:FPT} that we are given a nice tree decomposition of the input graph $G$. We refer the reader to the recent results of Korhonen~\cite{Korhonen21} for approximating optimal tree decompositions in \FPT time.

Summarizing the above discussion, for any graph $G$ it holds that
\begin{equation}\label{eq:parameters}
    \tw(G) -1 \leq \fvs(G) \leq \sfm(G) \leq \vc(G),
\end{equation}
while $\nd$, which satisfies $\nd(G) \leq 2^{\vc(G)} + \vc(G)$ (cf. \autoref{eq:nd-vc}), is easily seen to be incomparable to any of $\tw, \fvs$, or $\sfm$ (cf. for instance~\cite{Ganian12}).


\medskip
We now state several observations and preliminary results that will be used in the next sections. The following observations follow easily from the definitions of  \gel and $c$-\gel.

\begin{observation}[Araújo et al.~\cite{AraujoCGH12}]\label{obs:injective}
A graph $G$ admits a \gel if and only if it admits an injective \gel, that is, a good edge-labeling $\lambda: E(G) \to \mathds{R}$ such that for any two distinct edges $e,f \in E(G)$, $\lambda(e) \neq \lambda(f)$.
\end{observation}

\begin{observation}\label{obs:restrict_label_gel}
If a graph $G$ admits a \gel using at most $c$ distinct labels, then it admits a \gel $\lambda$ such that $\lambda :E(G)\to[c]$.
\end{observation}

\begin{observation}\label{obs:connectivity}
A graph $G$ admits a \gel (resp. $c$-\gel) if and only if every connected component of $G$ admits a \gel (resp. $c$-\gel).
\end{observation}

\begin{observation}\label{obs:good-subgraph}
If a graph $G$ admits a \gel (resp. $c$-\gel), then any subgraph of $G$ also admits a \gel (resp. $c$-\gel).
\end{observation}

\begin{observation}
\label{obs:C4}
An edge-labeling of $C_4$ with values in $\{1,2\}$ is good if and only if its edges take alternatively values one and two in a cyclic order.
\end{observation}

The next two results come from~\cite[Lemma 7]{AraujoCGH12} and~\cite[Lemma 10]{AraujoCGH12}, respectively. They were written for \gels instead of $c$-\gels, but the statements below follow from the same proofs.


\begin{lemma}[Araújo et al.~\cite{AraujoCGH12}]
\label{lem:cut_vertex}
Let $G$ be a graph,  $v$ be a cut-vertex in $G$, $C_1,\ldots,C_p$ the vertex sets of the connected components of $G \setminus v$, and $G_i = G[C_i \cup \{v\}]$ for $i \in [p]$. Then, for any non-negative integer $c$, $G$ admits a $c$-\gel if and only if every $G_i$ admits a $c$-\gel for $i \in [p]$.
\end{lemma}

\begin{lemma}[Araújo et al.~\cite{AraujoCGH12}]
\label{lem:2_separation}
Let $c\in\mathbb{N}$, $G$ be a graph, $\lambda$ be a $c$-edge-labeling of $G$, and $(A,B)$ be a separation of $G$ of order two such that $G[A\cap B]$ is an edge.
If $\lambda$ restricted to $G[A]$ and $\lambda$ restricted to $G[B]$ are $c$-\gels, then $\lambda$ is a $c$-\gel of $G$. Moreover, if both $G[A]$ and $G[B]$ are good, then $G$ is also good.
\end{lemma}

The following result comes from~\cite[Lemma 11]{AraujoCGH12}. Again, it was proved for \gels, but taking into account \autoref{obs:restrict_label_gel} and the fact that, in the proof of~\cite[Lemma 11]{AraujoCGH12}, a new large label is given to the edges of the matching cut while preserving the labels of both sides of the cut, the following holds.

\begin{lemma}[Araújo et al.~\cite{AraujoCGH12}]
\label{lem:matching_cut}
Let $c\in\mathbb{N}$.
Let $G$ be a graph and $[S,\bar{S}]$ be a matching cut in $G$.
If both $G[S]$ and $G[\bar{S}]$ admit a $c$-\gel, for some non-negative integer $c$, then $G$ admits a $(c+1)$-\gel.
\end{lemma}


\autoref{lem:matching_cut} motivates the following reduction rule.

\begin{redrule}\label{rule:matching-cut}
If a graph $G$ contains a matching cut $[S,\bar{S}]$, delete all edges between $S$ and  $\bar{S}$.
\end{redrule}


The safeness of \autoref{rule:matching-cut} for the sake of admitting a \gel is justified by \autoref{lem:matching_cut}: $G$ admits a \gel if and only if both $G[S]$ and $G[\bar{S}]$ admit a \gel.

\smallskip
We now state a useful equivalence for verifying that an edge-labeling is good. Let $G$ be a graph and $\lambda$ be an edge-labeling of $G$.
Let $C$ be an elementary cycle of $G$. A \textit{local minimum} (resp. \textit{local maximum}) of $C$, with respect to $\lambda$, is a subpath $P$ of $C$ consisting of edges with the same label,
such that the edges in $E(C) \setminus E(P)$ incident to the endpoints of $P$ have labels strictly larger (resp. smaller) than that of~$P$.


\begin{observation}[Bode et al.~\cite{BFT11}]\label{obs:local_minima}
 An edge-labeling of a graph $G$ is good if and only if every elementary cycle of $G$ admits at least two local minima, or at least two local maxima.
\end{observation}

\section{\NP-completeness of $c$-\GEL for every $c \geq 2$}
\label{sec:NP-complete}


In this section we prove that the  $c$-\GEL problem is \NP-complete for every $c \geq 2$. To do so, we first prove in \autoref{sec:NP} that $c$-\GEL is in \NP for every $c \geq 2$. We then prove in \autoref{sec:NPhard} the \NP-hardness of  {\sc 2-GEL}, and we use it in \autoref{sec:NP-hard-c-gel} to prove the \NP-hardness of  {\sc $c$-GEL} for every $c \geq 2$.

\subsection{The {\sc $c$-GEL} problem is in \NP for every $c \geq 2$}
\label{sec:NP}

In \cite[Theorem 5]{AraujoCGH12}, the authors prove that deciding whether an \emph{injective} edge-labeling is good is polynomial-time solvable. (It is worth mentioning that this result can also be easily deduced from much earlier work in the context of temporal graphs~\cite{B96}.)
By \autoref{obs:injective}, this is enough to prove that {\sc GEL} is in \NP.
However, to prove that {\sc $c$-GEL} is in \NP for some fixed $c$, we need to prove that deciding whether \emph{any} given edge-labeling with at most $c$ distinct values is good is polynomial-time solvable.
This is what we prove in the following result by generalizing the proof of \cite[Theorem 5]{AraujoCGH12} to the not-necessarily injective case.

\begin{theorem}\label{thm:NP}
    Given a graph $G$ and an edge-labeling $\lambda$ of $G$, there is an algorithm deciding whether $\lambda$ is good in polynomial time.
\end{theorem}

\begin{proof}
For each vertex $v\in V(G)$ we will check whether there are two increasing paths in $G$ beginning in $v$ and ending in the same vertex $u$.
If we find such a pair of increasing paths, we conclude that $\lambda$ is a bad edge-labeling.
Otherwise, if we find no such a pair for each $v\in V(G)$, we conclude that $\lambda$ is a good edge-labeling.
Let $c\leq |E(m)|$ be the number of labels of $\lambda$.
By \autoref{obs:restrict_label_gel}, we may assume that $\lambda:E(G)\to[c]$.

For each $i\in[c]$, let $G_{i,1},\dots,G_{i,{r_i}}$ be the maximal connected subgraphs of $G$ whose edges are all labeled $i$.
Note that if $G_{i,j}$ contains a cycle, then we can immediately conclude that $\lambda$ is a bad edge-labeling.
Hence, we can assume that $G_{i,j}$ is a tree for $j\in[r_i]$.

Let $v\in V(G).$ Let $V'=\{v\}$ and $E'=\emptyset$.
Let $G'=(V',E')$. Note that $G'$ is trivially a tree.
For $i\in[c]$ in increasing order, we do the following, while maintaining the property that $G'$ is a tree.

For $j\in[r_i]$, if $G_{i,j}$ contains exactly one vertex of $V'$, then we add $V(G_{i,j})$ to $V'$ and $E(G_{i,j})$ to $E'$.
Note that, given that $G_{i,j}$ is a tree, $G'$ remains a tree.

If $G_{i,j}$ contains at least two distinct vertices $x$ and $y$ of $V'$, then we claim that $\lambda$ is a bad edge-labeling.
Indeed, let $u$ be the least common ancestor of $x$ and $y$ in the tree $G'$ rooted at $v$.
Without loss of generality, we may assume that $u\ne x$.
Let $P$ be the path from $u$ to $x$ in $G'$, and let $P'$ be the path going from $u$ to $y$ in $G'$ and from $y$ to $x$ in $G_{i,j}$.
$P$ and $P'$ are two disjoint increasing paths from $u$ to $x$, hence proving the claim.

Note that after step $i$, all increasing paths of $G$ beginning in $v$ and finishing with an edge of label at most $i$ are present in $G'$.
Hence, after step $c$, if $G'$ is a tree, then there are no two increasing paths beginning in $v$ and finishing at the same vertex $u$.

Finally, it can be easily verified that the presented algorithm runs in time $\Ocal (c \cdot |E(G)| \cdot |V(G)|)$, where $c$ is the number of distinct labels used by $\lambda$.
\end{proof}

\subsection{\NP-hardness of {\sc 2-GEL}}
\label{sec:NPhard}

In this subsection we prove that {\sc 2-GEL} is \NP-hard.
We do so by reducing from the {\NAE 3-\SAT} problem, that is, the variation of 3-\SAT where each clause needs to be satisfied, but the assignments where all literals in a clause are all true or all false are forbidden.
The {\NAE 3-\SAT} problem is known to be \NP-complete even if each clause contains exactly three literals and each variable occurs exactly four times~\cite{DarmannD20}.



We hence need to find gadgets that will represent clauses and variables. We first define a gadget that will link variable gadgets to clause gadgets and propagate a label.

\subparagraph{Propagation gadget.}
The \emph{propagation gadget}, denoted by $P$, is the graph pictured in \autoref{fig:propagation_gadget}.
We call the \emph{bones} of the propagation gadget the edges $u_1u_2$ and $v_1v_2$.

\begin{figure}[h]
\centering
\vspace{-.25cm}
\includegraphics[scale=1]{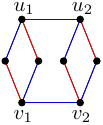}
\caption{The propagation gadget $P$ along with a 2-\gel: blue edges have label one and red edges have label two.}
\label{fig:propagation_gadget}
\end{figure}

\begin{lemma}\label{prop:propagation_gadget}
For any 2-\gel $\lambda$ of the propagation gadget $P$, the bones of $P$ have the same label.
\end{lemma}

\begin{proof}
Let $u_1u_2$ and $v_1v_2$ be the two bones of $P$ such that there are two 2-paths from $u_i$ to $v_i$ for $i\in[2]$.
By symmetry of the labels in a 2-\gel, we can assume that $\lambda(u_1u_2)=1$.
For $i\in[2]$, given that there are two 2-paths from $u_i$ to $v_i$, one of these two is strictly decreasing, and the other is strictly increasing (see \autoref{obs:C4}).
Hence, there is a 2-path $u_iw_iv_i$ such that $\lambda(u_iw_i)=1$ and $\lambda(w_iv_i)=2$.
If $\lambda(v_1v_2)=2$, then we would have two increasing path from $u_1$ to $v_1$: the path $u_1w_1v_1$ and the path $u_1u_2w_2v_2v_1$. This contradicts the fact that $\lambda$ is a 2-\gel.
Therefore, $\lambda(v_1v_2)=1$.
\end{proof}

We can now prove that {\sc 2-GEL} is \NP-hard.

\begin{theorem}\label{prop:NP_2_GEL}
The {\sc 2-GEL} problem is \NP-hard even on bipartite instances that admit a 3-\gel and with maximum degree at most ten.
\end{theorem}


\begin{proof}
We present a reduction from the restriction of the \NAE 3-\SAT problem where each clause contains exactly three literals and each variable occurs exactly four times, which is known to be \NP-complete~\cite{DarmannD20}. 
Let $\varphi$ be a formula on variables $\{x_1,\ldots,x_n\}$ and clauses $\{c_1,\ldots,c_m\}$.
We construct a graph $G$ as follows:
for each variable $x_i$, we create a 4-cycle $X_i$ and distinguish two consecutive edges of $X_i$ that we call $e_i$ and $\bar{e}_i$ (first picture of \autoref{fig:gadgets_2_gel}).
For each clause $c_j$, we create a 5-cycle $C_j$ and distinguish three consecutive edges of $C_j$ that we call $f_{j,1},f_{j,2},$ and $f_{j,3}$ (second picture of \autoref{fig:gadgets_2_gel}).
For each $j\in[m]$, let $l_{j,1},l_{j,2}$, and $l_{j,3}$ be the three literals in $c_j$.
For $a\in[3]$, if $l_{j,a}=x_i$ (resp. $l_{j,a}=\overline{{x}_i}$),  we add a propagation gadget $P_{i,j,a}$ with one bone identified with $l_i=e_i$ (resp. $l_i=\bar{e}_i$), and the other with $f_{j,a}$ (third picture of \autoref{fig:gadgets_2_gel}).
This completes the construction of $G$.

\begin{figure}[h]
\centering
\vspace{-.15cm}
\includegraphics{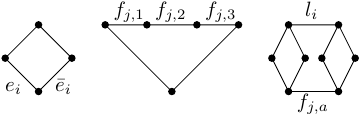}
\caption{
Gadgets for the reduction to {\sc 2-GEL}: from left to right,  $X_i$, $C_j$, and $P_{i,j,a}$.}
\label{fig:gadgets_2_gel}
\end{figure}

To get some intuition, note that the constructed graph $G$ closely follows the ``structure'' of the incidence graph of the formula $\varphi$ , with the vertices corresponding to variables and clauses being replaced by the gadgets $X_i$ and $C_j$, respectively, and the edges connecting variables and clauses being replaced by propagation gadgets (hence, propagating the same label between the corresponding bones). Note also that, since each variable occurs exactly four times in $\varphi$, $G$ has maximum degree ten, achieved at the vertices in the gadgets $X_i$ incident with $e_i$ and $\bar{e}_i$.

\begin{claim}
If $G$ admits a 2-\gel, then $\varphi$ has a \NAE satisfying assignment.
\end{claim}
\begin{proof}
Let $\lambda:E(G)\to[2]$ be a 2-\gel of $G$.
Then for each $i\in[n]$, if $e_i$ has label one, then we assign \T to $x_i$, and if $e_i$ has label two, then we assign \F to $x_i$.
We claim that this corresponds to a \NAE satisfying assignment of $\varphi$.
By \autoref{obs:C4}, if $e_i$ has label $b$ for $b\in[2]$, then $\bar{e_i}$ has label $3-b$ since $e_i$ and $\bar{e}_i$ are consecutive edges of a $C_4$ (the gadget variable).
Additionally, for $j\in[m]$ and $a\in[3]$,
if $l_{j,a}=x_i$ (resp. $l_{j,a}=\overline{{x}_i}$), then there is a propagation gadget $P_{i,a,j}$ with one bone identified with $e_i$ (resp. $\bar{e}_i$), and the other with $f_{j,a}$.
Therefore, by \autoref{prop:propagation_gadget}, $f_{j,a}$ and $e_i$ (resp. $\bar{e}_i$) have the same label.

Assume for a contradiction that this is not a \NAE satisfying assignment of $\varphi$.
Then there is a clause $c_j$ whose literals are all assigned to either \T or \F.
Hence, $f_{j,1}$, $f_{j,2},$ and $f_{j,3}$ have the same label.
But then there are two increasing paths in the $C_5$ gadget corresponding to $c_j$: one using edges $f_{j,1}$, $f_{j,2}$, and $f_{j,3}$, and the other using the two other edges.
This contradicts the fact that $\lambda$ is a \gel.
\end{proof}

\begin{claim}\label{claim:NAE_2_gel}
If $\varphi$ has a \NAE satisfying assignment, then $G$ admits a 2-\gel.
\end{claim}
\begin{proof}
For each $i\in[n]$, if $x_i$ is assigned \T (resp. \F), then we give label one (resp. two) to $e_i$.
Then, by \autoref{obs:C4}, there is only one way to label the $C_4$ containing $e_i$ in a 2-\gel.
Moreover, by \autoref{prop:propagation_gadget}, both bones of a propagation gadget must have the same label, and the rest of the propagation gadget are two $C_4$ that are labeled by alternating one and two in a cyclic ordering of the edges.
Given that the literals of a clause $c_j$ are not assigned all \T or all \F, we deduce that for every $j\in[m]$, $f_{j,1}$, $f_{j,2},$ and $f_{j,3}$ do not have the same label.
Hence, we can always assign to the two unlabeled edges neighboring $f_{j,1}$ and $f_{j,3}$ label one for one of them and label two for the other such that the labeling is a 2-\gel of the clause gadget.

Let us show that this 2-labeling is a \gel.
Let $A$ be the set of vertices of $G$ that are part of a variable gadget, $B$ be the set of vertices that are part of a clause gadget, and $S=V(G)\setminus (A\cup B)$. Hence, $S$ is the set of vertices of propagation gadgets that are not endpoints of bones.
By construction, $A\cap B=\emptyset$, $N_G(A)=N_G(B)=S$, and $S$ is an independent set.
Let $C$ be a cycle of $G$.
Observe that $C$ intersects $S$ an even number of times.
If $V(C)\cap S=\emptyset$, then $C$ is contained in either a variable or a clause gadget and thus has two local minima.
If $|V(C)\cap S|\ge4$, then given that every vertex in $S$ is incident to two edges, one labeled one and the other labeled two, $C$ immediately has two local minima.
Finally, if $|V(C)\cap S|=2$, then $C$ intersects exactly one propagation gadget $P_{i,j,a}$ and thus the associated variable gadget $X_i$ and clause gadget $C_{j,a}$.
By construction of the labeling, both bones $l_i$ and $f_{j,a}$ of the propagation gadget have same label $b$, and this label is also present in $X_i-l_i$ and in $C_{j,a}-f_{j,a}$.
Additionally, there is an edge of $C$ adjacent to $s$, and an edge adjacent to $s'$, both with label $3-b$, where $V(C)\cap S=\{s,s'\}$.
Therefore, no matter what is the path taken by $C$ in the variable and in the clause gadgets, $C$ has two local minima.
Hence, the labeling is indeed good.
\end{proof}

\begin{claim}
$G$ admits a 3-\gel.
\end{claim}
\begin{proof}
Let us define a first labeling with labels in $\{1,3\}$, that we will then modify to make it a 3-\gel.
We assign label one to each $e_i$, and complete $X_i$ and each $P_{i,j,a}$ with labels one and three so that the labeling is good on each $X_i$ and $P_{i,j,a}$.
If there is a $C_j$ such that $f_{j,1}$, $f_{j,2},$ and $f_{j,3}$ have the same label, then we replace the label of $f_{j,2}$ with label two.
Finally, we complete the labeling of each $C_j$ by labeling one unlabeled edge with label one and the other with label three such that the labeling is good on $C_j$ (such a labeling is always possible given that the $f_{j,a}$'s do not have the same label for $a\in[3]$).

Let us prove that this labeling is good.
We define $A,B,$ and $S$ as in \autoref{claim:NAE_2_gel}.
If $V(C)\cap S=\emptyset$, then $C$ is contained in either a variable or a clause gadget and thus have two local minima.
If $|V(C)\cap S|\ge4$, then given that every vertex in $S$ is incident to two edges, one labeled one and the other labeled three, $C$ immediately has two local minima.
Finally, if $|V(C)\cap S|=2$, then $C$ intersects exactly one propagation gadget $P_{i,j,a}$ and thus the associated variable gadget $X_i$ and clause gadget $C_{j,a}$.
Let $b\in\{1,3\}$ be the label of $l_i$.
By construction, the label of $f_{j,a}$ is either two or $b$, and the label $b$ is also present in $X_i-l_i$ and in $C_{j,a}-f_{j,a}$.
Additionally, there is an edge of $C$ adjacent to $s$, and an edge adjacent to $s'$, with label $3-b$, where $V(C)\cap S=\{s,s'\}$.
Therefore, no matter what is the path taken by $C$ in the variable and in the clause gadgets, $C$ has two local minima (one of them may be with label two).
Hence, the labeling is indeed good by \autoref{obs:local_minima}.
\end{proof}

Note that $G$ can be made bipartite by replacing each clause gadget with the one depicted in \autoref{fig:clause_gadget}. The proof works essentially the same way, and note that the maximum vertex degree in the constructed graph is still ten.
\begin{figure}[h]
\centering
\includegraphics[scale=1]{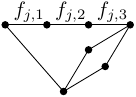}
\caption{Bipartite variation of the clause gadget.}
\label{fig:clause_gadget}
\end{figure}
\end{proof}

\subsection{\NP-hardness of {\sc $c$-GEL}}
\label{sec:NP-hard-c-gel}

We now prove that {\sc $c$-GEL} is \NP-hard for every $c \geq 3$ by reducing from {\sc $2$-GEL}.
To do so, we define two gadgets: the first one allows us to restrict the possible labels for an edge $e$, while the second one is a graph that admits a $c$-\gel but no $(c-1)$-\gel.

\subparagraph{Extremal gadget.}
We define the first gadget, called \emph{extremal gadget} and denoted by $X$, to be the graph represented in \autoref{fig:extremal_gadget}; the edge labeled $e$ is called the \emph{bone} of the extremal gadget.

\begin{figure}[h]
\centering
\includegraphics[scale=1]{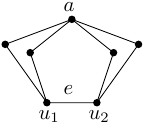}
\caption{The extremal gadget $X$.}
\label{fig:extremal_gadget}
\end{figure}

\begin{lemma}\label{prop:extremal_gadget}
The extremal gadget $X$ admits a 3-\gel.
Additionally, for any $c$-\gel of $X$, for some $c\in\mathbb{N}$, $c \geq 3$, the label of the bone $e$ of $X$ is in $[2,c-1]$.
\end{lemma}

\begin{proof}
Let $a$ be the vertex of $X$ of degree 4 and let $u_1$ and $u_2$ be the endpoints of $e$.

We first construct a 3-\gel $\lambda$ of $X$ as follows:
let $\lambda(e)=2$ and, for each $i\in[2]$, give labels $1,3$ to a path from $a$ to $u_i$ and labels $3,1$ to the other, in this order. Note that $\lambda$ is in fact a 3-\gel of $X$, since any increasing 2-path leaving from $a$ ends with label $3$ and any decreasing 2-path leaving from $a$ ends with label $1$, therefore they cannot be extended using $e$.

Now, we prove that the label of $e$ is in $[2, c-1]$ for every $c$-\gel of $X$ with $c \geq 3$.
Let $\lambda$ be a $c$-\gel of $X$.
Note that, for $i\in[2]$, there are two 2-paths from $a$ to $u_i$, and by \autoref{obs:C4} one of the 2-paths from $a$ to $u_i$ is increasing and the other is decreasing. Hence, there is one increasing path from $a$ to $u_1$ and one increasing path from $a$ to $u_2$.

We now show that we must have $\lambda(e)<c$: indeed, if $\lambda(e)=c$, the increasing 2-path from $a$ to $u_2$ can be extended to an increasing 3-path from $a$ to $u_1$ using $e$,  creating a second increasing path from $a$ to $u_1$ and contradicting the fact that $\lambda$ is a $c$-\gel.
Similarly, since there is one decreasing path from $a$ to $u_1$ and one decreasing path from $a$ to $u_2$, we conclude that $\lambda(e)>1$.

\end{proof}

\subparagraph{Color gadget.}
We now define the second gadget. Given $c\in\mathbb{N}$, the \emph{$c$-color gadget}, denoted by $D_c$, is defined as follows (see \autoref{fig:color_gadget} for an illustration); its vertex set is  $\{v,v_1,\dots,v_c\}\cup\{v_{i,j}\mid i,j\in[1,c], i< j\}$ and its edge set is $\{vv_i\mid i\in[1,c]\}\cup\{v_iv_{i,j}\mid i,j\in[1,c],i< j\}\cup\{v_jv_{i,j}\mid i,j\in[1,c],i< j\}$.
Note that $D_c$ has $c(c+1)/2+1$ vertices and has maximum degree $c$. 

\begin{figure}[h]
\centering
\includegraphics{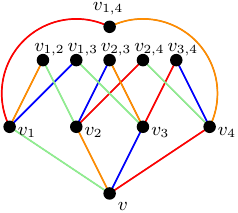}
\caption{The $c$-color gadget $D_c$ for $c=4$ with a $c$-\gel where each color represents a label.
}\label{fig:color_gadget}
\end{figure}

\begin{lemma}\label{prop:color_gadget_gel}
    For every $c\in\mathbb{N}_{\ge 1}$, $D_c$ admits a $c$-\gel.
\end{lemma}

\begin{proof}
The \emph{hypercube} of dimension $c$, denoted by $H_c$, is the graph with vertex set $[0,1]^c$ and such that two vertices are adjacent if and only if they differ on exactly one coordinate.
Remark also that $D_c$ is an induced subgraph of $H_c$ where $v=[0]^c$ and $v_i$ (resp. $v_{i,j}$) is the vertex of $H_c$ with exactly one (resp. two) one(s) at coordinate $i$ (resp. $i$ and $j$).

Recall that if $H_c$ admits a $c$-\gel, then so does $D_c$ by \autoref{obs:good-subgraph}. Thus, it suffices to prove that $H_c$ admits a $c$-gel, which we do by induction.
$H_1$ is an edge, which trivially admits a 1-\gel.
Suppose that $H_{c-1}$ admits a $(c-1)$-\gel.
Observe that the set $E_c$ of edges of $H_c$ whose endpoints differ only by their last coordinate (or any other set of all edges in the same direction) form a matching cut of $H_c$ such that $H_c-E_c$ is the disjoint union of two copies of $H_{c-1}$.
Therefore, by induction, \autoref{lem:matching_cut} and \autoref{obs:connectivity}, $H_c$ admits a $c$-\gel.
%
\end{proof}

We now prove that $D_c$ admits no $(c-1)$-\gel, given any $c$.

\begin{lemma}\label{prop:color_gadget_no_GEL}
    For every $c\in\mathbb{N}_{\ge 2}$, $D_c$ does not admit a $(c-1)$-\gel.
    More generally, there are no $i\ne j\in[1,c]$ such that $\lambda(vv_i)=\lambda(vv_j)$ for any \gel $\lambda$ of $D_c.$
\end{lemma}

\begin{proof}
Let $\lambda$ be an edge-labeling of $D_c$ on $c-1$ colors.
Given that $v$ has degree $c$, there are $i\ne j\in[1,c]$ such that $\lambda(vv_i)=\lambda(vv_j)$.
By symmetry, we can assume that the 2-path $v_iv_{i,j}v_j$ is an increasing path.
But then there two increasing paths from $v_i$ to $v_j$.
So $\lambda$ is not a \gel.
\end{proof}

By combining the two previous results,
we create a new gadget that forces an edge to have the greatest label in the labeling.

\subparagraph{Forced gadget.}
Let the \emph{$c$-forced gadget}, denoted by $F_c$ be the graph constructed as follows: take a $(c-1)$-color gadget $D_{c-1}$ and, for each $i\in[2,c-1]$, identify the edge $vv_i$ with the bone of an extremal gadget $X_i$.
We call the edge $vv_1$ the \emph{bone} of $F_c$.
Note that $F_c$ has maximum degree $3c-5$. Hence, from \autoref{prop:extremal_gadget} and \autoref{prop:color_gadget_gel}, we deduce the following result.

\begin{lemma}\label{prop:forced_gadget}
    For any $c\ge 3$, $F_c$ admits a $c$-\gel, and for any $c$-\gel of $F_c$, the bone of $F_c$ has label in $\{1,c\}$.
\end{lemma}

\begin{proof}
Let us first construct a $c$-\gel of $F_c$.
By \autoref{prop:color_gadget_gel}, $D_{c-1}$ has a $(c-1)$-\gel $\lambda$.
By \autoref{prop:color_gadget_no_GEL} and by symmetry, for each $i\in[2,c-1]$, we may assume that $\lambda(vv_i)=i$.
Additionally, for each extremal gadget $X_i$, define the same labeling as the one defined in the proof of \autoref{prop:extremal_gadget}, but with colors shifted from $[1,3]$ to $[i-1,i+1]$, so that the bone of $X_i$ (which is identified with $vv_i$) has label $i$.
Then, by \autoref{lem:2_separation}, the defined edge-labeling of $F_c$ is a \gel, and since it uses colors in $[c]$, it is a $c$-\gel.
Note that in this $c$-\gel, $vv_1$ has label $1$, but by symmetry, there is a $c$-\gel such that $vv_1$ has label $c$.

Let $\lambda$ be a $c$-\gel of $F_c$.
Let us prove that $vv_1$ has either label $1$ or label $c$.
Given the extremal gadget on each $vv_i$ for $i\in[2,c-1]$, by \autoref{prop:extremal_gadget}, we conclude that $\lambda(vv_i)\in[2,c-1]$.
Moreover, by \autoref{prop:color_gadget_no_GEL}, we know that each $vv_i$ has a different label.
Hence, all the labels in $[2,c-1]$ are taken by the edges $vv_i$ for $i\in[2,c-1]$, and $vv_1$ must have one of the leftover labels, hence in $\{1,c\}$.
\end{proof}

Finally, we make a reduction from {\sc 2-GEL} to {\sc $c$-GEL} using the $c$-forced gadgets. Note that, in contrast to \autoref{prop:NP_2_GEL}, we do not assume the input graph to be bipartite anymore, since the $c$-forced gadgets are not bipartite.



\begin{theorem}\label{thm:c-GEL-NPh}
    For any $c\ge 2$, the {\sc $c$-GEL} problem is \NP-hard even on instances that admit a $(c+1)$-\gel and with maximum degree at most $10(3c-5)$.
\end{theorem}
\begin{proof}
We reduce from {\sc 2-GEL} as follows.
Let $G$ be an instance of {\sc 2-GEL}.
By \autoref{prop:NP_2_GEL}, $G$ can be assumed to admit a $3$-\gel and to have maximum degree at most ten.
Let $G'$ be the graph obtained from $G$ by identifying every edge $e$ of $G$ with the bone of a $c$-forced gadget $F_c^e$. We now split the proof into three claims.

\begin{claim}
    If $G$ admits a 2-\gel, then $G'$ admits a $c$-\gel.
\end{claim}
\begin{proof}
Let $\lambda$ be a 2-\gel of $G$, which we can clearly assume to take values in $\{1,c\}$.
By \autoref{prop:forced_gadget}, each $F_c^e$ admits a $c$-\gel, and, by symmetry, we can choose the label of its bone to be any $\lambda(e)\in\{1,c\}$.
Finally, by \autoref{lem:2_separation}, this is a $c$-\gel of $G'$.
\end{proof}

\begin{claim}
    If $G'$ admits a $c$-\gel, then $G$ admits a $2$-\gel.
\end{claim}
\begin{proof}
Let $\lambda$ be a $c$-\gel of $G'$.
By \autoref{prop:forced_gadget}, every edge of $G$, that is a bone of a $c$-forced gadget, has a label in $\{1,c\}$.
Then the restriction of $\lambda$ to $G$ immediately gives a $2$-\gel.
\end{proof}

\begin{claim}
    $G'$ admits a $(c+1)$-\gel.
\end{claim}
\begin{proof}
Let $\lambda$ be a 3-\gel of $G$ with labels in 
$\{1,c,c+1\}$.
By \autoref{prop:forced_gadget}, each $F_c^e$ admits a $c$-\gel, and, by symmetry and translation, we can choose the label of its bone to be any $\lambda(e)\in\{1,c,c+1\}$.
Finally, by \autoref{lem:2_separation}, the combination of these labelings gives a $(c+1)$-\gel of $G'$.
\end{proof}

In the first two claims, we prove that $G$ admits a 2-\gel if and only if $G'$ admits a $c$-\gel, while in the third claim we prove that $G'$ also admits a $(c+1)$-gel.
Finally, note that, since $G$ has maximum degree ten and each edge of $G$ has been identified with the bone of a $c$-forced gadget, the maximum degree of $G'$ is indeed $10(3c-5)$.
\end{proof}

\section{Simple reduction rules and polynomial kernels for \GEL and $c$-\GEL}
\label{sec:kernels}

In this section we provide simple reduction rules and polynomial kernels for {\sc GEL} and {\sc $c$-GEL} parameterized by neighborhood diversity and vertex cover. These reduction rules will also be used by the \FPT algorithm presented in \autoref{sec:FPT-stars}. Our kernels are based on exhaustive application of the following three simple reduction rules to the input graph $G$, which can be applied in polynomial time. We would like to stress that these three reduction rules are safe for both {\sc GEL} and {\sc $c$-GEL} for any $c \geq 2$.

\begin{redrule}\label{rule:1}
    If $G$ contains $K_3$ or $K_{2,3}$ as a subgraph, report a \no-instance.
\end{redrule}


The safeness of \autoref{rule:1} is justified by \autoref{obs:good-subgraph} and the fact that $K_3$ or $K_{2,3}$ are bad graphs~\cite{BermondCP13}.

\begin{redrule}\label{rule:2}
   If $G$ is disconnected, consider each connected component separately.
\end{redrule}

The safeness of \autoref{rule:2} is justified by \autoref{obs:connectivity}. Note that, when applying this rule for obtaining a kernel, we should be careful about how the parameter is split among the different connected components. Indeed, if we just kept the same parameter for each component and kernelize each of them separately, the total size of the kernelized graph may depend on the number of components. Fortunately, the two parameters considered in this section, namely $\nd$ and $\vc$, satisfy that if a graph $G$ has connected components $C_1,\ldots,C_x$, then ${\sf p}(G)= \sum_{i \in x}{\sf p}(G_i)$, for ${\sf p}$ being $\nd$ or $\vc$. The fact that we can split the parameter appropriately among the connected components in polynomial time follows 
from the fact that the neighborhood diversity can be computed in polynomial time~\cite{Lampis12}
and the hypothesis that a vertex cover is given in \autoref{lem:kernel_vc}.

\begin{redrule}\label{rule:3}
  Let $v$ be a cut-vertex in $G$ and let $C$ be the vertex set of a connected component of $G \setminus v$. If $G[C \cup \{v\}]$ is good (or $c$-good if we deal with $c$-\GEL), delete $C$ from $G$.
\end{redrule}

Observe that \autoref{rule:3} allows, in particular, to remove vertices of degree one. The safeness of this rule is justified by \autoref{lem:cut_vertex}. Note that the application of \autoref{rule:3} can be done in polynomial time as far as we can decide whether $G[C \cup \{v\}]$ is good in polynomial time. In all the applications of \autoref{rule:3} discussed below this will indeed be the case (usually, because $C$ is of bounded size), and we will only need to apply \autoref{rule:3} to the considered configurations to obtain the claimed kernels.

We are now ready to present our polynomial kernels.

\begin{lemma}\label{lem:kernel_nd}
The {\sc GEL} and $c$-\GEL problems parameterized by the neighborhood diversity of the input graph admit a kernel of size at most $2k$.
\end{lemma}
\begin{proof}
We apply exhaustively \autoref{rule:1}, \autoref{rule:2}, and \autoref{rule:3} to the input graph $G$. We then compute in polynomial time, using the algorithm of Lampis~\cite{Lampis12}, an optimal partition of $V(G)$ into equivalence classes of types $V_1,\ldots,V_k$. We want to prove that each $V_i$ has size at most two. Note that, by \autoref{rule:1}, we may assume that all the $V_i$'s that induce cliques have size at most two. Suppose for contradiction that there exists a set $V_i$ inducing an independent set with $|V_i| \geq 3$. By \autoref{rule:2}, necessarily $V_i$ is adjacent to another set $V_j$. If $V_i$ is adjacent to another set $V_{\ell}$ with $j \neq \ell$, then $G[V_i \cup V_j \cup V_{\ell}]$ contains a $K_{2,3}$, which is impossible by \autoref{rule:1}. Thus, $V_i$ is only adjacent to $V_j$. If $|V_j| \geq 2$, then $G[V_i \cup V_j]$ contains a $K_{2,3}$, which is again impossible by \autoref{rule:1}. So necessarily $|V_j|=1$. But then $G[V_i \cup V_j]$ is a star, which is a good graph, and therefore $V_i$ should have been deleted by \autoref{rule:3}. Therefore, for any set $V_i$, it holds that $|V_i| \leq 2$, hence $|V(G)|\leq 2k$.
\end{proof}

Let us see that the analysis of the kernel size in \autoref{lem:kernel_nd} is asymptotically tight assuming that only Rules~\ref{rule:1},~\ref{rule:2}, and~\ref{rule:3} are applied.
For this, let $K_n^+$ be the graph obtained from $K_n$, the complete graph on $n$ vertices, by replacing every edge with a $C_4$. That is, for every edge $uv \in E(K_n)$, we delete it and we add two new vertices $u',v'$ and the edges $uu',u'v,vv',v'u$. It can be  verified that $|V(K_n^+)|=n^2$, that $\nd(K_n^+)=n + {n \choose 2}=: k$, and that none of the three rules can be applied to it. Simple calculations show that $|V(K_n^+)|= k\left(2 - \frac{1}{\Omega(\sqrt{k})}\right)$, which tends to $2k$ as $k$ grows. Hence, to improve the kernel size in \autoref{lem:kernel_nd}, new reduction rules would be needed.

\begin{lemma}\label{lem:kernel_vc}
The \GEL and $c$-\GEL problems parameterized by the size of a given vertex cover of the input graph admit a kernel of size at most $k^2$.
\end{lemma}
\begin{proof}
Let $X$ be a vertex cover of the input graph $G$ of size at most $k$, and let $v$ be a vertex in $V(G) \setminus X$. We apply exhaustively \autoref{rule:1}, \autoref{rule:2}, and \autoref{rule:3}. By \autoref{rule:2}, $v$ has at least one neighbor, and all its neighbors are in $X$. If $v$ has exactly one neighbor, then $v$ should have been deleted by \autoref{rule:3}. Thus, each vertex not in $X$ has at least two neighbors in $X$. This implies that $|V(G) \setminus X|$ is upper-bounded by the sum, over all pairs of vertices $u_1,u_2 \in X$, of the number of common neighbors of $u_1$ and $u_2$ in $V(G) \setminus X$. Now consider one pair of  vertices $u_1,u_2 \in X$. If $u_1$ and $u_2$ had at least three common neighbors in $V(G) \setminus X$, $u_1$ and $u_2$ together with their common neighborhood in $V(G) \setminus X$ would contain a $K_{2,3}$, which is impossible by \autoref{rule:1}. Therefore, any two vertices $u_1,u_2 \in X$ have at most two common neighbors in $V(G) \setminus X$, which implies that $|V(G)| \leq k + 2{k \choose 2} = k^2$.
\end{proof}

Again, it is easy to see that analysis of the kernel size of \autoref{lem:kernel_vc} is tight assuming that only Rules~\ref{rule:1},~\ref{rule:2}, and~\ref{rule:3} are applied. Indeed, consider the graph $K_k^+$ defined above with a given (minimum) vertex cover consisting of the $k$ original vertices of the clique. Then none of the reduction rules can be applied to $K_k^+$, which has $k^2$ vertices.


\section{\FPT algorithm for \GEL by the size of a star-forest modulator}
\label{sec:FPT-stars}

In this section we prove the following theorem.


\begin{theorem}\label{thm:kernel_sfm}
The {\sc GEL} problem parameterized by the size $k$ of a given star-forest modulator of the $n$-vertex input graph can be solved in time $2^{\Ocal(k^4 \log k)} \cdot n^{\Ocal(1)}$.
\end{theorem}

We stress that the above \FPT algorithm works for the \GEL problem, but not necessarily for $c$-\GEL for a fixed $c > 0$, since we make no attempt to minimize the number of labels.  The remainder of this section is devoted to prove \autoref{thm:kernel_sfm}, for which we need a number of definitions and intermediate results. For the sake of readability, we structure the proof of the theorem into several parts.

It is worth mentioning here that, at the end of this section, we explain how the presented algorithm can be simplified. We decided to present the more involved formulation because we would like our algorithm to be considered as a proof of concept of our technique, which in our opinion is more relevant than the particular result of \autoref{thm:kernel_sfm} itself.

\medskip

Let $X$ be a given star-forest modulator of the input graph $G$ of size $k$. For the sake of flexibility of our setting, we assume henceforth that the parameter $k$ is an upper bound on the size of the given modulator, i.e., that  $|X|\leq k$. Our goal is to decide whether $G$ admits a \gel, without aiming at optimizing the number of distinct labels. Note that every connected component $S$ of $G \setminus X$, seen as a subgraph, is a star composed of its center and its leaves.

\subparagraph*{Taming the stars.} We first show that we may assume that every star in $G \setminus X$ satisfies some simple properties, formalized as follows. We say that a  connected component $S$ of $G \setminus X$ is \textit{well-behaved} if the following hold:
    \begin{enumerate}
        \item Every leaf of $S$ has at least one neighbor in $X$.
        \item No three leaves of $S$ share a neighbor in $X$.
        \item $S$ contains at most $2k$ leaves, that is, $|S| \leq 2k + 1$.
    \end{enumerate}
We note that the third item above is in fact a consequence of the first two (cf. the proof of \autoref{claim:star-assumptions} below), but we prefer to keep it in the definition because it will be important for the algorithm.
\begin{claim}\label{claim:star-assumptions}
    We can assume that every connected component $S$ of $G \setminus X$ is well-behaved.
\end{claim}
\begin{proof}
    Let $s$ be the center of $S$.

    For the first item of the definition of well-behaved,  suppose that $v$ is a leaf of $S$ with no neighbor in $X$. Then $v$ constitutes a connected component of $G \setminus \{s\}$, and thus it can be safely deleted by \autoref{rule:3}.

    For the second item, if three leaves $v_1,v_2,v_3$ share a neighbor $u$ in $X$, then $G[\{s,v_1,v_2,v_3,u\}]$ contains a $K_{2,3}$, hence we can safely report a \no-instance by \autoref{rule:1}.

    For the third item, if $S$ contains at least $2k+1$ leaves, since by the first item we may assume that every leaf of $S$ has at least one neighbor in $X$, the fact that $|X| \leq k$ and the pigeonhole principle imply that $S$ contains three leaves  sharing a neighbor in $X$, contradicting the second item.
\end{proof}

By \autoref{claim:star-assumptions}, we assume henceforth that all the stars in $G \setminus X$ are well-behaved. We now partition the stars in $G \setminus X$ into three types, and analyze each of them separately. A connected component $S$ of $G \setminus X$ is
\begin{itemize}
    \item \textit{boring} if it contains a vertex with at least two neighbors in $X$;
    \item \textit{$0$-interesting} if it is not boring and the center of $S$ has no neighbor in $X$;
    \item \textit{$1$-interesting} if it is not boring and the center of $S$ has exactly one neighbor in $X$.
\end{itemize}

Let $B \subseteq V(G)$ be the set of vertices belonging to a boring star of $G \setminus X$. Note that, by definition, every leaf in a $0$-interesting or $1$-interesting star has exactly one neighbor in $X$. Each type of star as defined above behaves quite differently. Indeed, it is easy to show (cf. \autoref{claim:few-stars}) that the number of boring stars is bounded by a function of $k$, hence we will be able to use a brute-force approach on them. On the other hand, it is safe to just delete all $0$-interesting stars (cf. \autoref{claim:0interesting-stars}). Hence, it will only remain to deal with $1$-interesting stars, which turn out to be much more complicated, and for which we will create \FPT-many (that is, a function of $k$) appropriate 2-\SAT formulas whose satisfiability can be checked in polynomial time.

\subparagraph*{Boring stars are few, and $0$-interesting stars are easy.} We start by proving that there are not many boring stars.

\begin{claim}\label{claim:few-stars}
    We may assume that the number of boring stars in $G \setminus X$ is at most $k^2$.
\end{claim}
\begin{proof}
    The proof of this claim follows closely that of \autoref{lem:kernel_vc}. Let $x$ be the number of (well-behaved) boring stars in $G \setminus X$. Then, $x$ is upper-bounded by the sum, over all pairs of vertices $u_1,u_2 \in X$, of the number of common neighbors of $u_1$ and $u_2$ in $V(G) \setminus X$. Now consider one pair of  vertices $u_1,u_2 \in X$. If $u_1$ and $u_2$ had at least three common neighbors in $V(G) \setminus X$, $u_1$ and $u_2$ together with their common neighborhood in $V(G) \setminus X$ would contain a $K_{2,3}$, which is impossible by \autoref{rule:1}. Therefore, any two vertices $u_1,u_2 \in X$ have at most two common neighbors in $V(G) \setminus X$, which implies that $x \leq   2{k \choose 2} < k^2$.
\end{proof}


We now show that we can get rid of $0$-interesting stars. Before that, we need one more definition that will also be useful to deal with $1$-interesting stars. Let $S$ be a $0$-interesting or $1$-interesting star in $G \setminus X$, let $v$ be a leaf of $S$, and let $u$ be its unique neighbor in $X$. We say that $v$ is a \textit{type-2} leaf if there exists another leaf $v'$ of $S$ adjacent to $u$. Otherwise, we say that $v$ is a \textit{type-1} leaf. Note that, since we assume that all stars are well-behaved, there are no three leaves of the same star sharing a neighbor in $X$, although a vertex of $X$ can be adjacent to centers or leaves of arbitrarily many stars of $G \setminus X$.

\begin{claim}\label{claim:0interesting-stars}
Let $S$ be a $0$-interesting star in $G \setminus X$. Then $G$ is good  if and only if $G \setminus V(S)$ is good.
\end{claim}
\begin{proof}
If $G$ is good, then $G \setminus X$ is also good by \autoref{obs:good-subgraph}. Assume now that $H := G \setminus V(S)$ is good and let $\lambda: E(H) \to \mathds{R}$ be a \gel of $H$. By \autoref{obs:restrict_label_gel} we can assume that $\lambda: E(H) \to [c]$ for some positive integer $c$. Let $s$ be the center of $S$. We extend $\lambda$ into an edge-labeling $\lambda'$ of $G$, by using two new labels $0$ and $c+1$ for the edges in $E(G) \setminus E(H)$, as follows; see \autoref{fig:01interesting-stars}(a) for an illustration. Namely, if $v$ is a type-1 leaf of $S$ and $u$ is its neighbor in $X$, we set $\lambda'(sv)=c+1$ and $\lambda'(vu)=0$. On the other hand, if $v_1,v_2$ are two type-2 leaves of $S$ with common neighbor $w \in X$, we set $\lambda'(sv_1)=\lambda'(v_2w)=c+1$ and $\lambda'(sv_2)=\lambda'(v_1w)=0$. We claim that $\lambda'$ is a \gel of $G$.

By \autoref{obs:local_minima}, this is equivalent to verifying that every cycle admits at least two local minima. Since $\lambda'(e)=\lambda(e)$ for every edge $e \in E(H)$, it is enough to consider a cycle $C$ intersecting $S$, hence containing its center $s$. Then, by the topology of $S$ and the choice of $\lambda'$, necessarily $C$ contains two edges with label $0$ whose all incident edges in $C$ have a strictly greater label (in the cycle $C$ depicted in \autoref{fig:01interesting-stars}(a) with thick red edges, these two edges with label 0 are $vu$ and $sv_2$). Thus, $C$ admits two local minima and the claim follows.
\end{proof}

\begin{figure}[htb]
\begin{center}
\vspace{-.25cm}
\includegraphics[scale=1.00]{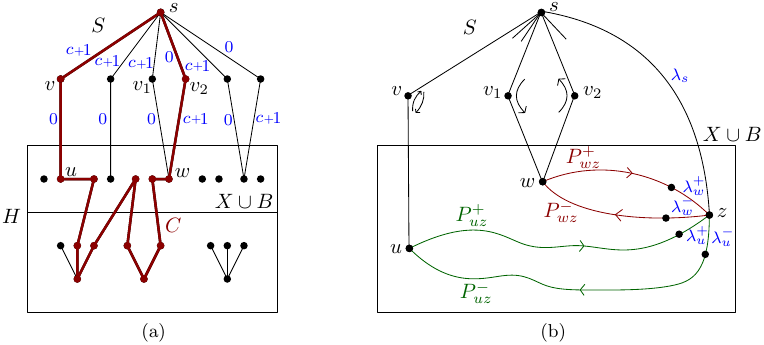}
\end{center}
\caption{(a) Illustration of the proof of \autoref{claim:0interesting-stars}, where a $0$-interesting star $S$ in $G \setminus X$ is depicted. The labels of the edges in $E(G) \setminus E(H)$ given by $\lambda'$ are depicted in blue. A cycle $C$ of $G$ containing the center $s$ of $S$ is depicted with thick red edges. (b) Interaction of a $1$-interesting star $S$ with the set $X \cup B$; other stars in $G \setminus X$ are not shown. The arrows in the paths indicate the direction in which the labels increase. Some labels are depicted in blue.} \label{fig:01interesting-stars}
\end{figure}

\autoref{claim:0interesting-stars} justifies the safeness of the following reduction rule, which can clearly be applied in polynomial time.


\begin{redrule}\label{rule:4}
If $G \setminus X$ contains a $0$-interesting star $S$, delete all the vertices of $V(S)$ from $G$.
\end{redrule}

After applying \autoref{rule:4} exhaustively, we can assume henceforth that all the stars in $G \setminus X$ are well-behaved and either boring or $1$-interesting.

\subparagraph*{$1$-interesting stars are hard.} Before proceeding with the algorithm, let us first give some intuition about why $1$-interesting stars are inherently more complicated than $0$-interesting stars, which will allow us to convey the underlying idea of our approach.
Consider the example in \autoref{fig:01interesting-stars}(b), where a $1$-interesting star $S$ with center $s$ is depicted.

For the sake of simplicity, consider only the interaction of $S$ with the modulator $X$ and the set of vertices $B$ belonging to boring stars; in fact, we will see later that keeping track of this interaction is enough, in the sense that we can easily get rid of ``problematic cycles'' (that is, those having less than two local minima) intersecting more than one $1$-interesting star by using an appropriate type of ``standard''  labelings. This type of labelings will also guarantee that all ``problematic cycles'' intersecting a $1$-interesting star contain the edge between its center and $X$ (cf. \autoref{claim:always-standard} for the details).

Suppose, similarly to the proof of \autoref{claim:0interesting-stars}, that a \gel $\lambda$ of $G[X \cup B]$ has been already found. Let $z$ be the neighbor of $s$ in $X$, and let $\lambda_s$ be the label that we need to choose for the edge $sz$.

Consider first a pair of type-2 leaves $v_1,v_2$ of $S$ with common neighbor $w \in X$. Regardless of which labels we choose for the 4-cycle induced by $\{s,v_1,v_2,w\}$, necessarily one of the paths from $s$ to $w$ will get increasing labels, and the other one will get decreasing labels; in \autoref{fig:01interesting-stars}(b), the arrows in the paths indicate the direction along which the labels increase. Since $\lambda$ is a \gel of $G[X \cup B]$, there is at most one increasing path $P^+_{wz}$ in $G[X \cup B]$ from $w$ to $z$, and at most one decreasing path $P^-_{wz}$ (which may intersect); see the red paths in \autoref{fig:01interesting-stars}(b). Then, since there is already an increasing path from $w$ to the center $s$ of $S$ within the 4-cycle induced by $\{s,v_1,v_2,w\}$, we need to be careful that, when we concatenate the increasing path $P^+_{wz}$ with the edge $zs$, a second increasing path from $w$ to $s$ does not appear. To prevent this, if we let $\lambda_w^+$ be the label of the edge of $P^+_{wz}$ incident with $z$, we need that \begin{equation}\label{eq:1interesting-1}
\lambda_w^+ > \lambda_s.
\end{equation}
The above equation imposes a constraint on the choice of $\lambda_s$ and already shows, in contrast to how we dealt with $0$-interesting stars in \autoref{claim:0interesting-stars}, that the labels of the edges incident with a $1$-interesting star {\sl cannot} be chosen obliviously to the labels of the rest of the graph. Symmetrically, since there is already a decreasing path from $w$ to the center $s$ of $S$ within the 4-cycle induced by $\{s,v_1,v_2,w\}$, we need to be careful that, when we concatenate the decreasing path $P^-_{wz}$ with the edge $zs$, a second decreasing path from $w$ to $s$ does not appear. To prevent this, if we let $\lambda_w^-$ be the label of the edge of $P^-_{wz}$ incident with $z$, we need that \begin{equation}\label{eq:1interesting-2}
\lambda_w^- < \lambda_s.
\end{equation}
Note that, if the increasing path $P^+_{wz}$ or the decreasing path  $P^-_{wz}$ does not exist in $G[X \cup B]$, then the corresponding constraint is void.

If there were only type-2 leaves in all the $1$-interesting stars in $G \setminus X$ (which can be arbitrarily many, not bounded by any function of $k$), the choice of the labels for $1$-interesting stars seems a manageable task, assuming that the labels of $G[X \cup B]$ have been already chosen: for every $1$-interesting star $S$ with center $s$, and for every neighbor $w \in X$ of a pair of type-2 leaves of $S$, add the constraints given by \autoref{eq:1interesting-1} and \autoref{eq:1interesting-2} about $\lambda_s$ to a system of linear inequalities of size polynomial in the size of $G$, whose feasibility can be checked in polynomial time using, for instance, Gaussian elimination. However, the presence of type-1 leaves makes our task more complicated, as we proceed to discuss.

Consider now a type-1 leaf $v$ of the star $S$ depicted in \autoref{fig:01interesting-stars}(b), and let $u$ be its neighbor in $X$ (note that $u \neq z$, as otherwise $\{s,v,u\}$ would induce a triangle and \autoref{rule:1} could be applied). Similarly as above, in $G[X \cup B]$ there is at most one increasing path $P_{uz}^+$ from $u$ to $z$, and at most one decreasing path $P_{uz}^-$. Let $\lambda_u^{+}$ (resp. $\lambda_u^{-}$) be the label of the last edge of $P_{uz}^+$ (resp. $P_{uz}^-$), which we assume to be known. Now consider the 2-edge path induced by $\{s,v,u\}$. Depending on the labels that we choose for the edges $sv$ and $vu$, this path will be increasing or decreasing from $s$ to $u$, but not both. In the former case, namely if the 2-edge path increases from $s$ to $u$, similarly to the above discussion we need to satisfy the constraint $\lambda_u^+ > \lambda_s$. In the latter case, namely if the 2-edge path decreases from $s$ to $u$,  we need to satisfy the constraint $\lambda_u^- < \lambda_s$. But in this case it is {\sl not} true anymore that we need to satisfy {\sl both} constraints, as which one needs to be satisfied depends on the {\sl choice of the direction of growth} of the 2-edge path between $s$ and $u$. Thus, it appears that, for this type of leaf, we need to satisfy the disjunctive constraint given by
 \begin{equation}\label{eq:1interesting-3}
\lambda_u^+ > \lambda_s\ \text{ or }\ \lambda_u^- < \lambda_s.
\end{equation}
To deal with the set of disjunctive constraints  of the form of \autoref{eq:1interesting-3}, one may try to use the existing literature on solving systems of linear equations with disjunctive constraints, such as~\cite{Koubarakis01}. Unfortunately, the existing results do not seem to be directly applicable to our setting.

We circumvent this by reformulating the problem in terms of what we call a \textit{labeling relation}, defined below. This reinterpretation allows, on the one hand, to guess in time \FPT all the labelings restricted to $G[X \cup B]$ and, more importantly, once the labeling of $G[X \cup B]$  has been fixed, it allows to formulate the problem in terms of the satisfiability of a $2$-\SAT formula, which can be decided in polynomial time. We now present the corresponding definitions and formally present the algorithm.

Before that, let us just observe that, while type-1 leaves seem to be more complicated to deal with than type-2 leaves,
we could get rid (even if we do not need it in our algorithm) of every $1$-interesting star $S$ having only leaves of type 1:
indeed, in that case, the edges linking the leaves of $S$ to $X$, together with the edge from the center of $S$ to $X$,
constitute the set of edges of a matching cut of $G$, and thus could be removed by \autoref{rule:matching-cut}. Then $V(S)$ could be removed by combining \autoref{rule:2} and the fact that a star is clearly good.


\subparagraph*{Reinterpretation of the problem with labeling relations.} In order to decide whether $G$ admits a \gel, we first observe the following. For the sake of an edge-labeling $\lambda:E(G) \to \mathds{R}$ being good, the actual values taken by $\lambda$ do not really matter: what matters is, for every pair of edges $e,f \in E(G)$, the {\sl relation} between $\lambda(e)$ and $\lambda(f)$, that is, whether $\lambda(e) > \lambda(f)$, $\lambda(e) < \lambda(f)$, or $\lambda(e) = \lambda(f)$. To formalize this point of view, we say that a function ${\sf rel}: E(G)\times E(G) \to \{0,1,2\}$ is a \textit{labeling relation} if it satisfies the following properties:
\begin{itemize}
    \item (Reflexivity) For every edge $e \in E(G)$, ${\sf rel}(e,e)=0$.
    \item (Symmetry) For every two distinct edges $e_1,e_2 \in E(G)$,
    \begin{itemize}
    \item ${\sf rel}(e_1,e_2)=0$ if and only if ${\sf rel}(e_2,e_1)=0$; and
    \item ${\sf rel}(e_1,e_2)=1$ if and only if ${\sf rel}(e_2,e_1)=2$.
    \end{itemize}
\item (Transitivity) For every three distinct edges $e_1,e_2,e_3 \in E(G)$,
    \begin{itemize}
    \item if ${\sf rel}(e_1,e_2) = {\sf rel}(e_2,e_3) =1 $, then ${\sf rel}(e_1,e_3)=1$;
    \item if ${\sf rel}(e_1,e_2) = {\sf rel}(e_2,e_3) =0 $, then ${\sf rel}(e_1,e_3)=0$; and
    \item if ${\sf rel}(e_1,e_2) + {\sf rel}(e_2,e_3) = 1$,  then ${\sf rel}(e_1,e_3) = 1$.
    \end{itemize}
\end{itemize}
To get some intuition, the fact that ${\sf rel}(e_1,e_2)=1$ (resp. ${\sf rel}(e_1,e_2)=2$) should be interpreted as `the label of $e_1$ is strictly greater (resp. smaller) than the label of $e_2$'. The fact that ${\sf rel}(e_1,e_2)=0$ should be interpreted as `$e_1$ and $e_2$ have the same label'. It is worth mentioning that a labeling relation of $G$ could also be viewed as a partial orientation of the edges of the line graph of $G$ satisfying the corresponding transitivity properties; we will not use this viewpoint in this section, but we will use it in the dynamic programming algorithm of \autoref{sec:DPdegree}.


To get some further intuition and relate edge-labelings to labeling relations, observe that if $\lambda:E(G) \to \mathds{R}$ is an edge-labeling of $G$, then the function ${\sf rel}_{\lambda}: E(G)\times E(G) \to \{0,1,2\}$ defined, for every $e,f \in E(G)$, by ${\sf rel}_{\lambda}(e,f)=0$ if and only if $\lambda(e)=\lambda(f)$ (which includes the case where $e=f$) and ${\sf rel}_{\lambda}(e,f)=1$ if and only if $\lambda(e) > \lambda(f)$, is easily seen to be a labeling relation. Conversely, given a labeling relation ${\sf rel}: E(G)\times E(G) \to \{0,1,2\}$, we can define an edge-labeling $\lambda_{{\sf rel}}:E(G) \to \mathds{R}$ by the following inductive procedure.

Let $e \in E(G)$ be such that there is no edge $f \in E(G)$ with ${\sf rel}(e,f)=1$ (note that such an edge exists by definition of a labeling relation). Then we set $\lambda_{{\sf rel}}(e)=0$ and $\lambda_{{\sf rel}}(g)=0$ for every edge $g \in E(G)$ such that ${\sf rel}(e,g)=0$. Assume inductively that the algorithm has attributed label $c$ for some integer $c \geq 0$, let $M \subseteq E(G)$ be the set of edges unlabeled so far, and let $E_c = \{e \in E(G) \mid \lambda_{{\sf rel}}(e) = c\}$. Then we set $\lambda_{{\sf rel}}(e)=c+1$ for every edge $e \in M$ such that
\begin{itemize}
    \item there exists $f \in E_c$ such that ${\sf rel}(e,f)=1$; and
    \item there do not exist  $f \in E_c$ and $g \in M$ such that ${\sf rel}(e,g)={\sf rel}(g,f)=1$.
\end{itemize}
We can naturally speak of an \textit{increasing} or \textit{decreasing} path with respect to a labeling relation ${\sf rel}$, by considering the corresponding increase or decrease in the labels given by the edge-labeling $\lambda_{{\sf rel}}$ defined above.

We say that a labeling relation ${\sf rel}: E(G)\times E(G) \to \{0,1,2\}$ is \textit{good} if the edge-labeling $\lambda_{{\sf rel}}:E(G) \to \mathds{R}$ defined above is good.

\begin{lemma}\label{claim:good-labeling-relation}
A graph $G$ admits a \gel if and only if it admits a good labeling relation.
\end{lemma}
\begin{proof}
Let first ${\sf rel}: E(G)\times E(G) \to \{0,1,2\}$ be a good labeling relation of $G$. Then, by definition,  the edge-labeling $\lambda_{{\sf rel}}:E(G) \to \mathds{R}$ is good, certifying that $G$ admits a \gel.

Conversely, let $\lambda:E(G) \to \mathds{R}$ be a \gel of $G$, and we claim that the labeling relation ${\sf rel}_{\lambda}$ defined above if good. That is, if we let $\hat{\lambda}:= \lambda_{{\sf rel}_{\lambda}}$, we need to prove that $\hat{\lambda}$ is a \gel of $G$. Since $\lambda$ is a \gel of $G$, it is enough to prove that, for any two edges $e,f \in E(G)$, $\lambda(e) > \lambda(f)$ if and only if $\hat{\lambda}(e) > \hat{\lambda}(f)$. Suppose first that $\lambda(e) > \lambda(f)$. Then, by the definition of ${\sf rel}_{\lambda}$, ${\sf rel}_{\lambda}(e,f)=1$, which implies by the definition of $\lambda_{{\sf rel}_{\lambda}}$ that $\hat{\lambda}(e) > \hat{\lambda}(f)$. On the other hand, if $\hat{\lambda}(e) > \hat{\lambda}(f)$, then necessarily ${\sf rel}_{\lambda}(e,f)=1$, implying in turn that $\lambda(e) > \lambda(f)$.
\end{proof}

By \autoref{claim:good-labeling-relation}, we can now focus on deciding whether $G$ admits a good labeling relation. The advantage of dealing with good labeling relations, with respect to good edge-labelings, is that if we have at hand a set of edges $F \subseteq E(G)$ of size bounded by $f(k)$ for some function $f$, then we can guess in time \FPT all possible good labeling relations restricted to pairs of edges in $F$. Note that, a priori, it is not clear how to guess such a restriction of a good edge-labeling in time \FPT.

Recall that, at this point, we can assume that every connected component of $G \setminus X$ is a well-behaved star that is either boring or $1$-interesting. \autoref{claim:few-stars} and the fact that all stars are well-behaved imply that $|B| \leq (2k+1)k^2$ and that $|E(B,B)| \leq 2k \cdot k^2$. Hence, since $|X| \leq k$,
\begin{equation}\label{eq:few-boring-edges}
|E(G[X \cup B])| \leq {k \choose 2} + 2k \cdot k^2 + (2k+1)k^2\cdot k = \Ocal(k^4).
\end{equation}


We now show that we can restrict ourselves to a particular type of good labeling relations that will simplify our task. To this end, let $L \subseteq E(G)$ be the set of edges incident with a leaf of a $1$-interesting star (that is, $L$ does not contain the unique edge between the center of a star and $X$). We say that a labeling relation ${\sf rel}: E(G)\times E(G) \to \{0,1,2\}$ is \textit{standard} if
\begin{itemize}
    \item for every two distinct edges $e,f \in E(G) \setminus L$, ${\sf rel}(e,f)\neq 0$;
    \item the edges in $L$ can be partitioned into two sets $L_{\sf small}$ and $L_{\sf big}$ such that
    \begin{itemize}
        \item if $e,f \in L_{\sf small}$ or $e,f \in L_{\sf big}$, then ${\sf rel}(e,f)=0$;
        \item if $e \in L_{\sf small}$ and $f \in E(G) \setminus L_{\sf small}$, then ${\sf rel}(e,f)=2$; and
        \item if $e \in L_{\sf big}$ and $f \in E(G) \setminus L_{\sf big}$, then ${\sf rel}(e,f)=1$.
    \end{itemize}
\end{itemize}

Translating the above definition to edge-labelings, a labeling relation is standard if the set $L \subseteq E(G)$ can be partitioned into two sets of ``very light'' and ``very heavy'' edges, namely $L_{\sf small}$ and $L_{\sf big}$, and all other edges of $G$ are linearly ordered, in the sense that there are no two edges in $E(G) \setminus L$ with the same label.

In the next lemma we prove that it is enough to look for a good {\sl standard} labeling relation. Before that, let us adapt \autoref{obs:local_minima} to the context of labeling relations. Let $C$ be a cycle of $G$ and let ${\sf rel}: E(G)\times E(G) \to \{0,1,2\}$ be a labeling relation. A \textit{local minimum} (resp. \textit{local maximum}) of $C$, with respect to ${\sf rel}$, is a subpath $P$ of $C$ such that for every $e,f \in E(P)$ it holds that ${\sf rel}(e,f)=0$ and, for every $g \in E(C) \setminus E(P)$ incident to an edge $e \in E(C)$, it holds that ${\sf rel}(e,g)=2$ (resp. ${\sf rel}(e,g)=1$). We can now reformulate \autoref{obs:local_minima} as follows.
 \begin{observation}\label{obs:local_minima_for_relations}
 A labeling relation {\sf rel} of a graph $G$ is good if and only if every cycle of $G$ admits at least two local minima, or at least two local maxima, with respect to {\sf rel}.
\end{observation}

\begin{lemma}\label{claim:always-standard}
$G$ admits a good labeling relation if and only if $G$ admits a good standard labeling relation.
\end{lemma}
\begin{proof}
    The only non-trivial implication is that if $G$ admits a good labeling relation ${\sf rel}$, then $G$ admits a good standard labeling relation ${\sf rel}'$. Let us build ${\sf rel}'$ from ${\sf rel}$ starting with ${\sf rel}'(e,f)={\sf rel}(e,f)$ for every $e,f \in E(G)$.

    The first property, namely that for every two distinct edges $e,f \in E(G) \setminus L$, ${\sf rel}'(e,f)\neq 0$, is easy to achieve. Indeed, let $F \subseteq E(G) \setminus L$ be an inclusion-wise maximal set of edges such that, for any two edges $e,f \in F$, ${\sf rel}(e,f) = 0$ (that is, all the edges in $F$ ``have the same label''). Then we just order $F$ arbitrarily as $e_1,\ldots,e_p$ and for $e_i,e_j \in F$ with $i\neq j$, we redefine ${\sf rel}'(e_i,e_j) = 1$ if and only if $i > j$. After applying this operation exhaustively, it is easy to check that, since ${\sf rel}$ is good, the current labeling relation ${\sf rel}'$ is also good.

    Let us now focus on the second property concerning the sets $L_{\sf small}$ and $L_{\sf big}$ that partition $L$. To this end, we further modify ${\sf rel}'$ as follows. Starting with $L_{\sf small}=L_{\sf big}=\emptyset$, we proceed to grow these two sets until they partition $L$ and, once this is done, we just redefine ${\sf rel}'$ so that it satisfies the three conditions in the definition of standard labeling relation concerning edges in $L$.

    To this end, let $S$ be a $1$-interesting star in $G \setminus X$ with center $s$; see \autoref{fig:01interesting-stars}(b). Let $v_1,v_2$ be a pair of type-2 leaves of $S$ with common neighbor $w \in X$. Since ${\sf rel}$ is good, necessarily ${\sf rel}(sv_1,v_1w)=1$ and ${\sf rel}(sv_2,v_2w)=2$, or ${\sf rel}(sv_1,v_1w)=2$ and ${\sf rel}(sv_2,v_2w)=1$. Suppose without loss of generality that the former holds. Then we
    add $sv_2$ and $v_1w$  to $L_{\sf small}$, and $sv_1$ and $v_2w$  to $L_{\sf big}$. Let now $v$ be a type-1 leaf of $S$ with neighbor $u \in X$. If ${\sf rel}(sv,vu) \leq 1$, we add $vu$ to $L_{\sf small}$ and $sv$ to $L_{\sf big}$, and otherwise (that is, if ${\sf rel}(sv,vu) = 2$), we add $sv$ to $L_{\sf small}$ and $vu$ to $L_{\sf big}$.

    It is easy to verify that the obtained ${\sf rel}'$ is still a labeling relation of $G$, and it is standard by construction. It just remains to verify that it is good. By \autoref{obs:local_minima_for_relations}, this is equivalent to verifying that every
 cycle $C$ of $G$ admits two local minima (or two local maxima) with respect to ${\sf rel}'$. For an edge $e \in L_{\sf small}$ in a cycle $C$, we denote by $S^C_e$ the maximal subpath of $C$ containing $e$ such that all its edges are in $L_{\sf small}$. Consider an arbitrary cycle $C$ of $G$. If $C$ does not intersect any $1$-interesting star, it admits two local minima because ${\sf rel}$ is good and the changes in ${\sf rel}'$ with respect to ${\sf rel}$ in $C$ can only increase the number of minima. Hence, we can assume that $C$ intersects a $1$-interesting star $S$, and thus it contains its center $s$. Let $z$ be the neighbor of $s$ in $X$.
 We distinguish two cases, which are clearly exhaustive:
 \begin{itemize}
     \item $C$ contains two leaves $v_1,v_2$ of $S$. If $v_1,v_2$ share a neighbor $w$ in $C$ other than $s$, then necessarily $w \in X$ and $C$ is a 4-cycle, which must admit two local minima because ${\sf rel}$ is good (and ${\sf rel}'$ coincides with ${\sf rel}$ for those edges). Otherwise (in this case, $C$ looks like the red cycle in \autoref{fig:01interesting-stars}(a)), by construction of the sets $L_{\sf small}$ and $L_{\sf big}$, necessarily $C$ contains two edges $e_1,e_2 \in L_{\sf small}$. We distinguish two cases. If at least one of $e_1$ or $e_2$ is {\sl not} incident with $s$, then the associated subpaths $S^C_{e_1}$ and $S^C_{e_2}$ of $C$ are vertex-disjoint, hence defining two local minima of $C$. Otherwise, that is, if both $e_1$ and $e_2$ are incident with $s$, let $e_1'$ (resp. $e_2'$) be the edge of $C$ incident with $e_1$ (resp. $e_2$) not containing $s$. Then both $e_1',e_2' \in L_{\sf big}$, and they are contained in two vertex-disjoint subpaths of $C$ that define two local maxima.
     \item $C$ contains only one leaf $v$ of $S$ (such as vertex $v$ depicted in \autoref{fig:01interesting-stars}(b)). In that case, necessarily $C$ contains the edge $sz$. Let $u$ be the neighbor of $v$ in $X$.
     Assume without loss of generality that $uv \in L_{\sf small}$ and $vs \in L_{\sf big}$, the other case being totally symmetric. Then, since $C$ admits two local minima with respect to ${\sf rel}$, the subpath $P$ from $u$ to $s$ in $C$ not containing $v$ cannot be increasing with respect to the edge-labeling associated with ${\sf rel}$, and this property is clearly maintained by ${\sf rel}'$. Hence, following $P$ starting from $u$, necessarily there are two consecutive edges $e_1,e_2 \in E(C)$ such that ${\sf rel}'(e_1,e_2)=1$. Then, the subpaths $S^C_{uv}$ and $S^C_{e_2}$ of $C$ are vertex-disjoint and define two local minima of $C$, concluding the proof.
 \end{itemize}\vspace{-.5cm}
\end{proof}

\subparagraph*{2-\SAT formulation.} We are now ready to present our algorithm to decide whether $G$ admits a \gel. By combining \autoref{claim:good-labeling-relation} and \autoref{claim:always-standard}, the problem is equivalent to deciding whether $G$ admits a good standard labeling relation. Let ${\sf rel}:E(G) \times E(G) \to \{0,1,2\}$ be the standard labeling relation we are looking for. Recall that $X \subseteq V(G)$ is the modulator to a star-forest with $|X| \leq k$, $B \subseteq V(G)$ is the set of vertices belonging to a boring star of $G \setminus X$, and $L \subseteq E(G)$ is the set of edges incident with a leaf of a $1$-interesting star. Let $F \subseteq E(G)$ be the set of edges joining a center of a $1$-interesting star with its neighbor in $X$. Then note that $E(G)$ can be partitioned into $E(G)= E(G[X \cup B]) \uplus L \uplus F$. We start by guessing ${\sf rel}$ restricted to pairs of edges in $E(G[X \cup B])$. Since ${\sf rel}$ is standard, this is equivalent to guessing a linear ordering of the edges in $E(G[X \cup B])$. Since by \autoref{eq:few-boring-edges} we have that $|E(G[X \cup B])| = \Ocal(k^4)$, we have $|E(G[X \cup B])|! = 2^{\Ocal(k^4 \log k)}$ many choices for this linear ordering. We stress that this is the only step of the algorithm that does not run in polynomial time. Clearly, we can discard any guess $\rho$ that is not good restricted to $G[X \cup B]$.

Assume henceforth that we have fixed the restriction $\rho$ of ${\sf rel}$ for pairs of edges in $E(G[X \cup B])$, and now our task is to decide whether there exists a good standard labeling relation ${\sf rel}_{\rho}:E(G) \to \{0,1,2\}$ that extends $\rho$. To complete the definition of ${\sf rel}_{\rho}$, it remains to define it for pairs where at least one edge belongs to $L \uplus F$. Since ${\sf rel}_{\rho}$ is standard, its definition for pairs containing an edge in $L$ is given by providing a partition of $L$ into the sets $L_{\sf small}$ and $L_{\sf big}$. To define ${\sf rel}_{\rho}$ for pairs containing an edge in $F$ we have much more freedom: we can insert every edge in $F$ anywhere within the total order of $E(G[X \cup B])$ guessed by $\rho$. To make all these choices, we resort to a formulation of the problem as a 2-\SAT formula.

To this end, fix an arbitrary ordering $e_1,\ldots,e_m$ of $E(G) \setminus L$. Note that this ordering is arbitrary and has nothing to do with $\rho$ and its extension to $L$, we will just use it for our formulation. For every fixed $\rho$, we proceed to define a 2-\SAT formula $\varphi_{\rho}$, and we will prove (cf. \autoref{claim:2SAT-equivalence}) that $G$ admits a good standard labeling relation ${\sf rel}_{\rho}$ that extends $\rho$ if and only if  $\varphi_{\rho}$ is satisfiable.

For every two indices $i,j \in [m]$ with $i < j$ (hence, for
 ${m \choose 2}$ pairs), introduce a binary variable $x_{i,j}$. For convenience, in this section we will use `1' (resp. `0') for a \T (resp. \F) assignment of a variable.
 We will interpret $x_{i,j}=1$  as ${\sf rel}_{\rho}(e_i,e_j)=1$ (that is, the label of $e_i$ is strictly greater than that of $e_j$), and $x_{i,j}=0$ as ${\sf rel}_{\rho}(e_i,e_j)=2$. Since $e_i,e_j \in E(G) \setminus L$ and the desired labeling relation ${\sf rel}_{\rho}$ is standard, we can safely disconsider the possibility that ${\sf rel}_{\rho}(e_i,e_j)=0$. The formula consists of the following clauses:
 \begin{itemize}
     \item For every pair of edges $e_i,e_j \in E(G[X \cup B])$ with $i < j$ (that is, pairs of edges for which ${\sf rel}_{\rho}$ is already fixed), if ${\sf rel}_{\rho}(e_i,e_j)=1$ (resp. ${\sf rel}_{\rho}(e_i,e_j)=2$), then add to $\varphi_{\rho}$ the clause containing only the literal $x_{i,j}$ (resp. $\overline{{x_{i,j}}}$).
     \item For every $1$-interesting star $S$ of $G \setminus X$ with center $s$ and neighbor $z \in X$, suppose that, for some index $i \in [m]$, $e_i$ is the edge between $s$ and $z$. The clauses defined in what follows are inspired by the previous discussion concerning \autoref{fig:01interesting-stars}(b) and \autoref{eq:1interesting-1}, \autoref{eq:1interesting-2}, and \autoref{eq:1interesting-3}, translated to the formalism of label relations and the corresponding literals:
     \begin{itemize}
         \item For every pair of type-2 leaves $v_1,v_2$ of $S$ with common neighbor $w \in X$:
         \begin{itemize}
             \item If $G[X \cup B]$ contains an increasing path (with respect to ${\sf rel}_{\rho}$) from $w$ to $z$, let $e_j$ be the last edge of this path. If $i>j$ (resp. $i<j$), add to $\varphi_{\rho}$ the clause containing only the literal $\overline{x_{i,j}}$ (resp. $x_{i,j}$). This clause plays the role of the constraint imposed by \autoref{eq:1interesting-1}.
             \item If $G[X \cup B]$ contains a decreasing path (with respect to ${\sf rel}_{\rho}$) from $w$ to $z$, let $e_j$ be the last edge of this path. If $i>j$ (resp. $i<j$), add to $\varphi_{\rho}$ the clause containing only the literal ${x_{i,j}}$ (resp. $\overline{x_{i,j}}$). This clause plays the role of the constraint imposed by \autoref{eq:1interesting-2}.
         \end{itemize}
        \item For every  type-1 leaf $v$ of $S$ with neighbor $u \in X$, if $G[X \cup B]$ contains both an increasing path from $u$ to $z$, and a decreasing path from $u$ to $z$, let $e_j$ (resp. $e_{\ell}$) be the last edge of this increasing (resp. decreasing) path. Now the goal is to add to $\varphi_{\rho}$ the clause playing the role of the disjunctive constraint imposed by \autoref{eq:1interesting-3}. But since the variables $x_{i,j}$ of $\varphi_{\rho}$ are only defined for $i < j$, we need to distinguish several cases:
        \begin{itemize}
            \item If $i < \min\{j,\ell\}$, then add to $\varphi_{\rho}$ the clause $(\overline{x_{i,j}}\vee x_{i,\ell})$.

            \item If $j < i < \ell$, then add to $\varphi_{\rho}$ the clause $(x_{j,i}\vee x_{i,\ell})$.

            \item If $\ell < i < j$, then add to $\varphi_{\rho}$ the clause $(\overline{x_{i,j}}\vee \overline{x_{\ell,i}})$.

            \item If $i > \max\{j,\ell\}$, then add to $\varphi_{\rho}$ the clause $(x_{j,i}\vee \overline{x_{\ell,i}})$.

        \end{itemize}

        Note that, for a type-1 leaf $v$, if at least one of the increasing and decreasing paths in $G[X \cup B]$ from $u$ to $z$ is missing, then we do not add any clause to $\varphi_{\rho}$. This makes sense, since if one of these paths is missing, it already yields a safe way to choose the direction of growth of the 2-edge path from $s$ to $u$.
     \end{itemize}
 \end{itemize}

 \subparagraph*{Guaranteeing the transitivity of the obtained labeling relation.}
 Before completing the definition of the 2-\SAT formula $\varphi_{\rho}$, we need to address the following issue.
 Recall that our goal is to extract, from a satisfying assignment of $\varphi_{\rho}$, a good standard labeling relation ${\sf rel}_{\rho}$ of $G$ that extends $\rho$.
 Note that a labeling relation, defined on pairs of edges of $G$, is required to satisfy the transitivity property, which is equivalent to saying that the metric on $E(G)$ defined by the growth of labels needs to satisfy the triangle inequality. 
 But, how is it guaranteed that the given assignment of the variables of ${\sf rel}_{\rho}$ indeed implies that the corresponding labeling satisfies the triangle inequality? A priori, the clauses that we added so far to $\varphi_{\rho}$ do not prevent the possibility that there exist three indices $i < j < \ell$ such that the corresponding variables are assigned the values
 \begin{equation}\label{eq:problem-triangle-inequality1}
     x_{i,j}=x_{j,\ell}=1 \ \text{ and }\ x_{i,\ell}=0.
 \end{equation}
 If we have an assignment satisfying \autoref{eq:problem-triangle-inequality1}, then clearly we will {\sl not} be able to extract a labeling relation from it, since it reads as `the label of the $i$-th edge is strictly greater than the label of the $j$-th edge, which is strictly greater than the label of the $\ell$-th edge, but the label of the $i$-th edge is strictly smaller than the label of the $\ell$-th edge'. Clearly, the other problematic assignment involving indices $i,j,\ell$ is
 \begin{equation}\label{eq:problem-triangle-inequality2}
     x_{i,j}=x_{j,\ell}=0 \ \text{ and }\ x_{i,\ell}=1.
 \end{equation}
A possible solution to this issue is to prevent ``by hand'' the assignments of \autoref{eq:problem-triangle-inequality1} and \autoref{eq:problem-triangle-inequality2}. This could be done by adding, for every three indices $i < j < \ell$, the following two clauses to $\varphi_{\rho}$:
\begin{equation}\label{eq:fixing-transitivity}
     (\overline{x_{i,j}} \vee \overline{x_{j,\ell}} \vee x_{i,\ell}) \wedge (x_{i,j} \vee x_{j,\ell} \vee \overline{x_{i,\ell}}).
 \end{equation}
Note that the first (resp. second) clause above prevents the assignment from \autoref{eq:problem-triangle-inequality1} (resp. \autoref{eq:problem-triangle-inequality2}). Of course, the problem of this approach is that the clauses in \autoref{eq:fixing-transitivity} involve three variables, which would result in a 3-\SAT formula instead of a 2-\SAT formula, as we need for being able to solve the satisfiability problem in polynomial time.

Fortunately, we can still impose the constraints of \autoref{eq:fixing-transitivity} by using clauses containing only two variables.
The idea is the following: in order to guarantee the transitivity of the obtained labeling relation, we do not need to add the clauses of \autoref{eq:fixing-transitivity} for {\sl every} triple $i < j < \ell$, but only for those that are bound by some of the clauses added so far to $\varphi_{\rho}$ corresponding to type-1 leaves of $1$-interesting stars.
Indeed, if there is no constraint at all among the pairs of indices of two variables $x_{i,j}$ and $x_{\ell,t}$, then we can choose any assignment for the variables $x_{a,b}$, with $a \in \{i,j\}$ and $b \in \{\ell,t\}$, so that the transitivity clauses of \autoref{eq:fixing-transitivity} are satisfied.

Thus, let us focus on these problematic clauses, namely those corresponding to a type-1 leaf $v$ of a $1$-interesting star $S$ of $G \setminus X$ as depicted in \autoref{fig:transitive} (which is a simplified version of \autoref{fig:01interesting-stars}(b) using the current notation), where assuming without loss of generality that $i < j < \ell$, the clause $(\overline{x_{i,j}}\vee x_{i,\ell})$ has been added to $\varphi_{\rho}$.
This clause $(\overline{x_{i,j}}\vee x_{i,\ell})$ involves three indices $i,j,\ell$, so we do need to add to $\varphi_{\rho}$ the clauses in \autoref{eq:fixing-transitivity} for guaranteeing the transitivity of the corresponding labeling relation. The crucial observation is that, since both edges $e_j$ and $e_{\ell}$ belong to $E(G[X \cup B])$, the value of $x_{j,\ell}$ has been already {\sl fixed} by the guess of $\rho$, which reduces the arity of the clauses in \autoref{eq:fixing-transitivity}  from three to two. Indeed, if $x_{j,\ell}=1$, then the clauses of \autoref{eq:fixing-transitivity} boil down to the single clause $(\overline{x_{i,j}} \vee x_{i,\ell})$ involving only two variables, while if
$x_{j,\ell}=0$, then the clauses of \autoref{eq:fixing-transitivity} boil down to $(x_{i,j} \vee \overline{x_{i,\ell}})$. Note that, in the former case, the newly added clause is redundant, since $(\overline{x_{i,j}} \vee x_{i,\ell})$ had been already added to $\varphi_{\rho}$.


\begin{figure}[h!tb]
\begin{center}
\vspace{-.25cm}
\includegraphics[scale=1.00]{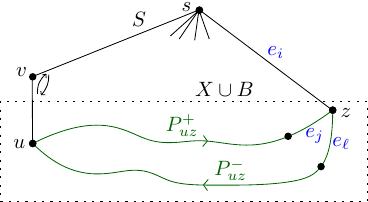}
\end{center}
\caption{Assuming that $i < j < \ell$, the clause $(\overline{x_{i,j}}\vee x_{i,\ell})$ has been added to $\varphi_{\rho}$. For guaranteeing transitivity, if $x_{j,\ell}=0$ according to $\rho$, we add to $\varphi_{\rho}$ the clause  $(x_{i,j} \vee \overline{x_{i,\ell}})$.}\label{fig:transitive}
\end{figure}

Thus, for guaranteeing the transitivity of the obtained labeling relation, it is enough, every time that we add to $\varphi_{\rho}$ a clause corresponding to a leaf of a $1$-interesting star, to add as well the clause described above (which is also described in the caption of \autoref{fig:transitive}). We call these newly added clauses the \textit{transitivity clauses} of $\varphi_{\rho}$.

This completes the construction of the 2-\SAT formula $\varphi_{\rho}$.

\subparagraph*{Equivalence between admitting a \gel and satisfying the 2-\SAT formula.} We now prove that the satisfiability of $\varphi_{\rho}$ indeed corresponds to the existence of a good standard labeling relation of $G$ that extends $\rho$.

\begin{lemma}\label{claim:2SAT-equivalence}
    $G$ admits a good standard labeling relation ${\sf rel}_{\rho}$ that extends $\rho$ if and only if the 2-\SAT formula $\varphi_{\rho}$ is satisfiable.
\end{lemma}
\begin{proof}
Suppose first that $G$ admits a good standard labeling relation ${\sf rel}_{\rho}$ that extends $\rho$, and we proceed to define an assignment $\alpha$ of the variables occurring in $\varphi_{\rho}$ that satisfies all its clauses.

For every pair of edges $e_i,e_j \in E(G[X \cup B])$ with $i < j$ (that is, pairs of edges for which ${\sf rel}_{\rho}$ is already fixed), if ${\sf rel}_{\rho}(e_i,e_j)=1$ (resp. ${\sf rel}_{\rho}(e_i,e_j)=2$), then we set $\alpha(x_{i,j})=1$ (resp. $\alpha(x_{i,j})=0$). Clearly, the clauses of $\varphi_{\rho}$ corresponding to the values fixed by $\rho$ are satisfied by this assignment.

Consider now a $1$-interesting star $S$ of $G \setminus X$ with center $s$ and neighbor $z \in X$, and suppose that, for some index $i \in [m]$, $e_i$ is the edge between $s$ and $z$.

Suppose first that $v_1,v_2$ is a pair of type-2 leaves of $S$ with common neighbor $w \in X$. Since by hypothesis ${\sf rel}_{\rho}$ is a good standard labeling relation of $G$, $G[X \cup B]$ contains at most one increasing path (with respect to ${\sf rel}_{\rho}$) from $w$ to $z$.
If this is indeed the case, let $e_j$ be the last edge of this path.
If $i>j$ (resp. $i<j$), then we set $\alpha(x_{i,j})=0$ (resp. $\alpha(x_{i,j})=1$).
Symmetrically, $G[X \cup B]$ contains at most one decreasing path (with respect to ${\sf rel}_{\rho}$) from $w$ to $z$.
If this is indeed the case, let $e_j$ be the last edge of this path. If $i>j$ (resp. $i<j$), then we set $\alpha(x_{i,j})=1$ (resp. $\alpha(x_{i,j})=0$).
By definition of $\varphi_{\rho}$, it can be easily checked that the clauses corresponding to type-2 leaves are satisfied by this assignment.

Suppose now that $v$ is a type-1 leaf $v$ of $S$ with neighbor $u \in X$. Again, since ${\sf rel}_{\rho}$ is a good standard labeling relation of $G$, $G[X \cup B]$ contains at most one increasing path (with respect to ${\sf rel}_{\rho}$) from $u$ to $z$, and at most one decreasing path from $u$ to $z$. If both these paths exist, let $e_j$ (resp. $e_{\ell}$) be the last edge of this increasing (resp. decreasing) path. Since ${\sf rel}_{\rho}$ is good, necessarily ${\sf rel}_{\rho}(e_i,e_j)=2$ or ${\sf rel}_{\rho}(e_i,e_{\ell})=1$ (or both); see \autoref{fig:transitive}. Following the definition of the clauses added to $\varphi_{\rho}$, we distinguish several cases depending on the relative ordering of the indices $i,j,\ell$, and we define the assignments of the variables involving $i,j,\ell$ (recall that a variable $x_{a,b}$ only exists in $\varphi_{\rho}$ if $a < b$) as defined in \autoref{tab:assignment}, where the last row shows the clause of $\varphi_{\rho}$ satisfied by the corresponding assignment, using the fact that ${\sf rel}_{\rho}(e_i,e_j)=2$ or ${\sf rel}_{\rho}(e_i,e_{\ell})=1$ (or both).
\begin{table}[h]
\centering
\begin{tabular}{|c||c|c|c|c|}
\hline
relation $\ \ \setminus\ \ $ order of $i,j, \ell$ & $i < \min\{j,\ell\}$ & $j < i < \ell$ & $\ell < i < j$ & $i > \max\{j,\ell\}$\\
\hline
\hline
 if ${\sf rel}_{\rho}(e_i,e_j)=1$, we define & $\alpha(x_{i,j})=1$ & $\alpha(x_{j,i})=0$ & $\alpha(x_{i,j})=1$ & $\alpha(x_{j,i})=0$ \\
if ${\sf rel}_{\rho}(e_i,e_j)=2$, we define & $\alpha(x_{i,j})=0$ & $\alpha(x_{j,i})=1$ & $\alpha(x_{i,j})=0$ & $\alpha(x_{j,i})=1$ \\

\hline

if ${\sf rel}_{\rho}(e_i,e_\ell)=1$, we define & $\alpha(x_{i,\ell})=1$ & $\alpha(x_{i,\ell})=1$ & $\alpha(x_{\ell,i})=0$ & $\alpha(x_{\ell,i})=0$ \\
if ${\sf rel}_{\rho}(e_i,e_\ell)=2$, we define  & $\alpha(x_{i,\ell})=0$ & $\alpha(x_{i,\ell})=0$ & $\alpha(x_{\ell,i})=1$ & $\alpha(x_{\ell,i})=1$ \\
\hline
clause of $\varphi_{\rho}$ satisfied by $\alpha$ & $(\overline{x_{i,j}}\vee x_{i,\ell})$ & $(x_{j,i}\vee x_{i,\ell})$ & $(\overline{x_{i,j}}\vee \overline{x_{\ell,i}})$ & $(x_{j,i}\vee \overline{x_{\ell,i}})$\\
\hline
\end{tabular}
\caption{Assignment $\alpha$ of the variables involving indices $i,j,\ell \in [m]$ depending on their relative order and on the values of ${\sf rel}_{\rho}(e_i,e_j)$ and ${\sf rel}_{\rho}(e_i,e_{\ell})$. The last row shows the clause of $\varphi_{\rho}$ satisfied by the corresponding assignment, using the fact that, since ${\sf rel}_{\rho}$ is good, necessarily  ${\sf rel}_{\rho}(e_i,e_j)=2$ or ${\sf rel}_{\rho}(e_i,e_{\ell})=1$ (or both).
}
\label{tab:assignment}
\end{table}

To complete the definition of the assignment $\alpha$, we take its transitive closure, that is, for any three indices $i < j < \ell$, if
$\alpha(x_{i,j})=\alpha(x_{j,\ell})=1$, we set $\alpha(x_{i,\ell})=1$, and if $\alpha(x_{i,j})=\alpha(x_{j,\ell})=0$, we set $\alpha(x_{i,\ell})=0$. Finally, we choose any value for the variables of $\varphi_{\rho}$ that were not considered so far. This assignment satisfies the transitivity clauses of $\varphi_{\rho}$.

\medskip

Conversely, suppose now that the 2-\SAT formula $\varphi_{\rho}$ described above is satisfiable, let $\alpha$ be a satisfying assignment of the variables, and we proceed to define from it a good standard labeling relation ${\sf rel}_{\rho}$ of $G$ that extends $\rho$. By the symmetry of a labeling relation and the definition of a standard one, we only need to define the sets $L_{\sf small}$ and $L_{\sf big}$ that partition the set $L \subseteq E(G)$, and to define ${\sf rel}_{\rho}(e_i,e_j)$ for indices $i,j \in [m]$ with $i< j$. Let us first define the latter. For any pair $i,j \in [m]$ with $i< j$, we define ${\sf rel}_{\rho}(e_i,e_j)=1$ if $\alpha(x_{i,j})=1$, and ${\sf rel}_{\rho}(e_i,e_j)=2$ otherwise. Note that, if
 both $e_i$ and $e_j$ belong to $E(G[X \cup B])$, then the definition of $\varphi_{\rho}$  implies that ${\sf rel}_{\rho}(e_i,e_j)$ indeed coincides with $\rho$.

 Let us now define the partition of $L$ into $L_{\sf small}$ and $L_{\sf big}$. To this end, let $S$ be a $1$-interesting star of $G \setminus X$ with center $s$, and let $z \in X$ be the neighbor of $s$.
 Consider first a pair $v_1,v_2$ of type-2 leaves of $S$ with common neighbor $w \in X$; see \autoref{fig:01interesting-stars}(b).
 Then we add the edges $sv_1$ and $wv_2$ to $L_{\sf small}$, and the edges $sv_2$ and $wv_1$ to $L_{\sf big}$. Consider now a type-1 leaf $v$ of $S$ with neighbor $u \in X$; see \autoref{fig:transitive}.
 Since we can assume that ${\sf rel}_{\rho}$ restricted to $G[X \cup B]$ is good (as otherwise, this choice of $\rho$ can be discarded), $G[X \cup B]$ contains at most one increasing path (according to ${\sf rel}_{\rho}$) $P_{uz}^+$ from $u$ to $z$, and at most one decreasing path (according to ${\sf rel}_{\rho}$) $P_{uz}^-$ from $u$ to $z$. If at least one of the paths is missing, we add arbitrarily one of the edges $sv$ and $vu$ into $L_{\sf small}$, and the other one into $L_{\sf big}$.
 If both paths $P_{uz}^+$ and $P_{uz}^-$ exist, let $e_j$ (resp. $e_{\ell}$) be the edge of $P_{uz}^+$ (resp. $P_{uz}^-$) incident with $z$.
 Then the clauses that we added to $\varphi_{\rho}$ (inspired by \autoref{eq:1interesting-3}) imply that ${\sf rel}_{\rho}(e_i,e_j)=2$, or ${\sf rel}_{\rho}(e_i,e_{\ell})=1$, or both.
 If both hold, we add arbitrarily one of the edges $sv$ and $vu$ into $L_{\sf small}$, and the other one into $L_{\sf big}$. If only ${\sf rel}_{\rho}(e_i,e_j)=2$ holds, we add $sv$ into $L_{\sf small}$ and  $vu$ into $L_{\sf big}$. Finally, if only ${\sf rel}_{\rho}(e_i,e_{\ell})=1$ holds, we add $sv$ into $L_{\sf big}$ and  $vu$ into $L_{\sf small}$. This completes the definition of ${\sf rel}_{\rho}$.

 By construction, it is easy to verify that ${\sf rel}_{\rho}$ is a standard labeling relation of $G$ that extends $\rho$.
 To show that it is good, by \autoref{obs:local_minima_for_relations},
 it is equivalent to prove that every cycle $C$ of $G$ admits at least two local minima (with respect to ${\sf rel}_{\rho}$). The remaining of the proof uses arguments similar to those used in the proof of \autoref{claim:always-standard}, and we omit it here (for completeness, it can be found in the extended full version of the paper available online~\cite{GEL-arXiv}).
\end{proof}

\subparagraph*{Wrapping up the algorithm.} We finally have all the ingredients to formally state our algorithm in \autoref{alg:FPT-star-modulator}, where we include comments that justify each of the steps.

\begin{algorithm}[h!tb]

\SetKwInOut{Input}{Input} \SetKwInOut{Output}{Output}

\medskip

\Input{An $n$-vertex graph and a star-forest modulator $X \subseteq V(G)$ with $|X|\leq k$.}\Output{A \gel of $G$, if it exists, or the report that $G$ is bad.}

\medskip

Apply Rules~\ref{rule:1},~\ref{rule:2}, and~\ref{rule:3} exhaustively, in any order.

\tcc{Every star in $G \setminus X$ is well-behaved (\autoref{claim:star-assumptions})}

\tcc{$G \setminus X$ contains at most $k^2$ boring stars (\autoref{claim:few-stars})}

If $G \setminus X$ contains a $0$-interesting star, apply \autoref{rule:4} and delete it.

\tcc{The above line is safe by \autoref{claim:0interesting-stars}}

\tcc{By  \autoref{claim:good-labeling-relation} and \autoref{claim:always-standard}, we focus on deciding whether $G$ admits a good standard labeling relation}

Guess all possible $2^{\Ocal(k^4 \log k)}$ good standard labeling relations $\rho$ of $G[X \cup B]$.\label{algo-line-guessing}

For each such a $\rho$, build the corresponding 2-\SAT formula $\varphi_{\rho}$.

If for some $\rho$, $\varphi_{\rho}$ is satisfiable, output the corresponding \gel of $G$.\label{algo-line-find-gel}

Otherwise, report that $G$ is bad.

\tcc{The above two lines are correct by \autoref{claim:2SAT-equivalence}}

\caption{\FPT algorithm for {\sc GEL} parameterized by the size of a star-forest modulator.}
\label{alg:FPT-star-modulator}
\end{algorithm}

The correctness of \autoref{alg:FPT-star-modulator} is justified by the corresponding claims and lemmas. Note that, in \autoref{algo-line-find-gel}, we can indeed output the corresponding \gel of $G$. Indeed, suppose that $\varphi_{\rho}$ is satisfiable for some $\rho$. Then, following the proof of \autoref{claim:2SAT-equivalence}, the satisfying assignment of $\varphi_{\rho}$ allows to construct in polynomial time a good standard labeling relation ${\sf rel}_{\rho}$ of $G$ that extends $\rho$. Given ${\sf rel}_{\rho}$, we can construct in polynomial the corresponding good edge-labeling of $G$ as explained above \autoref{claim:good-labeling-relation} (as used in its proof).

As for the claimed running time, all steps of \autoref{alg:FPT-star-modulator}  run in polynomial time except the guess of $\rho$ in \autoref{algo-line-guessing}. Since by \autoref{eq:few-boring-edges} we have that
$|E(G[X \cup B])| = \Ocal(k^4)$, and guessing $\rho$ is equivalent to guessing a linear ordering of $E(G[X \cup B])$, there are indeed $2^{\Ocal(k^4 \log k)}$ possible guesses.

The proof of \autoref{thm:kernel_sfm} is now complete.

\subparagraph*{Possible simplification of the algorithm.} To conclude this section, we would like to mention that the above algorithm can be substantially simplified by considering the following observation, which we explain reusing  \autoref{fig:01interesting-stars}(b). Let us start by focusing on type-2 leaves, that is, on the red paths between $w$ and $z$, namely $P_{wz}^+$ that increases from $w$ to $z$, and $P_{wz}^-$ that increases from $z$ to $w$. In the discussion about 1-interesting starts below \autoref{fig:01interesting-stars}, we argued that $\lambda_s$ needs to satisfy two inequalities with respect to $\lambda_w^+$ and $\lambda_w^-$, namely that $\lambda_w^+ > \lambda_s$ (cf. \autoref{eq:1interesting-1}) and that $\lambda_w^- < \lambda_s$ (cf. \autoref{eq:1interesting-2}). The simplification is obtained by analyzing the relation between $\lambda_w^+$ and $\lambda_w^-$. Namely, suppose that $\lambda_w^- \geq \lambda_w^+$. In that case, it can be easily verified that the cycle enclosed by the union of $P_{wz}^+$ and $P_{wz}^-$ does not admit two local minima (or maxima), and therefore \autoref{obs:local_minima} implies that the partial edge-labeling that we guessed in $G[X \cup B]$ is {\sl not} good and can be safely discarded. Thus, we can always assume that $\lambda_w^- < \lambda_w^+$, hence if we combine \autoref{eq:1interesting-1} and \autoref{eq:1interesting-2}  into $\lambda_w^- < \lambda_s < \lambda_w^+$, it turns out that this constraint is always feasible. Therefore, in order to choose the label $\lambda_s$, in order to satisfy all the constraints imposed by the different type-2 leaves with a common neighbor $w_i \in X$  (provided that no $w_i$ is adjacent to $z$, as otherwise we would have a $K_{2,3}$), it suffices to verify that the common intersection of the intervals $[\lambda_{w_i}^- , \lambda_{w_i}^+]$ is non-empty, and we can choose  $\lambda_s$ as any number in that intersection.

Let us now focus on the more complicated case, namely on type-1 leaves of the considered star $S$, corresponding to vertices $v$ and $u$ illustrated in \autoref{fig:01interesting-stars}(b). Similarly as above, if we consider the two green paths $P_{uz}^+$ and $P_{uz}^-$ between $u$ and $z$, we can safely assume that $\lambda_u^- < \lambda_u^+$, as otherwise the guessed partial edge-labeling is not good. Then, the disjunctive constraint given by \autoref{eq:1interesting-3} boils down to $\lambda_u^- < \lambda_s < \lambda_u^+$ (that is, there is no disjunction anymore), which gives us a constraint on the choice of  $\lambda_s$ that is always feasible, and we can deal with the joint constraints given by the different type-1 leaves as we did above (namely, verifying that the corresponding intervals have a non-empty intersection).

Summarizing, in \autoref{alg:FPT-star-modulator}, for each of the guessed good edge-labelings of $G[X \cup B]$ (no need to translate them into good standard labeling relations, we can just guess a linear ordering of the edges), we can just verify whether it can be extended to the whole graph by dealing with each star $S$ {\sl separately} as discussed above.

Nevertheless, we decided not to include this simplification in the formal description of the algorithm, as we think that the approach that we presented of using a 2-SAT formulation via labeling relations (in particular, guaranteeing their transitivity) is interesting by itself, and we hope that it will find further applications beyond the parameterization by the size of a star-forest modulator. Indeed, we believe that our approach is robust enough so that it has the potential to be applicable whenever we have disjunctive constraints in order to attribute labels, which is a seemingly non-trivial task.

\section{\FPT algorithms parameterized by treewidth and something else}
\label{sec:FPT}

In this section our goal is to solve the {\sc GEL} problem parameterized by the treewidth of the input graph, denoted by $\tw$, together with an additional parameter, by using dynamic programming (DP) algorithms.
This additional parameter is $c$ in \autoref{sec:DPcolors} (thus, we solve {\sc $c$-GEL}), and the  maximum degree $\Delta$ of the input graph in \autoref{sec:DPdegree}.
Both DP algorithms are quite similar. Hence, we present the one parameterized by $\tw+c$ in full detail, and explain what changes for the one parameterized by $\tw+\Delta$.








\subsection{Treewidth and number of labels}
\label{sec:DPcolors}

When we consider as parameters both the treewidth $\tw$ of the input graph and the number $c$ of distinct labels, the $c$-\GEL problem can easily be expressed in \MSOL (for completeness, such a formulation can be found in the extended full version of the paper available online~\cite{GEL-arXiv}), and then the fact that it is \FPT parameterized by $\tw+c$ follows by Courcelle's theorem~\cite{Courcelle90}, but without an explicit parametric dependence, which is our contribution in this section.

\subparagraph{Dynamic programming algorithm.} In this section, we design a DP to solve the {\sc $c$-GEL} problem. Note that we only do so for $c>1$, since the case $c=1$ trivially reduces to finding a cycle in the input graph.

\begin{theorem}\label{th:tw+c}
For $c>1$, there is an algorithm that solves the {\sc $c$-GEL} problem in time $c^{\Ocal(\tw^2)}\cdot n$ on $n$-vertex graphs of treewidth at most $\tw$.
\end{theorem}

To prove this, let us first give a few more definitions.

\subparagraph{Bad cycles.}
We say that a cycle is \emph{bad} if it has at most one local minimum. Note that, by \autoref{obs:local_minima}, an edge-labeling is good if and only if it induces no bad cycle.

\subparagraph{Local extrema in walks.}
A \emph{walk} in $G$ is a sequence $v_1v_2\dots v_p$ of not necessarily distinct vertices, such that any $v_iv_{i+1}$ is an edge for $i\in[p-1]$.
In particular, paths are walks where all vertices are distinct.
Similarly to cycle, a \emph{local minimum} (resp. \emph{maximum}) in a walk $(v_1,\dots,v_p)$, with respect to $\lambda$, is a subwalk $P=(v_i,\dots,v_j)$ consisting of edges with the same label, such that the edges $v_{i-1}v_i$ and $v_jv_{j+1}$, if they exist, have labels strictly larger (resp. smaller) than that of $P$.
Note that if all labels in a walk are constant, then the walk is both a local minimum and a local maximum.
A \emph{local extremum} is either a local minimum or a local maximum.
A local extremum in the walk is said to be \emph{internal} if it does not contain $v_1v_2$ nor $v_{p-1}v_p$.

\medskip
The idea for the DP is the following.
We fix a nice tree decomposition $(T,\{B_x\mid x\in V(T)\})$ of $G$ and, for each $x\in V(T)$, from the leaves to the root, we want to store all edge-labelings $\lambda:E(G)\to[c]$ such that there is no bad cycle in $G_x$ for $\lambda$.
Remember that $G_x$ is the graph generated by set of all vertices in a bag that is a descendant of $x$ in $T$ and are not in $B_x$. 

Of course, there are too many such labelings, so we only store signatures of such labelings with only the essential information.
Hence, for $t\in V(T)$ and for a labeling $\lambda$, we only care about:
\begin{itemize}
    \item the restriction of $\lambda$ to $E(G[B_x])$, and
    \item the paths between vertices of $B_x$, and whose internal vertices are in $G_x$, that may potentially be contained in a bad cycle of $G$, with respect to $\lambda$.
\end{itemize}
Again, there are too many such paths, so we will assign a type to each path and keep track of each type of paths present for each pair $u,v$ of vertices of $B_x$.

\subparagraph{Types of paths.}
Given that a bad cycle has at most one local minimum and one local maximum, a subpath $P$ of a bad cycle with endpoints $u$ to $v$ can be of five different types depending on its internal local extrema. Note that the endpoints always belong to a local extremum, but that it is not internal, and thus not taken into account.
The possible types of $P$ are the following (see \autoref{fig:label-type} for an illustration):

\begin{itemize}
    \item {\bf Type $\emptyset$:} It has no internal local extremum, or
    \item {\bf Type {\sf m}:} It has one internal local minimum and no internal local maximum, or
    \item {\bf Type {\sf M}:} It has one internal local maximum and no internal local minimum, or
    \item {\bf Type {\sf mM}:} From $u$ to $v$, it has one internal local minimum and one internal local maximum, in this order, or
    \item {\bf Type {\sf Mm}:} From $u$ to $v$, it has one internal local maximum and one internal local minimum, in  this order.
\end{itemize}
\begin{figure}
    \centering
    \includegraphics{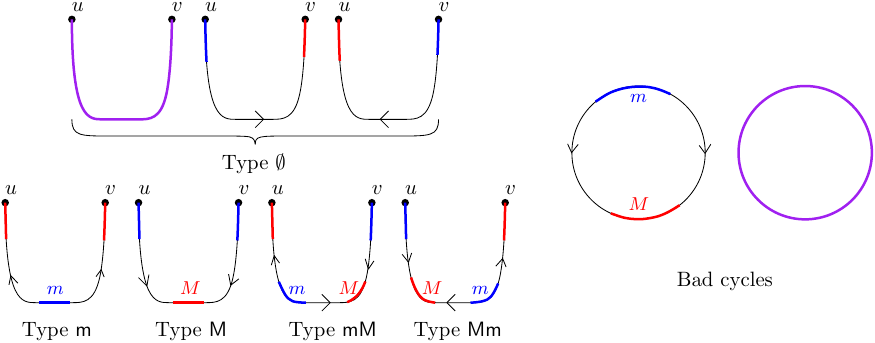}
    \caption{Types of paths that may be subpaths of a bad cycle. Local minima are represented in blue, and local maxima in red. The purple represent a subpath that is both a local minimum and a local maximum (which means that all labels are equal). Internal local minima (resp. maxima) are labeled $m$ (resp. M).
    An arrow between a local minima and a local maxima represents the existence of an increasing path between them.
    }
    \label{fig:label-type}
\end{figure}
More generally, for any walk $P$ in $G$ from a vertex $u$ to a vertex $v$, which may not be distinct, 
we classify it to have one of the types described above if it corresponds to the definition, and the {\bf type {\sf good}} otherwise, meaning that it has at least two internal local extrema that are either both minima or maxima, and thus cannot be a subpaths of a bad cycle.

\subparagraph{Label-types of paths.}
More precisely, we will consider here label-types.
The \emph{label-type} of a path $P$ in $G$ from a vertex $u$ to a vertex $v$ with respect to $\lambda$ is $l_1\tau l_2$, where $\tau\in\{\emptyset,{\sf m, M, mM, Mm, good}\}$ is the type of $P$ and $l_1$ (resp. $l_2$) is the label of the edge incident to $u$ (resp. $v$).

\subparagraph{Partial-order of types.}
We define a partial-order $\le$ on types of paths as follows.
Given two types $\tau,\tau'\in\{\emptyset,{\sf m, M, mM, Mm, good}\}$, we say that $\tau\le\tau'$ if and only if $\tau$ is a substring of $\tau'$ or $\tau'={\sf good}$.
In other words, $\emptyset\le\{{\sf m,M}\}\le\{{\sf mM,Mm}\}\le{\sf good}$.
This partial order is extended to label-types by saying that $l_1\tau l_2\le l_1'\tau l_2'$ if and only if $l_i=l_i'$ for $i\in[2]$ and $\tau\le\tau'$.

We do so because we can observe the following.

\begin{observation}\label{obs:type}
Let $G$ be a graph, $\lambda$ be an edge-labeling of $G$, $C$ be a bad cycle in $G$ with respect to $\lambda$, and $P$ be a subpath of $C$ from a vertex $u$ to a vertex $v$ of label-type $l_1\tau l_2$.
Then, for any type $\tau'\le\tau$, and for any path $P'$ from $u$ to $v$ of label-type $l_1\tau' l_2$ that is internally vertex-disjoint from $C$, the cycle $C'$ obtained from $C$ by replacing $P$ with $P'$ is also bad.
\end{observation}

Hence, for $x\in V(T)$ and for any $u,v\in B_x$, what we will store instead of every $(u,v)$-path that may be a subpath of a bad cycle, is every label-type of a $(u,v)$-path, or precisely, the minimal such label-types, since by \autoref{obs:type}, they are enough to detect bad cycles.

\subparagraph{Signature of a labeling.}
Let $\lambda$ be an edge-labeling of $G$ and $x$ be a node of $T$.
Let $\lambda_x:E(G[B_x])\to[c]$ be the restriction of $\lambda$ to $B_x$.
For $(u,v)\in B_x^2$ with $u,v$ distinct (resp. $u=v$), let $F_x^{u,v}$ be the set of label-types of the $(u,v)$-paths whose internal vertices are in $G_x$ (resp. cycles containing $v$ and whose other vertices are in $G_x$).
We stress that the paths that we consider do not contain any edge of $E(G[B_x])$, which simplifies the DP.
Let $f_x$ be the function that maps each pair $(u,v)\in B_x^2$ to the minimal elements of $F_x^{u,v}$.
Then, the \emph{signature} of $\lambda$ at the node $x$, denoted by $\sig_\lambda^x$, is the pair $(\lambda_x,f_x)$.

From the partial-order on types, we can see that, for any $\Tcal\subseteq \{\emptyset,{\sf m, M, mM, Mm, good}\}$, there are at most two minimal elements in $\Tcal$, and that a set of minimal elements must be one of the following: either $\{\tau\}$ for $\tau\in\{\emptyset,{\sf m, M, mM, Mm, good}\}$, or $\{\sf m, M\}$, or $\{\sf mM, Mm\}$. There are therefore eight possible sets of minimal elements (six singletons, and two sets of size two), and since we have at most $c^2$ choices for the labels $l_1,l_2$ of a label-type $l_1 \tau l_2$, for each pair $(u,v)$ the function $f_x$ can take up to $8c^2$ values.
Thus, the number of distinct signatures is at most $c^{\tw^2}\cdot(8c^2)^{\tw^2}$.

\subparagraph{Signature of a bag.}
Now, we define the \emph{signature} of a vertex $x$ of $T$, denoted by $\sig_x$ to be the set of all signatures $\sig_\lambda^x$ where $\lambda:E(G)\to[c]$ is such that there is no bad cycle in $G_x$.
{Given that we assume the root-bag $B_r$ to be empty, the only possible signature in $\sig_r$ is the empty signature, that we denote by $(\emptyset,\emptyset)$.
We have $(\emptyset,\emptyset)\in\sig_r$ (which is equivalent to saying that $\sig_r\neq\emptyset$) if and only if there is a labeling that induces no bad cycle in $G_r=G$, i.e., if and only if $G$ admits a $c$-\gel.
}

\smallskip
Let us show how to construct $\sig_x$ inductively from its children, if any.

\subsubsection{Leaf bag}
If $B_x$ is a leaf bag, then we trivially have $$\sig_x=\{(\emptyset,\emptyset)\}.$$

\subsubsection{Introduce bag}
If $B_x$ is an introduce bag, then there is an introduced vertex $v\in V(G)$ such that $B_x=B_y\cup\{v\}$, where $y$ is the unique child of $x$ in $T$.
Let $\Lambda$ be the set of at most $c^{tw-1}$ distinct labelings of $E_v$, where $E_v$ is the set of edges incident to $v$. Note that  $E_v$ is contained in $E(G[B_x])$. 
Let us prove that $$\sig_x=\{(\lambda_y\cup\ell,f_y')\mid (\lambda_y,f_y)\in\sig_y,\ell\in \Lambda\},$$
where $\lambda_y\cup\ell$ is the labeling of $G[B_x]$ that is equal to $\lambda_y$ restricted to $G[B_y]$ and to $\ell$ restricted to $E_v$,
and $f_y'$ is equal to $f_y$ on $B_y^2$ and equal to  $\emptyset$ on all $(v,u)$ or $(u,v)$ with $u\in B_x$.
Note that $I:=\{(\lambda_y\cup\ell,f'_y)\mid (\lambda_y,f_y)\in\sig_y,\ell\in \Lambda\}$ can be constructed from $\sig_y$ in time $\Ocal(c^{\tw-1}\cdot|\sig_y|)$. In the next two claims we prove that $\sig_x = I$.

\begin{claim}
$\sig_x\subseteq I$.
\end{claim}
\begin{proof}
Let $(\lambda_x,f_x)\in\sig_x$.
Then there is a labeling $\lambda$ such that $\sig_\lambda^x=(\lambda_x,f_x)$ and such that there is no bad cycle in $G_x$ with respect to $\lambda$.
Thus, $\sig_\lambda^y=(\lambda_y,f_y)\in\sig_y$.
We have $E(G[B_x])=E(G[B_y])\cup E_v$.
Therefore, $\lambda_x=\lambda_y\cup\ell$, where $\ell\in\Lambda$ is the restriction of $\lambda$ to $E_v$.
Additionally, $N_G(v)\cap V_x=\emptyset$ by the definition of tree decomposition and of $G_x$, so, for any $u\in B_x$, for any $(u,v)$-path or $(v,u)$-path $P$ of $G$, $P$ cannot have all its internal vertices in $G_x$.
Therefore, any path between vertices of $B_x$ whose internal vertices are in $G_x$ is also a path between vertices of $B_y$ whose internal vertices are in $G_y$.
Therefore, $f_x=f_y'$, and thus $(\lambda_x,f_x)\in I$.

\end{proof}

\begin{claim}
$\sig_x\supseteq I$.
\end{claim}
\begin{proof}
Let $(\lambda_y\cup\ell,f_y')\in I$.
Given that $(\lambda_y,f_y)\in\sig_y$, there is a labeling $\lambda'$ such that $\sig_{\lambda'}^y=(\lambda_y,f_y)$ and such that there is no bad cycle in $G_y$ with respect to $\lambda'$.
Given that $v\notin V(G_x)$, we get that
$G_x=G_y$ and there is also no bad cycle in $G_x$ with respect to $\lambda'$; so $\sig_{\lambda'}^x\in\sig_x$.
But then, given that cycles in $G_x$ do not use any edge of $B_x$, the same can be said of any other labeling whose restriction to $G_x$ is the same as $\lambda'$.
This is in particular the case of the labeling $\lambda$ obtained from $\lambda'$ by replacing the labels in $E_v$ with the ones of $\ell$.
Hence, $\sig_\lambda^x\in\sig_x$.
Additionally, any path between vertices of $B_x$ whose internal vertices are in $G_x$ is also a path between vertices of $B_y$ whose internal vertices are in $G_y$, so $\sig_\lambda^x=(\lambda_y\cup\ell,f_y')$.
Therefore, $I\subseteq\sig_x$.
\end{proof}

\subsubsection{Join bag}
If $B_x$ is a join bag, then $B_x=B_{y_1}=B_{y_2}$, where $y_1$ and $y_2$ are the unique children of $x$ in $T$.
Let us show that $$\sig_x=\{(\lambda_y,\min(f_{y_1}\cup f_{y_2}))\mid \forall i\in[2], (\lambda_y,f_{y_i})\in\sig_{y_i}\},$$
where $\min(f_{y_1}\cup f_{y_2})$ is the function that maps $(u,v)$ to the minimum label-types of the union $f_{y_1}(u,v)\cup f_{y_2}(u,v)$.
Note that $J:=\{(\lambda_y,\min(f_{y_1}\cup f_{y_2}))\mid \forall i\in[2], (\lambda_y,f_{y_i})\in\sig_{y_i}\}$ can be constructed in time $\Ocal(c^2\tw^2\cdot|\sig_{y_1}|\cdot|\sig_{y_2}|)$: for each $i\in[2]$, for each $(\lambda_{y_i},f_{y_i})\in\sig_{y_i}$, for each $(u,v)\in B_x^2$, for each $(l_1,l_2)\in[c]^2$, we take the minimum over the label-types in $f_{y_1}(u,v)$ and $f_{y_2}(u,v)$. In the next two claims we prove that $\sig_x = J$.

\begin{claim}
$\sig_x\subseteq J$.
\end{claim}
\begin{proof}
Let $(\lambda_x,f_x)\in\sig_x$.
Then there is a labeling $\lambda$ such that $\sig_\lambda^x=(\lambda_x,f_x)$ and such that there is no bad cycle in $G_x$ with respect to $\lambda$.
Thus, for $i\in[2]$, $\sig_\lambda^{y_i}=(\lambda_{y_i},f_{y_i})\in\sig_{y_i}$.
Given that $E(G[B_x])=E(G[B_{y_1}])=E(G[B_{y_2}])$, it implies that $\lambda_x=\lambda_{y_1}=\lambda_{y_2}$.
Additionally, $G_x$ is the disjoint union of $G_{y_1}$ and $G_{y_2}$, which implies that any path between vertices of $B_x$ whose internal vertices are in $G_x$ is a path between vertices of $B_{y_i}$ whose internal vertices are in $G_{y_i}$ for some $i\in[2]$.
However, a label-type in $f_{y_1}(u,v)$ may be strictly smaller than a label-type in $f_{y_1}(u,v)$, or vice-versa, so we take the label-types in $f_x(u,v)$ to be the minimum label-types in the union of $f_{y_1}(u,v)$ and $f_{y_2}(u,v)$.
Therefore, $f_x=\min(f_{y_1}\cup f_{y_2})$.
Hence, $\sig_x\subseteq J$.
\end{proof}

\begin{claim}
$J\subseteq \sig_x$.
\end{claim}
\begin{proof}
Let $(\lambda_y,\min(f_{y_1}\cup f_{y_2}))\in J$.
Then, for $i\in[2]$, there is a labeling $\lambda_i$ of $G$ such that $\sig_{\lambda_i}^{y_i}=(\lambda_y,f_{y_i})\in\sig_{y_i}$.
Let $\lambda$ be a labeling of $G$ that is equal to $\lambda_y$ on $E(G[B_x])$, equal to $\lambda_{y_1}$ on $E(V_{y_1}\cup B_{y_1},V_{y_1})$, and equal to $\lambda_{y_2}$ on $E(V_{y_2}\cup B_{y_2},V_{y_2})$.
We can do so because $E(G[B_x])$, $E(V_{y_1}\cup B_{y_1},V_{y_1})$ and $E(V_{y_2}\cup B_{y_2},V_{y_2})$ are disjoint.
Any path between vertices of $B_x$ whose internal vertices are in $G_x$ is a path between vertices of $B_{y_i}$ whose internal vertices are in $G_{y_i}$ for some $i\in[2]$.
Therefore, $\sig_\lambda^x=(\lambda_y,\min(f_{y_1}\cup f_{y_2}))$.
Given that $G_x$ is the disjoint union of $G_{y_1}$ and $G_{y_2}$
and that there is no bad cycle in $G_{y_i}$ with respect to $\lambda_i$, it implies that there is no bad cycle in $G_x$ with respect to $\lambda$.
Hence, $\sig_x\supseteq J$.
Therefore, $\sig_\lambda^x\in\sig_x$.
\end{proof}

\subsubsection{Forget bag}

If $B_x$ is a forget bag, then there is a forgotten vertex $v\in V(G)$ such that $B_x=B_y\setminus\{v\}$, where $y$ is the unique child of $x$ in $T$.
This case is the hardest one given that $G_x=G[V_y\cup\{v\}]$.
Indeed, since $G_x$ is obtained from $G_y$ by adding vertex $v$ (recall our definition of $G_x$ below the definition of nice tree decomposition),
there might be new cycles in $G_x$ (containing $v$), so we need to remove signatures representing labelings that contain bad cycles. 
Additionally, for any $(u,w)\in B_x^2$, we need to add the label-type of $(u,w)$-paths
containing $v$.
We define $\sig_x$ from $\sig_y$ as follows.

\subparagraph{Restriction to $B_x$.}
We have $E(G[B_x])=E(G[B_y])\setminus E_v$, where $E_v$ is the set of edges incident to $v$.
Therefore, for any labeling $\lambda$ such that $\sig_\lambda^x=(\lambda_x,f_x)\in\sig_x$, we have $\sig_\lambda^y=(\lambda_y,f_y)\in\sig_y$, where $\lambda_x$ is obtained 
by restricting $\lambda_y$ to $E(G[B_y])\setminus E_v$, which we write as $\lambda_y|_{E(G[B_y])\setminus E_v}$.

Hence, we observe the following.
\begin{observation}\label{obs:F1}
For any labeling $\lambda$ 
with signature $(\lambda_x,f_x)$ at $x$ and $(\lambda_y,f_y)$ at $y$, we have $\lambda_x=\lambda_y|_{E(G[B_y])\setminus E_v}$.
\end{observation}

\subparagraph{Bad cycles.}
We have $G_y=G_x\setminus\{v\}$.
Therefore, any cycle $C$ of $G_x$
is either a cycle of $G_y$, or a cycle containing $v$.
Let $\lambda$ be a labeling whose signature $(\lambda_y,f_y)$ at node $y$ is in $\sig_y$.
We get that $\lambda$ may induce a bad cycle in $G_x$, but this bad cycle $C$ must contain $v$.
Fortunately, this bad cycle can be detected using the label-type stored in $f_y(v,v)$.
Indeed, there is a bad cycle containing $v$ and vertices of $G_y$ if and only if there is a $(v,v)$-path whose internal vertices are in $G_y$ and whose label-type is one of the following:
\begin{itemize}
    \item $l_1\tau l_2$ with $\tau\in\{\emptyset,{\sf m, M}\}$, or
    \item $l_1{\sf mM}l_2$ with $l_1\le l_2$, or
    \item $l_1{\sf Mm}l_2$ with $l_1\ge l_2$.
\end{itemize}

\begin{figure}
    \centering
    \includegraphics{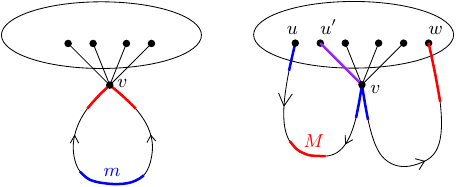}
    \caption{Here, $v$ is a forgotten vertex. Left: Example of a bad cycle. Right: the $(u,w)$-path obtained from the concatenation of a $(u,v)$-path of type {\sf M} and a $(v,w)$-path of type $\emptyset$, whose label on the edge incident to $v$ strictly smaller than its label on the edge incident to $w$, is of type {\sf good}. The $(u',v)$-path is either of type $\emptyset$ or {\sf m} depending on whether the label of $u'v$ is smaller or strictly bigger than the label of the other edge incident to $v$.}
    \label{fig:forget}
\end{figure}
See \autoref{fig:forget} for an example of a bad cycle.
Note that $l_1{\sf mM}l_2$ with $l_1> l_2$ does not give a bad cycle, because then $v$ belongs to a local minimum of the cycle, and there are hence two of them.
Same for $l_1{\sf Mm}l_2$ with $l_1 < l_2$.
For the other cases presented above however, there are at most one local minimum and one local maximum in the cycle, which is thus bad.
Therefore, by \autoref{obs:type}, $\lambda$ induces a bad cycle in $G_x$ if and only if one of the minimal label-type stored in $f_y(v,v)$ is one of those described above, which can be checked in time $\Ocal(c^2)$.
Let $\sig_y'$ be the set  obtained from $\sig_y$ by removing all signatures such that $G_x$ contains a bad cycle.

Therefore, we observe the following.
\begin{observation}\label{obs:F2}
For any labeling $\lambda$, if $\sig_\lambda^y\in\sig_y$,
then $\lambda$ induces no bad cycle in $G_x$ if and only if $\sig_\lambda^y\in\sig_y'$.
\end{observation}

\subparagraph{New paths.}
Let $\lambda$ be a labeling whose signature at $y$ is $(\lambda_y,f_y)$ and whose signature at $x$ is $(\lambda_x,f_x)$.
Let $u,w\in B_x$ and $P$ be a $(u,w)$-path whose internal vertices are in $G_x$. 
Then either
\begin{itemize}
    \item $P$ is a path whose internal vertices are in $G_y$, or
    \item $v$ is a vertex of $P$.
\end{itemize}
In the second case, $P$ is composed of a $(u,v)$-path $P_u$ and a $(v,w)$-path $P_w$ such that $P_u$ (resp. $P_w$) is either the edge $uv$ (resp. $vw$), or a path whose internal vertices are in $G_y$.
If $P_u$ (resp. $P_w$) is an edge, then it has label-type $\ell\emptyset\ell$, where $\ell=\lambda(uv)$ (resp. $\ell=\lambda(vw)$).
When $v\in V(P)$, observe that the label-type of $P$ can easily be deduced from the label-types $l_1\tau l_2$ of $P_u$ and $l_1'\tau' l_2'$ of $P_w$.
For instance, if $\tau={\sf m}$ and $\tau'={\sf M}$, then the label-type of $P$ is $l_1{\sf mM}l_2'$ if $l_2\le l_1'$, and $l_1{\sf good}l_2'$ if $l_2> l_1'$ because then the edge incident to $v$ in $P_u$ (resp. $P_w$) belongs to a local maxima (resp. minima).
We say that the label-type $l_1\tau^*l_2'$ of $P$ is the \emph{concatenation} of $l_1\tau l_2$ and $l_1'\tau' l_2'$.
See \autoref{fig:forget} for other examples of concatenations of label-types.
Therefore, the signature of $\lambda$ at $y$ is enough to compute the label-type of $P$.

This gives us an idea of how to construct $f_x$ from $f_y$.
Let $L_{u,w}$ be the set of label-types that are either:
\begin{itemize}
    \item in $f_y(u,w)$
    \item or obtained from the concatenation of
\begin{itemize}
    \item a label-type that is either in $f_y(u,v)$, or is the label-type of $uv$, if it is an edge, and
    \item a label-type that is either in $f_y(v,w)$, or is the label-type of $vw$, if it is an edge.
\end{itemize}
\end{itemize}
Note that $L_{u,w}$ can be computed in time $\Ocal(c^4)$.
Then we define $f_y^2:B_x^2\to[c]^2\times[8]$ to be the function that, to any pair $u,w\in B_x$, maps the minimal elements of $L_{u,w}$.

Let us prove the following.
\begin{claim}\label{obs:F3}
For any labeling $\lambda$ with $sig_\lambda^x=(\lambda_x,f_x)$ and $sig_\lambda^y=(\lambda_y,f_y)$, we have $f_x=f_y^2$.
\end{claim}
\begin{proof}
Let $u,w\in B_x$.

Let us first prove that $f_x(u,w)\subseteq L_{u,w}$.
Let $\sigma\in f_x(u,w)$ be a label-type.
Then there is a $(u,w)$-path $P$ whose internal vertices are in $G_x$ and whose label-type is $\sigma$.
\emph{If the internal vertices of $P$ are in $G_y$}, then we claim that $\sigma\in f_y(u,w)$.
Indeed, if that is not the case, then there is a $(u,w)$-path $P'$ whose internal vertices are in $G_y$ and whose label-type is $\sigma'<\sigma$.
But then, given that $G_y\subseteq G_x$, $\sigma'\in f_x(u,w)$, a contradiction.
So $\sigma\in f_y(u,w)\subseteq L_{u,v}$.
\emph{Otherwise, $v$ is a vertex of $P$.}
Thus, $P$ is composed of a $(u,v)$-path $P_u$ and a $(v,w)$-path $P_w$ such that $P_u$ (resp. $P_w$) is either the edge $uv$ (resp. $vw$), or a path whose internal vertices are in $G_y$.
If $P_u$ is not an edge, let us show that the label-type $\sigma_u$ of $P_u$ is in $f_y(u,v)$.
Assume that is not the case.
Then there is a $(u,v)$-path $P_u'$ whose internal vertices are in $G_y$ and whose label-type is $\sigma_u'<\sigma_u$.
The concatenation of $P_u'$ and $P_w$ is a walk $W$ from $u$ to $w$ (since their internal vertices may intersect) whose label-type $\sigma'$ is strictly smaller than $\sigma$.
But then, there is a subwalk of $W$ that is a path $P'$ from $u$ to $w$ whose label-type $\sigma''$ is smaller than $\sigma'$.
Therefore, $\sigma''<\sigma$, contradicting the minimality of $\sigma$.
Similarly, if $P_w$ is not an edge, then the label-type $\sigma_w$ of $P_w$ is in $f_y(v,w)$.
So $\sigma\in L_{u,v}$.

Let us now prove that, for each $\sigma\in f_y^2(u,w)$, there is a $(u,w)$-path $P$ whose internal vertices are in $G_x$ and whose label-type is $\sigma$.
\emph{If $\sigma\in f_y(u,w)$}, then there is a $(u,w)$-path $P$ whose internal vertices are in $G_y$ and whose label-type is $\sigma$.
Given that $G_y\subseteq G_x$, we thus have what we want.
\emph{Otherwise}, $\sigma$ is the concatenation of $\sigma_u$ and $\sigma_w$, where $\sigma_u$ (resp. $\sigma_w$) is either the label-type of the edge $uv$ (resp. $vw$) or belongs to $f_y(u,v)$ (resp. $f_y(v,w)$).
Thus, there is $P_u$ (resp. $P_w$) that is a $(u,v)$-path (resp. $(v,w)$-path) whose label-type is $\sigma_u$ (resp. $\sigma_w$), and that is either the edge $uv$ (resp. $vw$), or whose internal vertices are in $G_y$.
\emph{If $P_u$ and $P_w$ are internally vertex-disjoint}, then the concatenation of $P_u$ and $P_w$ is a $(u,w)$-path $P$ whose label-type is $\sigma$.
\emph{If $P_u$ and $P_w$ are not internally vertex-disjoint however},
which may happen only when their internal vertices are in $G_y$, then
their concatenation is a walk $W$ from $u$ to $w$.
But then, there is a subwalk of $W$ that is a $(u,w)$-path $P$ which does not contain $v$, and thus whose internal vertices are in $G_y$.
In particular, the label-type of $P$ is $\sigma'\le\sigma$.
Then there is a label $\sigma''\le\sigma'$ that belongs to $f_y(u,w)\subseteq L(u,w)$.
But then, given that $\sigma''\le\sigma$ and that $\sigma,\sigma''\in L(u,w)$, we conclude by minimality of $\sigma$ that $\sigma''=\sigma'=\sigma$.
Thus, $P$ is a $(u,w)$-path whose internal vertices are in $G_y\subseteq G_x$ and whose label-type is $\sigma$.

We now prove that $f_x(u,w)\subseteq f_y^2(u,w)$.
Let $\sigma\in f_x(u,w)\subseteq L(u,v)$.
There is $\sigma'\le \sigma$ such that $\sigma'\in f_y^2(u,w)$.
Thus, there is a $(u,w)$-path $P$ whose internal vertices are in $G_x$ and whose label-type is $\sigma$.
Hence, by minimality of $\sigma$, $\sigma=\sigma'\in f_y^2(u,w)$.

We finally prove that $f_x(u,w)\supseteq f_y^2(u,w)$.
Let $\sigma\in f_y^2(u,w)$.
Thus, there is a $(u,w)$-path $P$ whose internal vertices are in $G_x$ and whose label-type is $\sigma$.
There is $\sigma'\le \sigma$ such that $\sigma'\in f_x(u,w)$.
Hence, by minimality of $\sigma$, $\sigma=\sigma'\in f_x(u,w)$.
\end{proof}

From \autoref{obs:F1}, \autoref{obs:F2}, and \autoref{obs:F3}, we immediately deduce that
$$\sig_x=\{(\lambda_y|_{E(G[B_y])\setminus E_v},f_y^2)\mid (\lambda_y,f_y)\in\sig'_y\}.$$
Note that $\{(\lambda_y|_{E(G[B_y])\setminus E_v},f_y^2)\mid (\lambda_y,f_y)\in\sig'_y\}$ can be obtained in time $\Ocal(\tw^2\cdot c^4\cdot|\sig_y|)$.

\subsubsection{Complexity}

As said above, the number of distinct signatures is at most $c^{\tw^2}\cdot(8c^2)^{\tw^2}$.
Given that we can assume the nice tree decomposition to have $\Ocal(n)$ nodes, the operations on a leaf bag, introduce bag, join bag, and forget bag take respectively time $\Ocal(1)$, $\Ocal(c^{\tw-1}\cdot|\sig_y|)$, $\Ocal(c^2\tw^2\cdot|\sig_{y_1}|\cdot|\sig_{y_2}|)$, and $\Ocal(\tw^2\cdot c^4\cdot|\sig_y|)$, we conclude that the DP solves {\sc $c$-GEL} on $G$ in time $c^{\Ocal(\tw^2)}\cdot n$.

Note that, using standard backtracking techniques, a $c$-\gel can be computed within the same running time.

\subsection{Treewidth and maximum degree}
\label{sec:DPdegree}

 Let us first note that, in contrast to \autoref{sec:DPcolors}, it is not evident how to express the \GEL problem in \MSOL when parameterized by treewidth and the maximum degree (instead of the number of distinct labels). We define {\sc Minimum GEL} as the problem of, given an input graph $G$, finding the minimum $c\in\mathbb{N}$ such that $G$ is $c$-good. We now wish to design a DP to prove the following.

\begin{theorem}\label{th:tw+D}
There is an algorithm that solves {\sc Minimum GEL} in time $2^{\Ocal(\tw\Delta^2+\tw^2\log\Delta)}\cdot n$ on $n$-vertex graphs of maximum degree $\Delta$ and of treewidth at most $\tw$.
\end{theorem}

The idea is actually very similar to the one 
for \autoref{th:tw+c}.
We fix a nice tree decomposition $(T,\{B_x\mid x\in V(T)\}$ of $G$.
For each vertex $x$ of $T$, we store the signature of the labelings $\lambda$ of $G$ that do not induce a bad cycle in $G_x$.
However, contrary to before, we do not have a bound on the number of values taken by $\lambda$.
Therefore, we cannot cannot use the same signature as in \autoref{th:tw+c}.

To get around this problem, instead of labelings, we shall consider a partial orientation of the line graph of $G$ (see \autoref{sec:line_graph}).
Indeed, we prove that a graph admits a \gel if and only if its line graph admits a "good" partial orientation (\autoref{claim:good-orientation}).
The signature of a partial orientation at a node $x$ is:
\begin{itemize}
    \item the partial orientation of the line graph induced $B_x$ and its neighbours in $G_x$, and
    \item the minimal "degree-type" of $(u,v)$-paths for $(u,v)\in B_x^2$ whose internal vertices are in $G_x$.
\end{itemize}
Here, the degree-type is similar to the label-type, but instead of storing the label $l_1$ and $l_2$ of the extremal edges $e_1$ and $e_2$, we store $e_1$ and $e_2$ directly.

Then, the signature of the bag of the node $x$ is essentially the set of signature of partial orientations whose restriction to $G_x$ is good.

\subsubsection{Line graph}\label{sec:line_graph}

\subparagraph{Line graphs.}
Let $G$ be a graph.
The \emph{line graph} of $G$ is the graph $L$ with vertex set $E(G)$ and such that, for $e,f\in E(G)$, $ef\in E(L)$ if and only if $e$ and $f$ have a common endpoint.
Note that the line graph of a cycle is a cycle 
of same length.

\subparagraph{Partial orientations.}
Let $G$ be a graph.
A \emph{partial orientation} $O$ of $G$ is a set of couples of vertices $(u,v)$ of $G$ such that $uv\in E(G)$ and $(v,u)\notin O$.
For $uv\in E(G)$ such that $(u,v),(v,u)\notin O$, we say that the edge $uv$ is \emph{not oriented}.
Otherwise, we say that $uv$ is \emph{oriented}.
A \emph{source} (resp. \emph{sink}) in $(G,O)$ is an induced subgraph $H$ of $G$ whose edges are not oriented and such that, for each edge $uv\in E(G)$ such that $u\in V(H)$ and $v\notin V(H)$, $(u,v)\in O$ (resp. $(v,u)\in O$).
A \emph{partial dicycle} in $(G,O)$ is a cycle $C$ of $G$ such that $(C,O')$ has no source nor sink, where $O'$ is the restriction of $O$ to $C$, that is, it has at least one oriented edge, and all the oriented edges go in the same direction according to a cyclic ordering of its vertices. Note that a partial dicycle may have edges that are not oriented, but that a cycle with no oriented edges is not a partial dicycle.
\begin{figure}
    \centering
    \includegraphics{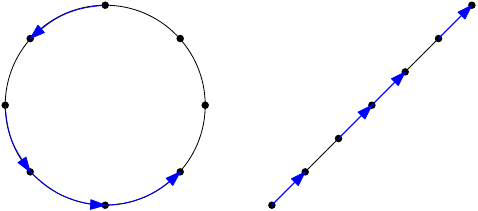}
    \caption{A partial dicycle and a partial dipath of order 5.}
    \label{fig:di}
\end{figure}
A partial orientation $O$ of $G$ is said to be \emph{transitive} if $(G,O)$ has no partial dicycle. We use this term by analogy with the transitivity of labeling relations.
A \emph{partial  dipath} in $(G,O)$ is a path $P=p_1,\dots,p_t$ 
such that, for $i\in[t-1]$, either the edge $p_ip_{i+1}$ is not oriented, or $(p_i,p_{i+1})\in O$.
The \emph{order} of $P$ is the number of oriented edges plus one.
See \autoref{fig:di} for an illustration of a partial dicycle and a partial dipath.

\subparagraph{Good partial orientation of a line graph.}
Let $G$ be a graph and $L$ be its line graph.
Let $O$ be a partial orientation of $L$.
$O$ is said to be \emph{good} if it is transitive and, for every cycle $C$ of $G$, the corresponding cycle in $L$ has at least two sources, or at least two sinks.

\begin{observation}
If $(L,O)$ admits a partial dicycle, then either there is a cycle of $G$ whose edges induce a partial dicycle of $(L,O)$, or there is a vertex $v$ of $G$ and three edges adjacent to $v$ that induce a partial dicycle of $(L,O)$.
\end{observation}

Therefore, $O$ is good if and only if:
\begin{itemize}
\item for any cycle $C$ of $G$, the corresponding cycle in $L$ has at least two sources and
\item for any vertex $v$ of $G$, there are no three edges adjacent to $v$ that induce a partial dicycle of $(L,O)$.
\end{itemize}

\subparagraph{From labelings to orientations.}
Let $G$ be a graph, $\lambda$ be a \gel of $G$, and $L$ be the line graph of $G$.
We set $O_\lambda$ to be the partial orientation of $L$ such that $(e,f)\in O_\lambda$ if and only if
$ef\in E(L)$ and $\lambda(e)<\lambda(f)$.

\begin{lemma}\label{claim:l_to_o}
Let $G$ be a graph, $\lambda$ be a $c$-\gel of $G$, and $L$ be the line graph of $G$.
Then $O_\lambda$ is good and the partial dipaths of $(L,O_\lambda)$ have order at most $c$.
\end{lemma}
\begin{proof}
Let $C$ be a partial dicycle in $L$ with vertices $c_1,\dots,c_p$.
Given that a partial dicycle has at least one oriented edge, we can assume without loss of generality that $(c_1,c_2)\in O_\lambda$.
Then, by definition, $\lambda(c_1)<\lambda(c_2)\le\dots\le\lambda(c_p)\le\lambda(c_1)$.
This is not possible, so $O_\lambda$ is transitive.

Let $C$ be a cycle of $G$ and $C'$ be the corresponding cycle in $L$.
As $\lambda$ is a $c$-\gel, by \autoref{obs:local_minima}, $C$ has at least two local minima.
By definition of $O_\lambda$, the edges of a local minimum of $C$ are vertices of a source of $C'$.
Therefore, $C'$ has two sources.
So $O_\lambda$ is good.

If a partial dipath $P$ of $(L,O)$ with vertex set $p_1,\dots,p_t$ has order $d$, then there are indices $i_1<\dots< i_d$ such that $\lambda(p_{i_1})<\dots<\lambda(p_d)$.
Given that $\lambda$ takes its values in $[c]$, we conclude that the partial dipaths of $(L,O_\lambda)$ have order at most $c$.
\end{proof}

\subparagraph{From orientations to labelings.}
Let $G$ be a graph, $L$ be the line graph of $G$, and $O$ be a transitive partial orientation of $L$.
Let $c$ be the maximum order of a partial dipath in $(L,O)$.
We set $\lambda_O:E(G)\to[c]$ to be the edge-labeling of $G$ obtained as follows.
Let $L_0:=L$ and $O_0:=O$.
For $i\in[c]$ in an increasing order, we do the following.
Let $E_i$ be the set of $e\in V(L_{i-1})$ that belong to a source of $(L_{i-1},O_{i-1})$.
For $e\in E_i$, we set $\lambda_O(e)=i$.
Then, we set $L_i:=L_{i-1}\setminus E_i$ and $O_i$ to be the restriction of $O_{i-1}$ to $L_i$.
Given that a maximal partial dipath goes from a source to a sink, the maximum order of a path in $(L_i,O_i)$ is $c-i$.
Hence, $L_{c-1}$ has partial dipaths of order at most one, that is, no oriented edges, and thus $L_c$ is the empty graph.
Therefore, $\lambda_O$ is indeed a labeling of $G$ taking its values in $[c]$.

\begin{lemma}\label{claim:l_to_o}
Let $G$ be a graph, $L$ be the line graph of $G$, and $O$ be a good partial orientation of $L$.
Then, $\lambda_O$ is \gel of $G$.
Moreover, if $c$ is the maximum order of a partial dipath in $(L,O)$, then $\lambda_O$ is $c$-\gel.
\end{lemma}
\begin{proof}
Let $C$ be a cycle in $G$.
Let $C'$ be the corresponding cycle in $L$.
$C'$ has at least two sources with respect to $O$, which equivalently means that $C$ has at least two local minima, where the vertices of a source in $C'$ are the edges of a local minimum in $C$.
\end{proof}

%
%

\begin{corollary}\label{claim:good-orientation}
A graph $G$ admits a $c$-\gel if and only if its line graph admits a good partial orientation whose partial dipaths have order at most $c$.
\end{corollary}

\subsubsection{Dynamic program}

The DP is now basically the same as in \autoref{sec:DPcolors} but here we exchange labelings with partial orientation of the line graph, since they are equivalent by, and label-types with degree-types, that is, we remember the first and last edges of paths instead of their label.
We can afford to do so given that the degree of the graph is bounded by $\Delta$.

\subparagraph{Degree-types of paths.}
Since the number of labels may be unbounded, so is the number of label-types.
Instead, we consider degree-types.
Let $G$ be a graph, $u,v\in V(G)$, $P$ be a path from $u$ to $v$ and $O$ be a partial orientation of the line graph of $G$.
Then the \emph{degree-type} of $P$ is $e_1\tau e_2$, where $e_1$ (resp. $e_2$) is the edge of $P$ adjacent to $u$ (resp. $v$) and $\tau\in\{\emptyset, {\sf m, M, mM, Mm}\}$ is the type of $P$ with respect to $O$, i.e. with respect to $\lambda_O$.

Again, we say that $e_1\tau e_2\le e_1'\tau' e_2'$ if and only if $e_1=e_1'$, $e_2=e_2'$, and $\tau\le\tau'$.
Note that \autoref{obs:type} still holds when taking the degree-type instead of the label-type.

\subparagraph{Signature of a partial orientation.}
Let $O$ be a partial orientation of the line graph of $G$ and $x$ be a node of $T$.
Let $F_x$ be the graph with vertex set the union of $B_x$ and its neighbors in $G_x$, and with edge set the edges with one endpoint in $B_x$ and the other in $B_x\cup V_x$.
Let $O_x$ be the restriction of $O$ to the line graph of $F_x$.
For $(u,v)\in B_x^2$, let $F_x^{u,v}$ be the set of degree-types, with respect to $O$, of the $(u,v)$-paths whose internal vertices are in $G_x$.
Let $f_x$ be the function that maps each pair $(u,v)\in B_x^2$ to the minimal elements of $F_x^{u,v}$.
Then the \emph{signature} of $O$ at $x$, denoted by $\sig_O^x$, is the pair $(O_x, f_x)$.

Given that $F_x$ has $\tw\cdot\Delta$ edges, and that each of them is incident to at most $2\Delta$ other edges,
there are $\tw\cdot\Delta^2$ edges in the line graph of $F_x$.
Therefore, there are $3^{\tw\cdot\Delta^2}$ choices for $O_x$.
Additionally, there are again eight possible sets of minimal elements, so there are $(8\Delta^2)^{\tw^2}$ choices for $f_x$.
Therefore, there at most $3^{\tw\cdot\Delta^2}\cdot(8\Delta^2)^{\tw^2}$ distinct signatures.

\subparagraph{Signature of a bag.}
Now, the \emph{signature} of a vertex $x$ of $T$, denoted by $\bar\sig_x$, is the set of all signatures $\sig_O^x$ where $O$ is a partial orientation of the line graph of $G$ such that:
\begin{itemize}
\item for any cycle $C$ of $G_x$, the corresponding cycle in the line graph of $G_x$ has at least two sources and
\item for any vertex $v$ of $G_x$, there are no three edges adjacent to $v$ that induce a partial dicycle with respect to $O$.
\end{itemize}
%
Hence, for any partial orientation $O$ at the root $r$ of $T$ whose signature is in $\bar\sig_r$, $O$ is a good partial orientation of $G$.

The forget, join, and introduction operations are almost exactly the same as in \autoref{sec:DPcolors}.
Therefore, we describe how to obtain $\bar\sig_x$ from its children for each vertex $x$ of $T$, but without the formal proof, which is left as an exercise to the reader.

\subparagraph{Leaf bag.}
If $B_x$ is a leaf bag, then we trivially have $$\bar\sig_x=\{(\emptyset,\emptyset)\}.$$

\subparagraph{Introduce bag.}
If $B_x$ is an introduce bag, then there is an introduced vertex $v\in V(G)$ such that $B_x=B_y\cup\{v\}$, where $y$ is the unique child of $x$ in $T$.
Let $E_v$ be the set of vertices of $B_x$ adjacent to $v$.
Given a partial orientation $O_y$ on the line graph of $F_y$,
we define $D_{O_y}$ to be the set of all partial orientations $O_x$ on the line graph of $F_x$ whose restriction to the line graph of $F_y$ is $O_y$.
We can do so because $F_y=F_x\setminus E_v$.
Given that the maximum degree is $\Delta$ and that $|B_x|\le\tw$,
we have $|E(F_x)|\le\Delta\cdot \tw$, and thus
$|D_{O_y}|\le3^{\Delta^2\cdot \tw}$ (each edge of $F_x$ is incident to at most $2\Delta$ other edges).

Then $$\bar\sig_x=\{(O_x,f_y')\mid (O_y,f_y)\in\bar\sig_y,O_x\in D_{O_y}\},$$
where $f_y'$ is equal to $f_y$ on $B_y^2$ and equal to the $\emptyset$ on all $(v,u)$ or $(u,v)$ with $u\in B_x$,
and can be constructed from $\bar\sig_y$ in time $\Ocal(3^{\Delta^2\cdot \tw}\cdot|\bar\sig_y|)$.

\subparagraph{Join bag.}
If $B_x$ is a join bag, then $B_x=B_{y_1}=B_{y_2}$, where $y_1$ and $y_2$ are the unique children of $x$ in $T$.
Then $\bar\sig_x$ is equal to
$$\{(O_{y_1}\cup O_{y_2},\min(f_{y_1}\cup f_{y_2}))\mid \forall i\in[2], (O_{y_i},f_{y_i})\in\bar\sig_{y_i} \text{ and } O_{y_1}\cap E(G[B_x])=O_{y_2}\cap E(G[B_x])\},$$
where $\min(f_{y_1}\cup f_{y_2})$ is defined as in \autoref{sec:DPcolors},
and can be constructed in time $\Ocal(\Delta^2\cdot \tw^2\cdot|\bar\sig_{y_1}|\cdot|\sig_{y_2}|)$.

\subparagraph{Forget bag.}
If $B_x$ is a forget bag then there is an introduced vertex $v\in V(G)$ such that $B_x=B_y\cup\{v\}$, where $y$ is the unique child of $x$ in $T$.

Given $(O_y,f_y)\in\sig_y$, let $O_y\setminus E_v$ denote the restriction of $O_y$ to the line graph of $F_x=F_y\setminus E_v$.

As in \autoref{sec:DPcolors}, we need to reject some of the elements $(O_y,f_y)\in\bar\sig_y$.
Here, we want that there is a partial orientation $O$ of the line graph of $G$ whose restriction at node $y$ is $(O_y,f_y)$ and
\begin{enumerate}
\item for any cycle $C$ of $G_x$, the corresponding cycle in the line graph of $G_x$ has at least two sources and
\item for any vertex $u$ of $G_x$, there are no three edges adjacent to $u$ that induce a partial dicycle with respect to $O$.
\end{enumerate}
By induction on $y$, given that $V_x=V_y\cup\{y\}$,
it is enough for item (1) to check that
 for any cycle $C$ of $G_x$ intersecting $v$, the corresponding cycle in the line graph of $G_x$ has at least two sources.
 Additionally, it is enough for item (2) to check that there are no three edges adjacent to $v$ that induce a partial dicycle with respect to $O$.
Hence, for item (1), as in \autoref{sec:DPcolors}, we reject the elements $(O_y,f_y)\in\bar\sig_y$ such that $e_1\tau e_2\in f_y(v,v)$, where:
\begin{itemize}
    \item $\tau\in\{\emptyset,{\sf m, M}\}$, or
    \item $\tau={\sf mM}$ and $(e_2,e_1)\notin O_y$, or
    \item $\tau={\sf Mn}$ and $(e_1,e_2)\notin O_y$.
\end{itemize}
As for item (2), we reject the elements $(O_y,f_y)\in\bar\sig_y$ such that there are three edges $e_1,e_2,e_3$ adjacent to $v$ such that $(e_1,e_2)\in O_y$, $(e_3,e_2)\notin O_y$, and $(e_1,e_3)\notin O_y$.
Then let $\bar\sig_y'$ be the set of unrejected elements of $\bar\sig_y$.

Finally, we define $f_y^2$ as in \autoref{sec:DPcolors}, but with degree-types instead of label-types.
Then $$\bar\sig_x=\{(O_y\setminus E_v, f_y^2)\mid(O_y,f_y)\in\bar\sig_y'\},$$
and can be computed in time $\Ocal(\tw^2\cdot\Delta^4\cdot|\bar\sig_y|)$.

\subparagraph{Complexity.}
Given that the number of distinct signatures is at most $3^{\tw\cdot\Delta^2}\cdot(8\Delta^2)^{\tw^2}$,
the running time is thus $2^{\Ocal(\tw\Delta^2+\tw^2\log\Delta)}\cdot n$.


\medskip \noindent \textbf{Acknowledgement}. We would like to thank the reviewers for helpful and thorough comments that improved the presentation of the manuscript.

\bibliography{references}

\newpage

\appendix

\section{\NP-completeness of deciding the existence of a UPP-orientation}
\label{sec:UPP-NPh}


In this section we answer the open question raised by Bermond et al.~\cite{BermondCP13} about the complexity of finding UPP-orientations (recall \autoref{footnote1}, where it was mentioned that very recently Dohnalová et al.~\cite{another-proof-UPP} independently proved the \NP-hardness of a related problem). We first need some definitions, some of which were already given in the introduction.
Given an undirected graph $G$, an \emph{orientation of $G$} is a directed graph $D$ obtained from $G$ by replacing each edge $uv\in E(G)$ by exactly one of the arcs $(u,v)$ or $(v,u)$. We say that $D$ is a \emph{DAG} if it does not contain any directed cycle. Recall that a digraph $D$ is called an \textit{UPP-digraph} if it satisfies the \emph{unique path property} (UPP): for every ordered pair $x,y\in V(G)$, there is at most one directed $(x,y)$-path in $D$ (called from now on \emph{$(x,y)$-dipath}). We call an orientation of $G$ that is also a UPP-digraph an \emph{UPP-orientation}. Note that if $D$ is a UPP-digraph, then $H$ is also a UPP-digraph for every subgraph $H$ of $D$. It is also easy to see that, up to isomorphism, the only UPP-orientation of the triangle is the directed cycle. A cycle $C = (x,y,z)$ in $G$ has the nice property that any UPP-orientation of $C$ contains already an $(x,y)$-dipath and a $(y,x)$-dipath. Hence it works in $G$ as if an edge $xy\in E(G)$ could be, and was forced to be, oriented in both ways in all UPP-orientations of $G$. In other words, if $G$ contains a cycle $(x,y,z)$, then no $(x,y)$-dipath and no $(y,x)$-dipath not contained in $D[\{x,y,z\}]$ can be created. This is going to be exploited in our constructions. In \autoref{fig:UPP_variable_block} and \autoref{fig:UPP_variable} double-oriented edges actually represent this situation, i.e. there is a triangle (the third vertex is not depicted) and there must be in any UPP-orientation an $(x,y)$-dipath and a $(y,x)$-dipath, represented by the double-oriented edge, if $x$ and $y$ are the endpoints.



\begin{theorem}\label{thm:UPP-hard}
    Given a graph $G$, deciding whether $G$ admits a UPP-orientation is $\NP$-complete even if $G$ has maximum degree at most five.
\end{theorem}
\begin{proof}
    Given an orientation $D$ of $G$, one can decide whether $D$ is a UPP-orientation simply by computing the maximum number of disjoint $s,t$-paths for every pair of vertices $s,t$, which is well-known to be doable in polynomial time. Hence, the problem is in \NP.

    To prove $\NP$-hardness, similarly to what we did in \autoref{sec:NPhard}, we present a reduction from \NAE 3-{\SAT} where each clause contains exactly three literals~\cite{DarmannD20}. Our basic building block is the graph presented in \autoref{fig:UPP_variable_block}, about which we prove the following important property. With a slight abuse of notation, in the sequel when we refer to an arc $(x,y)$ corresponding to a double-oriented edge in~\autoref{fig:UPP_variable_block} and~\autoref{fig:UPP_variable}, we mean the $(x,y)$-dipath contained in the corresponding triangle.

\begin{figure}[h!tb]
\begin{center}
\input{UPP_variable_block-v2}
\end{center}
\caption{Building block for the variable gadget together with possible orientations. Thick double-oriented edges represent triangles.
}
\label{fig:UPP_variable_block}
\end{figure}
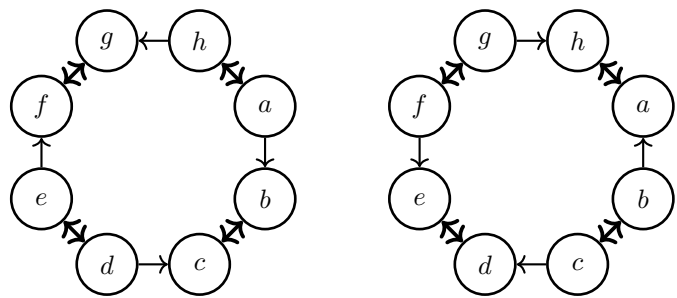

\begin{claim}\label{lemma:UPP_octagone}
Let $G$ be the cycle on eight vertices together with four additional vertices that form disjoint triangles with edges of the cycle. Then in any UPP-orientation of $G$ the edges that do not belong to triangles can only be oriented as depicted in~\autoref{fig:UPP_variable_block}.
\end{claim}
\begin{proof}
We use the notation of \autoref{fig:UPP_variable_block}.
We prove that no two consecutive edges among $\{ab,cd,ef,gh\}$ can be oriented in the same direction. Suppose by contradiction that $ab$ and $cd$ are both oriented in the clockwise direction. Then at least one between $ef$  and $gh$ must be oriented in the counter-clockwise direction as otherwise we would have the $(a,h)$-dipaths: $(a,h)$ and $(a,b,c,d,e,f,g,h)$.

Now suppose that exactly one of them is oriented in counter-clockwise direction. If it is $ef$, then we have the two $(f,e)$-dipaths: $(f,e)$ and $(f,g,h,a,b,c,d,e)$. And if it is $gh$, then we have the $(h,g)$-dipaths: $(h,g)$ and $(h,a,b,c,d,e,f,g)$.

Finally, if both $ef$ and $gh$ are oriented in the counter-clockwise direction, then we get the $(h,e)$-dipaths: $(h,g,f,e)$ and $(h,a,b,c,d,e)$.

It follows that there cannot be two consecutive edges that do not belong to triangles oriented in the same direction, as we wanted to prove. Starting by orienting $ab$, one can see that the two given orientations are the only possible left ones. Finally, observe that the longest dipath in any of the orientations is of length three. This means that if $D$ is an orientation containing an $(x,y)$-dipath for some pair $x,y\in \{a,\ldots,h\}$, then the other $(x,y)$-path in $G$ defined by the graph depicted in~\autoref{fig:UPP_variable_block} has length at least 5 and thus cannot be a dipath in $D$.
\end{proof}

We now construct the variable gadgets. So let $\varphi$ be a formula on variables $\{x_1,\ldots, x_n\}$ and clauses $\{c_1,\ldots,c_m\}$. For each $x_i$, 
``glue'' together $4m$ copies of our building block depicted in~\autoref{fig:UPP_variable_block} through edges $ab$ and $ef$ as represented in~\autoref{fig:UPP_variable}, each 4 copies of which will be related to a clause. In fact, we could add $4p_i$ copies instead, where $p_i$ is the number of appearances of variable $x_i$. However, adding $4m$ copies will make presentation simpler.
To simplify notation, we index vertices related to the appearance of $x_i$ in clause $c_j$ by the underscript $i,j$. Observe that the following property holds:

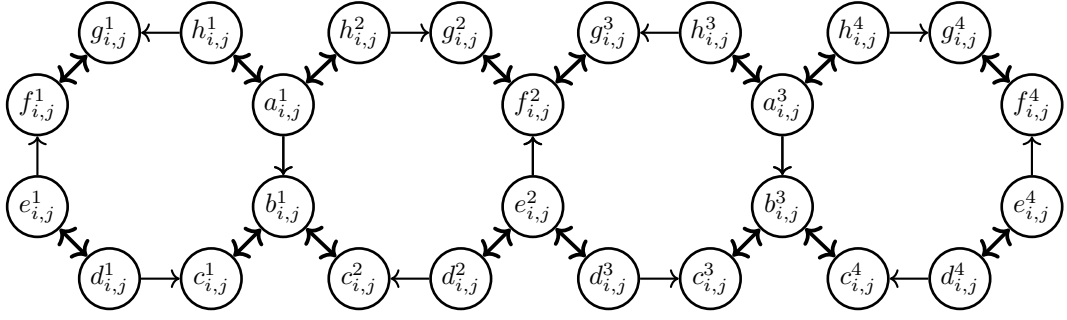
\begin{figure}
\begin{center}
\input{UPP_variable}
\end{center}
\caption{Part of variable gadget related to the appearance of $x_i$ in clause $c_j$. Thick double-oriented edges represent triangles. We oriented the edges supposing that the first $cd$ is oriented from $d$ to $c$.}\label{fig:UPP_variable}
\end{figure}

\begin{observation}
\label{lemma:variable_orientation}
If $D$ is a UPP-orientation of a variable gadget, then all odd copies of the edge $cd$ are oriented in the same way, i.e., either they are all from $c$ to $d$ or all from $d$ to~$c$.
\end{observation}
The above property will be used to signal the truth value of $x_i$. If the odd copies of $cd$ are oriented from $c$ to $d$, then we interpret as $x_i$ being \T; otherwise, it is interpreted as being \F.

Now, we show how to build the clause gadgets. So consider clause $c_m$ on literals $\ell_1,\ell_2,\ell_3$. We add two new vertices $f_j$ and $l_j$ and link them and the odd copies of $cd$ within the variable gadgets in a way that not all the edges can be oriented in the same direction. Formally (see \autoref{fig:UPP_clause} to follow the construction), for each $h\in[3]$, let $x_{i_h}$ be the variable related to $\ell_h$ and, for each $p\in[2]$, let $(\alpha^p_h,\beta^p_h)$ be equal to $(c^p_{i_h,j},d^p_{i_h,j})$ if $\ell_h$ is equal to $x_{i_h}$; and if $\ell_h$ is equal to $\overline x_{i_h}$, let $(\alpha^p_h,\beta^p_h)$ be equal to $(d^p_{i_h,j},c^p_{i_h,j})$.  In what follows, by ``link $x$ to $y$'' we mean linking through a triangle and each triangle is created by adding a new vertex. So, link $f_j$ to $\alpha^1_1$ and $\alpha^2_1$; for each $p\in[2]$, link $\beta^p_1$ to $\alpha^p_2$; for each $p\in[2]$, link $\beta^p_2$ to $\alpha^p_3$; finally, link $\beta^1_3$ and $\beta^2_3$ to $l_j$. See \autoref{fig:UPP_clause_example} for an example.

\begin{figure}[h!tb]
\begin{center}
\input{UPP_clause}
\end{center}
\caption{Gadget related to clause $c_j$ on literals $\ell_1$, $\ell_2$ and $\ell_3$. Edges labeled with $\ell_i$ represent the first and third copies of $cd$ in the part of the gadget of $\ell_i$ related to $c_j$. Thick edges represent triangles.}\label{fig:UPP_clause}
\end{figure}
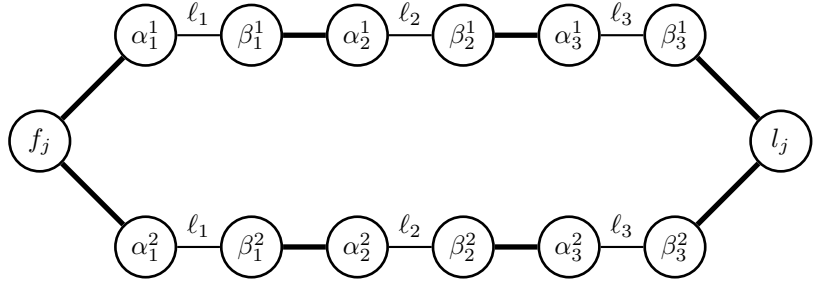

\begin{figure}[h!tb]
\begin{center}
\input{UPP_clause_example}
\end{center}
\caption{Gadget related to clause $c_j = (x_1\vee \overline x_2\vee x_3)$. Thick edges represent triangles.}\label{fig:UPP_clause_example}
\end{figure}
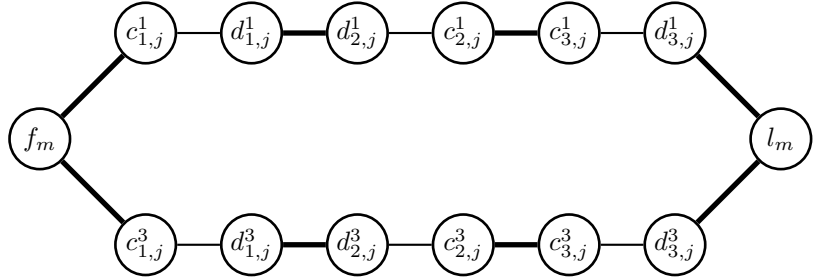

Denote by $G$ the obtained graph. First observe that every vertex in $\{f_j,l_j\mid j\in [m]\}$ has degree exactly four, while every vertex in $\{x^p_{i,j}\mid x\in \{g,h\},p\in [4],i\in[n],j\in[m]\}$ has degreee exactly three.
Additionally, vertices in $\{x^p_{i,j}\mid x\in\{a,b\},p\in[4],i\in[n],j\in[m]\}$ have degree five and vertices in $\{x^p_{i,j}\mid x\in\{e,f\},p\in[4],i\in[n],j\in[m]\}$ have degree either equal to three or to five.
Finally, for every $i\in [n]$ and every $j\in [m]$, we have that either variable $x_i$ does not appear in $c_j$, in which case the vertices in $S_{i,j}= \{c^1_{i,j},c^3_{i,j},d^1_{i,j},d^3_{i,j}\}$ also have degree three, or $x_i$ does appear in $c_j$, in which case exactly two edges incident to each vertex in $S_{i,j}$ were added to $G$ (either by linking $f_j$ to $\alpha_h^p$, by linking $l_j$ to $\beta_h^p$ or by linking $\beta_h^p$ to $\alpha_{h+1}^p$), in which case the vertices in $S_{i,j}$ have degree five.

We now prove that $\varphi$ has a \NAE satisfying assignment if and only if $G$ has a UPP-orientation. First consider a \NAE satisfying assignment to $\varphi$ and, for each variable $x_i$, orient all the odd occurrences of $cd$ in the gadget of $x_i$ from $c$ to $d$ if $x_i$ is \T, and from $d$ to $c$, otherwise. Observe that the orientations of all other edges are implied by \autoref{lemma:UPP_octagone}; denote by $D$ the obtained orientation. We need to argue that there are no two vertices $x,y$ such that $D$ contains two $(x,y)$-dipaths. Suppose otherwise. We first argue that it cannot happen that both paths are contained in a variable gadget.
Again, just notice that the longest possible dipath in $D$ constrained to a variable gadget has endpoints either $h^1_{i,j}$ and $g^3_{i,j}$ or $c^1_{i,j}$ and $d^3_{i,j}$, depending on the orientation of the edge between $a^1_{i,j}$ and $b^1_{i,j}$, and in this case the endpoints are in disjoint $C_8$'s; it also contains a dipath ending in neighboring $C_8$'s. Furthermore, note that there is a dipath passing through the edge between $a^1_{i,j}$ and $b^1_{i,j}$, but it ends in neighboring $C_8$'s as well. In each of the possible cases, any other path within the same pair of endpoints is too long and cannot be a dipath.
Now, observe \autoref{fig:UPP_clause_orientation} to see that the paths linking $x$ and $y$ also cannot be within a clause gadget. Finally, note that if $x$ and $y$ are not within the same clause gadget and are not within the same variable gadget, then they are too far apart to be connected through any dipath. To see that recall that all vertices are related to some clause, even if they are not within that clause.

\begin{figure}[h!tb]
\begin{center}
\input{UPP_clause_orientation}
\end{center}
\caption{Possible orientations of a clause gadget. Thick double-oriented edges represent triangles. The labels inside the cycles represent the truth assignments of the clause's literals ($T$ meaning \T and $F$ meaning \F).}\label{fig:UPP_clause_orientation}%
\end{figure}
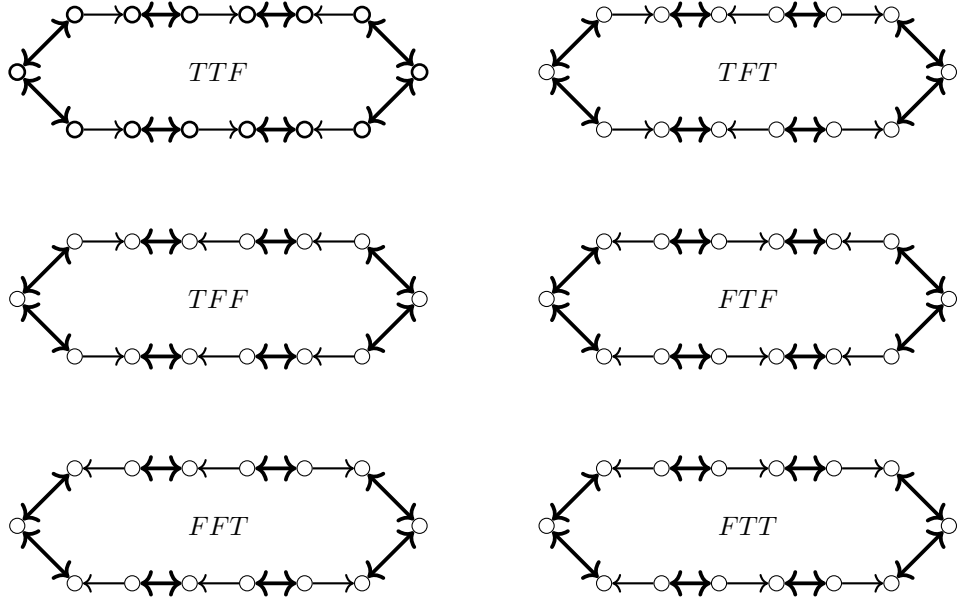

Now, let $D$ be a UPP-orientation of $G$. By \autoref{lemma:variable_orientation}, we know that all odd copies of the edges $cd$ within the same variable gadget are oriented in the same way. We then set a variable $x_i$ to \T if these edges are oriented from $c$ to $d$ and to \F, otherwise. Again by \autoref{lemma:variable_orientation}, given a clause $c_j$, we also get that the edges related to the same literal are oriented in the same direction within the two $(f_j,l_j)$-paths contained in the gadget related to $c_j$. It thus follows that the literal edges within such a gadget 
cannot all be oriented in the same direction as otherwise $D$ would contain either two $(f_j,l_j)$-dipaths or two $(l_j,f_j)$-dipaths. This means that not all literals have the same truth value and hence we have a \NAE truth assignment for $\varphi$. 
\end{proof}

\end{document}

%% file: UPP_variable_block-v2.tex
\begin{tikzpicture}[scale=.8]
  \pgfsetlinewidth{1pt}

  \tikzset{vertex/.style={circle, minimum size=0.8cm, draw, inner sep=1pt}}

    \node[vertex] (a) at (22.5:2) {$a$};
    \node[vertex] (b) at (-22.5:2) {$b$};
    \node[vertex] (c) at (-67.5:2) {$c$};
    \node[vertex] (d) at (-112.5:2) {$d$};
    \node[vertex] (e) at (-157.5:2) {$e$};
    \node[vertex] (f) at (-202.5:2) {$f$};
    \node[vertex] (g) at (-247.5:2) {$g$};
    \node[vertex] (h) at (-292.5:2) {$h$};

       \draw[line width=0.5mm,<->] (b)--(c);
       \draw[line width=0.5mm,<->] (d)--(e);
       \draw[line width=0.5mm,<->] (f)--(g);
       \draw[line width=0.5mm,<->] (h)--(a);

       \draw[thick,->] (a)--(b);
       \draw[thick,<-] (c)--(d);
       \draw[thick,->] (e)--(f);
       \draw[thick,<-] (g)--(h);
       

    \node[vertex,xshift=5cm] (ap) at (22.5:2) {$a$};
    \node[vertex,xshift=5cm] (bp) at (-22.5:2) {$b$};
    \node[vertex,xshift=5cm] (cp) at (-67.5:2) {$c$};
    \node[vertex,xshift=5cm] (dp) at (-112.5:2) {$d$};
    \node[vertex,xshift=5cm] (ep) at (-157.5:2) {$e$};
    \node[vertex,xshift=5cm] (fp) at (-202.5:2) {$f$};
    \node[vertex,xshift=5cm] (gp) at (-247.5:2) {$g$};
    \node[vertex,xshift=5cm] (hp) at (-292.5:2) {$h$};

       \draw[line width=0.5mm,<->] (bp)--(cp);
       \draw[line width=0.5mm,<->] (dp)--(ep);
       \draw[line width=0.5mm,<->] (fp)--(gp);
       \draw[line width=0.5mm,<->] (hp)--(ap);

       \draw[thick,<-] (ap)--(bp);
       \draw[thick,->] (cp)--(dp);
       \draw[thick,<-] (ep)--(fp);
       \draw[thick,->] (gp)--(hp);
     
  \end{tikzpicture}

%% file: UPP_variable.tex
\begin{tikzpicture}[scale=0.88]
  \pgfsetlinewidth{1pt}
  \pgfdeclarelayer{bg}    
   \pgfsetlayers{bg,main}  

  \tikzset{vertex/.style={circle, minimum size=0.8cm, draw, inner sep=1pt}}

    \node[vertex] (a1) at (22.5:2) {$a_{i,j}^1$};
    \node[vertex] (b1) at (-22.5:2) {$b_{i,j}^1$};
    \node[vertex] (c1) at (-67.5:2) {$c_{i,j}^1$};
    \node[vertex] (d1) at (-112.5:2) {$d_{i,j}^1$};
    \node[vertex] (e1) at (-157.5:2) {$e_{i,j}^1$};
    \node[vertex] (f1) at (-202.5:2) {$f_{i,j}^1$};
    \node[vertex] (g1) at (-247.5:2) {$g_{i,j}^1$};
    \node[vertex] (h1) at (-292.5:2) {$h_{i,j}^1$};

    \node[vertex,xshift=3.3cm] (c2) at (-112.5:2) {$c_{i,j}^2$};
    \node[vertex,xshift=3.3cm] (d2) at (-67.5:2) {$d_{i,j}^2$};
    \node[vertex,xshift=3.3cm] (e2) at (-22.5:2) {$e_{i,j}^2$};
    \node[vertex,xshift=3.3cm] (f2) at (22.5:2) {$f_{i,j}^2$};
    \node[vertex,xshift=3.3cm] (g2) at (-292.5:2) {$g_{i,j}^2$};
    \node[vertex,xshift=3.3cm] (h2) at (-247.5:2) {$h_{i,j}^2$};

    \node[vertex,xshift=6.6cm] (a3) at (22.5:2) {$a_{i,j}^3$};
    \node[vertex,xshift=6.6cm] (b3) at (-22.5:2) {$b_{i,j}^3$};
    \node[vertex,xshift=6.6cm] (c3) at (-67.5:2) {$c_{i,j}^3$};
    \node[vertex,xshift=6.6cm] (d3) at (-112.5:2) {$d_{i,j}^3$};
    \node[vertex,xshift=6.6cm] (g3) at (-247.5:2) {$g_{i,j}^3$};
    \node[vertex,xshift=6.6cm] (h3) at (-292.5:2) {$h_{i,j}^3$};

    \node[vertex,xshift=9.9cm] (c4) at (-112.5:2) {$c_{i,j}^4$};
    \node[vertex,xshift=9.9cm] (d4) at (-67.5:2) {$d_{i,j}^4$};
    \node[vertex,xshift=9.9cm] (e4) at (-22.5:2) {$e_{i,j}^4$};
    \node[vertex,xshift=9.9cm] (f4) at (22.5:2) {$f_{i,j}^4$};
    \node[vertex,xshift=9.9cm] (g4) at (-292.5:2) {$g_{i,j}^4$};
    \node[vertex,xshift=9.9cm] (h4) at (-247.5:2) {$h_{i,j}^4$};

    \begin{pgfonlayer}{bg}    
       \draw[line width=0.5mm,<->] (b1)--(c1);
       \draw[line width=0.5mm,<->] (d1)--(e1);
       \draw[line width=0.5mm,<->] (f1)--(g1);
       \draw[line width=0.5mm,<->] (h1)--(a1);

       \draw[thick,->] (a1)--(b1);
       \draw[thick,<-] (c1)--(d1);
       \draw[thick,->] (e1)--(f1);
       \draw[thick,<-] (g1)--(h1);

       \draw[line width=0.5mm,<->] (b1)--(c2);
       \draw[line width=0.5mm,<->] (d2)--(e2);
       \draw[line width=0.5mm,<->] (f2)--(g2);
       \draw[line width=0.5mm,<->] (h2)--(a1);

       \draw[thick,->] (a1)--(b1);
       \draw[thick,<-] (c2)--(d2);
       \draw[thick,->] (e2)--(f2);
       \draw[thick,<-] (g2)--(h2);

       \draw[line width=0.5mm,<->] (b3)--(c3);
       \draw[line width=0.5mm,<->] (d3)--(e2);
       \draw[line width=0.5mm,<->] (f2)--(g3);
       \draw[line width=0.5mm,<->] (h3)--(a3);

       \draw[thick,->] (a3)--(b3);
       \draw[thick,<-] (c3)--(d3);
       \draw[thick,->] (e2)--(f2);
       \draw[thick,<-] (g3)--(h3);

       \draw[line width=0.5mm,<->] (b3)--(c4);
       \draw[line width=0.5mm,<->] (d4)--(e4);
       \draw[line width=0.5mm,<->] (f4)--(g4);
       \draw[line width=0.5mm,<->] (h4)--(a3);

       \draw[thick,->] (a3)--(b3);
       \draw[thick,<-] (c4)--(d4);
       \draw[thick,->] (e4)--(f4);
       \draw[thick,<-] (g4)--(h4);
       
    \end{pgfonlayer}

  \end{tikzpicture}

%% file: UPP_clause.tex
\begin{tikzpicture}[scale=0.7]
  \pgfsetlinewidth{1pt}
  \pgfdeclarelayer{bg}    
   \pgfsetlayers{bg,main}  

  \tikzset{vertex/.style={circle, minimum size=0.8cm, draw, inner sep=1pt}}

    \node[vertex] (s) at (0,0) {$f_j$};
    \node[vertex] (a) at (2,2) {$\alpha^1_1$};
    \node[vertex] (b) at (4,2) {$\beta^1_1$};
    \node[vertex] (c) at (6,2) {$\alpha^1_2$};
    \node[vertex] (d) at (8,2) {$\beta^1_2$};
    \node[vertex] (e) at (10,2) {$\alpha^1_3$};
    \node[vertex] (f) at (12,2) {$\beta^1_3$};
    \node[vertex] (ap) at (2,-2) {$\alpha^2_1$};
    \node[vertex] (bp) at (4,-2) {$\beta^2_1$};
    \node[vertex] (cp) at (6,-2) {$\alpha^2_2$};
    \node[vertex] (dp) at (8,-2) {$\beta^2_2$};
    \node[vertex] (ep) at (10,-2) {$\alpha^2_3$};
    \node[vertex] (fp) at (12,-2) {$\beta^2_3$};
    \node[vertex] (l) at (14,0) {$l_j$};

    \begin{pgfonlayer}{bg}    
       \draw[line width=0.7mm] (s)--(a);
       \draw[thick] (a)--(b) node [above, midway] {$\ell_1$};
       \draw[line width=0.7mm] (b)--(c);
       \draw[thick] (c)--(d)  node [above, midway] {$\ell_2$};
       \draw[line width=0.7mm] (d)--(e);
       \draw[thick] (e)--(f)  node [above, midway] {$\ell_3$};
       \draw[line width=0.7mm] (f)--(l);

       \draw[line width=0.7mm] (s)--(ap);
       \draw[thick] (ap)--(bp)  node [above, midway] {$\ell_1$};
       \draw[line width=0.7mm] (bp)--(cp);
       \draw[thick] (cp)--(dp) node [above, midway] {$\ell_2$};
       \draw[line width=0.7mm] (dp)--(ep);
       \draw[thick] (ep)--(fp) node [above, midway] {$\ell_3$};
       \draw[line width=0.7mm] (fp)--(l);
    \end{pgfonlayer}

  \end{tikzpicture}

%% file: UPP_clause_example.tex
\begin{tikzpicture}[scale=0.7]
  \pgfsetlinewidth{1pt}
  \pgfdeclarelayer{bg}    
   \pgfsetlayers{bg,main}  

  \tikzset{vertex/.style={circle, minimum size=0.8cm, draw, inner sep=1pt}}

    \node[vertex] (s) at (0,0) {$f_m$};
    \node[vertex] (a) at (2,2) {$c^1_{1,j}$};
    \node[vertex] (b) at (4,2) {$d^1_{1,j}$};
    \node[vertex] (c) at (6,2) {$d^1_{2,j}$};
    \node[vertex] (d) at (8,2) {$c^1_{2,j}$};
    \node[vertex] (e) at (10,2) {$c^1_{3,j}$};
    \node[vertex] (f) at (12,2) {$d^1_{3,j}$};
    
    \node[vertex] (ap) at (2,-2) {$c^3_{1,j}$};
    \node[vertex] (bp) at (4,-2) {$d^3_{1,j}$};
    \node[vertex] (cp) at (6,-2) {$d^3_{2,j}$};
    \node[vertex] (dp) at (8,-2) {$c^3_{2,j}$};
    \node[vertex] (ep) at (10,-2) {$c^3_{3,j}$};
    \node[vertex] (fp) at (12,-2) {$d^3_{3,j}$};
    \node[vertex] (l) at (14,0) {$l_m$};

    \begin{pgfonlayer}{bg}    
       \draw[line width=0.7mm] (s)--(a);
       \draw[thick] (a)--(b);
       \draw[line width=0.7mm] (b)--(c);
       \draw[thick] (c)--(d);
       \draw[line width=0.7mm] (d)--(e);
       \draw[thick] (e)--(f);
       \draw[line width=0.7mm] (f)--(l);

       \draw[line width=0.7mm] (s)--(ap);
       \draw[thick] (ap)--(bp);
       \draw[line width=0.7mm] (bp)--(cp);
       \draw[thick] (cp)--(dp);
       \draw[line width=0.7mm] (dp)--(ep);
       \draw[thick] (ep)--(fp);
       \draw[line width=0.7mm] (fp)--(l);
    \end{pgfonlayer}

  \end{tikzpicture}

%% file: UPP_clause_orientation.tex
\begin{tikzpicture}[scale=0.38]
  \pgfsetlinewidth{1pt}
  \pgfdeclarelayer{bg}    
   \pgfsetlayers{bg,main}  

  \tikzset{vertex/.style={circle, minimum size=0.2cm, draw, inner sep=1pt}}

    \node[vertex] (s) at (0,0) {};
    \node[vertex] (a) at (2,2) {};
    \node[vertex] (b) at (4,2) {};
    \node[vertex] (c) at (6,2) {};
    \node[vertex] (d) at (8,2) {};
    \node[vertex] (e) at (10,2) {};
    \node[vertex] (f) at (12,2) {};
    \node[vertex] (ap) at (2,-2) {};
    \node[vertex] (bp) at (4,-2) {};
    \node[vertex] (cp) at (6,-2) {};
    \node[vertex] (dp) at (8,-2) {};
    \node[vertex] (ep) at (10,-2) {};
    \node[vertex] (fp) at (12,-2) {};
    \node[vertex] (l) at (14,0) {};
    \node at (7,0) {$TTF$};
    
    \begin{pgfonlayer}{bg}    
       \draw[line width=0.5mm,<->] (s)--(a);
       \draw[thick,->] (a)--(b);
       \draw[line width=0.5mm,<->] (b)--(c);
       \draw[thick,->] (c)--(d);
       \draw[line width=0.5mm,<->] (d)--(e);
       \draw[thick,<-] (e)--(f);
       \draw[line width=0.5mm,<->] (f)--(l);

       \draw[line width=0.5mm,<->] (s)--(ap);
       \draw[thick,->] (ap)--(bp);
       \draw[line width=0.5mm,<->] (bp)--(cp);
       \draw[thick,->] (cp)--(dp);
       \draw[line width=0.5mm,<->] (dp)--(ep);
       \draw[thick,<-] (ep)--(fp);
       \draw[line width=0.5mm,<->] (fp)--(l);
    \end{pgfonlayer}

    \node[xshift=7cm,vertex] (s) at (0,0) {};
    \node[xshift=7cm, vertex] (a) at (2,2) {};
    \node[xshift=7cm, vertex] (b) at (4,2) {};
    \node[xshift=7cm, vertex] (c) at (6,2) {};
    \node[xshift=7cm, vertex] (d) at (8,2) {};
    \node[xshift=7cm, vertex] (e) at (10,2) {};
    \node[xshift=7cm, vertex] (f) at (12,2) {};
    \node[xshift=7cm, vertex] (ap) at (2,-2) {};
    \node[xshift=7cm, vertex] (bp) at (4,-2) {};
    \node[xshift=7cm, vertex] (cp) at (6,-2) {};
    \node[xshift=7cm, vertex] (dp) at (8,-2) {};
    \node[xshift=7cm, vertex] (ep) at (10,-2) {};
    \node[xshift=7cm, vertex] (fp) at (12,-2) {};
    \node[xshift=7cm, vertex] (l) at (14,0) {};
    \node[xshift=7cm] at (7,0) {$TFT$};
    
    \begin{pgfonlayer}{bg}    
       \draw[line width=0.5mm,<->] (s)--(a);
       \draw[thick,->] (a)--(b);
       \draw[line width=0.5mm,<->] (b)--(c);
       \draw[thick,<-] (c)--(d);
       \draw[line width=0.5mm,<->] (d)--(e);
       \draw[thick,->] (e)--(f);
       \draw[line width=0.5mm,<->] (f)--(l);

       \draw[line width=0.5mm,<->] (s)--(ap);
       \draw[thick,->] (ap)--(bp);
       \draw[line width=0.5mm,<->] (bp)--(cp);
       \draw[thick,<-] (cp)--(dp);
       \draw[line width=0.5mm,<->] (dp)--(ep);
       \draw[thick,->] (ep)--(fp);
       \draw[line width=0.5mm,<->] (fp)--(l);
    \end{pgfonlayer}

    \node[yshift=-3cm,vertex] (s) at (0,0) {};
    \node[yshift=-3cm, vertex] (a) at (2,2) {};
    \node[yshift=-3cm, vertex] (b) at (4,2) {};
    \node[yshift=-3cm, vertex] (c) at (6,2) {};
    \node[yshift=-3cm, vertex] (d) at (8,2) {};
    \node[yshift=-3cm, vertex] (e) at (10,2) {};
    \node[yshift=-3cm, vertex] (f) at (12,2) {};
    \node[yshift=-3cm, vertex] (ap) at (2,-2) {};
    \node[yshift=-3cm, vertex] (bp) at (4,-2) {};
    \node[yshift=-3cm, vertex] (cp) at (6,-2) {};
    \node[yshift=-3cm, vertex] (dp) at (8,-2) {};
    \node[yshift=-3cm, vertex] (ep) at (10,-2) {};
    \node[yshift=-3cm, vertex] (fp) at (12,-2) {};
    \node[yshift=-3cm, vertex] (l) at (14,0) {};
    \node[yshift=-3cm] at (7,0) {$TFF$};
    
    \begin{pgfonlayer}{bg}    
       \draw[line width=0.5mm,<->] (s)--(a);
       \draw[thick,->] (a)--(b);
       \draw[line width=0.5mm,<->] (b)--(c);
       \draw[thick,<-] (c)--(d);
       \draw[line width=0.5mm,<->] (d)--(e);
       \draw[thick,<-] (e)--(f);
       \draw[line width=0.5mm,<->] (f)--(l);

       \draw[line width=0.5mm,<->] (s)--(ap);
       \draw[thick,->] (ap)--(bp);
       \draw[line width=0.5mm,<->] (bp)--(cp);
       \draw[thick,<-] (cp)--(dp);
       \draw[line width=0.5mm,<->] (dp)--(ep);
       \draw[thick,<-] (ep)--(fp);
       \draw[line width=0.5mm,<->] (fp)--(l);
    \end{pgfonlayer}

    \node[xshift=7cm,yshift=-3cm,vertex] (s) at (0,0) {};
    \node[xshift=7cm,yshift=-3cm, vertex] (a) at (2,2) {};
    \node[xshift=7cm,yshift=-3cm, vertex] (b) at (4,2) {};
    \node[xshift=7cm,yshift=-3cm, vertex] (c) at (6,2) {};
    \node[xshift=7cm,yshift=-3cm, vertex] (d) at (8,2) {};
    \node[xshift=7cm,yshift=-3cm, vertex] (e) at (10,2) {};
    \node[xshift=7cm,yshift=-3cm, vertex] (f) at (12,2) {};
    \node[xshift=7cm,yshift=-3cm, vertex] (ap) at (2,-2) {};
    \node[xshift=7cm,yshift=-3cm, vertex] (bp) at (4,-2) {};
    \node[xshift=7cm,yshift=-3cm, vertex] (cp) at (6,-2) {};
    \node[xshift=7cm,yshift=-3cm, vertex] (dp) at (8,-2) {};
    \node[xshift=7cm,yshift=-3cm, vertex] (ep) at (10,-2) {};
    \node[xshift=7cm,yshift=-3cm, vertex] (fp) at (12,-2) {};
    \node[xshift=7cm,yshift=-3cm, vertex] (l) at (14,0) {};
    \node[xshift=7cm,yshift=-3cm] at (7,0) {$FTF$};
    
    \begin{pgfonlayer}{bg}    
       \draw[line width=0.5mm,<->] (s)--(a);
       \draw[thick,<-] (a)--(b);
       \draw[line width=0.5mm,<->] (b)--(c);
       \draw[thick,->] (c)--(d);
       \draw[line width=0.5mm,<->] (d)--(e);
       \draw[thick,<-] (e)--(f);
       \draw[line width=0.5mm,<->] (f)--(l);

       \draw[line width=0.5mm,<->] (s)--(ap);
       \draw[thick,<-] (ap)--(bp);
       \draw[line width=0.5mm,<->] (bp)--(cp);
       \draw[thick,->] (cp)--(dp);
       \draw[line width=0.5mm,<->] (dp)--(ep);
       \draw[thick,<-] (ep)--(fp);
       \draw[line width=0.5mm,<->] (fp)--(l);
    \end{pgfonlayer}

    \node[yshift=-6cm,vertex] (s) at (0,0) {};
    \node[yshift=-6cm, vertex] (a) at (2,2) {};
    \node[yshift=-6cm, vertex] (b) at (4,2) {};
    \node[yshift=-6cm, vertex] (c) at (6,2) {};
    \node[yshift=-6cm, vertex] (d) at (8,2) {};
    \node[yshift=-6cm, vertex] (e) at (10,2) {};
    \node[yshift=-6cm, vertex] (f) at (12,2) {};
    \node[yshift=-6cm, vertex] (ap) at (2,-2) {};
    \node[yshift=-6cm, vertex] (bp) at (4,-2) {};
    \node[yshift=-6cm, vertex] (cp) at (6,-2) {};
    \node[yshift=-6cm, vertex] (dp) at (8,-2) {};
    \node[yshift=-6cm, vertex] (ep) at (10,-2) {};
    \node[yshift=-6cm, vertex] (fp) at (12,-2) {};
    \node[yshift=-6cm, vertex] (l) at (14,0) {};
    \node[yshift=-6cm] at (7,0) {$FFT$};
    
    \begin{pgfonlayer}{bg}    
       \draw[line width=0.5mm,<->] (s)--(a);
       \draw[thick,<-] (a)--(b);
       \draw[line width=0.5mm,<->] (b)--(c);
       \draw[thick,<-] (c)--(d);
       \draw[line width=0.5mm,<->] (d)--(e);
       \draw[thick,->] (e)--(f);
       \draw[line width=0.5mm,<->] (f)--(l);

       \draw[line width=0.5mm,<->] (s)--(ap);
       \draw[thick,<-] (ap)--(bp);
       \draw[line width=0.5mm,<->] (bp)--(cp);
       \draw[thick,<-] (cp)--(dp);
       \draw[line width=0.5mm,<->] (dp)--(ep);
       \draw[thick,->] (ep)--(fp);
       \draw[line width=0.5mm,<->] (fp)--(l);
    \end{pgfonlayer}

    \node[xshift=7cm,yshift=-6cm,vertex] (s) at (0,0) {};
    \node[xshift=7cm,yshift=-6cm, vertex] (a) at (2,2) {};
    \node[xshift=7cm,yshift=-6cm, vertex] (b) at (4,2) {};
    \node[xshift=7cm,yshift=-6cm, vertex] (c) at (6,2) {};
    \node[xshift=7cm,yshift=-6cm, vertex] (d) at (8,2) {};
    \node[xshift=7cm,yshift=-6cm, vertex] (e) at (10,2) {};
    \node[xshift=7cm,yshift=-6cm, vertex] (f) at (12,2) {};
    \node[xshift=7cm,yshift=-6cm, vertex] (ap) at (2,-2) {};
    \node[xshift=7cm,yshift=-6cm, vertex] (bp) at (4,-2) {};
    \node[xshift=7cm,yshift=-6cm, vertex] (cp) at (6,-2) {};
    \node[xshift=7cm,yshift=-6cm, vertex] (dp) at (8,-2) {};
    \node[xshift=7cm,yshift=-6cm, vertex] (ep) at (10,-2) {};
    \node[xshift=7cm,yshift=-6cm, vertex] (fp) at (12,-2) {};
    \node[xshift=7cm,yshift=-6cm, vertex] (l) at (14,0) {};
    \node[xshift=7cm,yshift=-6cm] at (7,0) {$FTT$};
    
    \begin{pgfonlayer}{bg}    
       \draw[line width=0.5mm,<->] (s)--(a);
       \draw[thick,<-] (a)--(b);
       \draw[line width=0.5mm,<->] (b)--(c);
       \draw[thick,->] (c)--(d);
       \draw[line width=0.5mm,<->] (d)--(e);
       \draw[thick,->] (e)--(f);
       \draw[line width=0.5mm,<->] (f)--(l);

       \draw[line width=0.5mm,<->] (s)--(ap);
       \draw[thick,<-] (ap)--(bp);
       \draw[line width=0.5mm,<->] (bp)--(cp);
       \draw[thick,->] (cp)--(dp);
       \draw[line width=0.5mm,<->] (dp)--(ep);
       \draw[thick,->] (ep)--(fp);
       \draw[line width=0.5mm,<->] (fp)--(l);
    \end{pgfonlayer}

  \end{tikzpicture}